\documentclass[12pt]{article}
\usepackage[utf8]{inputenc}
\usepackage[english]{babel}
\usepackage{amsmath}
\usepackage{mathrsfs}
\usepackage{amssymb}
\usepackage{amsthm}
\usepackage{amsfonts}
\usepackage{xspace}
\usepackage[normalem]{ulem}
\usepackage{graphicx}
\usepackage{multirow}
\usepackage{appendix}
\usepackage{enumerate}
\usepackage{url} 
\usepackage{authblk}
\usepackage{float}
\usepackage{xcolor}
\usepackage{subcaption}
\newcommand\numberthis{\addtocounter{equation}{1}\tag{\theequation}}
\usepackage{multirow}
\usepackage{makecell}
\usepackage[flushleft]{threeparttable}

\usepackage{xcolor}
\usepackage{soul}
\usepackage{tikz}
\usetikzlibrary{calc}

\usepackage[top=1in,bottom=1.5in,left=1in,right=1in]{geometry}

\usepackage{natbib}
\usepackage[unicode=true,bookmarks=true,bookmarksnumbered=false,bookmarksopen=false,breaklinks=false,pdfborder={0 0 1},backref=false,colorlinks=true]{hyperref}
\hypersetup{citecolor=blue}
\usepackage{url}

\newtheorem{theorem}{Theorem}

\newtheorem{remark}{Remark}

\newtheorem{assumption}{Assumption}

\title{Model-robust and efficient covariate adjustment for cluster-randomized experiments}
\author{\small
Bingkai Wang$^{1}$,
Chan Park$^{1}$, Dylan S. Small$^1$, and Fan Li$^2$
\vspace{10pt}

$^1$The Statistics and Data Science Department of the Wharton School, University of Pennsylvania, Philadelphia, PA, USA

$^2$Department of Biostatistics and Center for Methods in Implementation and Prevention Science, Yale University School of Public Health, New Haven, CT, USA}

\begin{document}
\def\spacingset#1{\renewcommand{\baselinestretch}%
{#1}\small\normalsize} \spacingset{1}

\date{\vspace{-5ex}}

\maketitle
\begin{abstract}
Cluster-randomized experiments are increasingly used to evaluate interventions in routine practice conditions, and researchers often adopt model-based methods with covariate adjustment in the statistical analyses. However, the validity of model-based covariate adjustment is unclear when the working models are misspecified, leading to ambiguity of estimands and risk of bias. In this article, we first adapt two conventional model-based methods, generalized estimating equations and linear mixed models, with weighted g-computation to achieve robust inference for cluster-average and individual-average treatment effects. {To further overcome the limitations of model-based covariate adjustment methods}, we propose an efficient estimator for each estimand that allows for flexible covariate adjustment and additionally addresses cluster size variation dependent on treatment assignment and other cluster characteristics. Such cluster size variations often occur post-randomization and, if ignored, can lead to bias of model-based estimators. For our proposed {efficient covariate-adjusted} estimator, we prove that when the nuisance functions are consistently estimated by machine learning algorithms, the estimator is consistent, asymptotically normal, and efficient. When the nuisance functions are estimated via parametric working models, the estimator is triply-robust. Simulation studies and analyses of three real-world cluster-randomized experiments demonstrate that the proposed methods are superior to existing alternatives.
\end{abstract}

\noindent%

\spacingset{1.5}
\setcounter{page}{1}
\section{Introduction }\label{sec:intro}

Cluster-randomized experiments refer to study designs that randomize treatment at the cluster level; clusters can be villages, hospitals, or worksites \citep{murray1998design,donner2000design}. Cluster randomization is often used to study group-level interventions or to prevent treatment contamination under individual randomization, and is increasingly adopted in pragmatic clinical trials evaluating interventions in routine practice conditions. {For example, among $10$ demonstration projects across different disease areas (from 2012 to 2017) supported by the United States \emph{National Institute of Health Pragmatic Clinical Trials Collaboratory}, eight studies adopted cluster randomization (Table 1 in \citealp{weinfurt2017pragmatic}).}


In the analyses of cluster-randomized experiments, covariate adjustment is essential to address baseline chance imbalances and improve the precision of the treatment effect estimators. However, challenges in covariate adjustment persist due to the complexity associated with the multilevel data structure under cluster randomization. First, although model-based methods, including generalized estimating equations (GEE, \citealp{liang1986longitudinal}) and generalized linear mixed models \citep{breslow1993approximate}, are commonly used to perform covariate adjustment in cluster-randomized experiments, it remains unclear whether the treatment effect coefficient corresponds to a clearly-defined estimand of interest, especially when the working model is misspecified. Even in the absence of covariates, \citet{wang2022two} has demonstrated that the treatment coefficient estimator from the GEE with an exchangeable working correlation structure corresponds to an ambiguous estimand when the cluster size is informative for causal effects. Second, participants are frequently sampled after cluster randomization such that the observed cluster size may depend on treatment assignment and other cluster attributes. Failure to address such cluster-dependent sampling schemes, as in standard techniques, can lead to bias of treatment effect estimators. For example, when the observed cluster size depends on cluster-level covariates, \citet{bugni2022inference} demonstrated that the standard difference-in-means estimator is biased for typical estimands of interest. To date, robust causal inference methods that can simultaneously maximize the precision gain from covariate adjustment and address cluster-dependent sampling are unavailable for cluster-randomized experiments.

In this article, {our primary contribution is to provide new covariate-adjusted estimators for cluster-randomized experiments that target clearly-defined estimands with minimal model assumptions, improve precision over standard methods, and address cluster-dependent sampling. We focus on} 
two classes of causal estimands: the cluster-average treatment effect and individual-average treatment effect. The former estimand addresses the question of ``what is the expected change in outcome associated with treatment for a typical cluster with its natural source population?'' and gives equal weight to each cluster, whereas the latter estimand addresses the question of ``what is the expected change in outcome associated with treatment for a typical individual?'' and gives equal weight to each individual. The two estimands represent treatment effects at different levels and differ when there is treatment effect heterogeneity by cluster size \citep{kahan2022estimands}. To estimate both estimands, we first adapt GEE and linear mixed models through weighted g-computation, and provide a set of sufficient conditions to achieve model-robust inference, i.e., consistency and asymptotic normality under arbitrary working model misspecification. These new insights help clarify when conventional multilevel regression models, which are routinely used as standard practice \citep{turner2017review1b}, provide valid average causal effect estimators even if the regression model formulation differs from the unknown and potentially complex data-generating process. {To the best of our knowledge, this entire set of sufficient conditions for typical model-based methods to achieve model-robust inference has not been elucidated in the prior literature, and can inform current practice in analyzing cluster-randomized experiments.}

{The weighted g-computation estimator based on GEE or linear mixed models, however, is subject to two potential limitations. First, it can become biased when the sufficient conditions for model robustness fail to hold, e.g., when the observed cluster size depends on treatment assignment and cluster-level covariates. Second, such an estimator, albeit robust under certain types of model misspecification, is not guaranteed to improve estimation precision through covariate adjustment, and the linear working models may be too restrictive to maximize the precision gain. These two limitations motivate us to develop more principled estimators that can outperform the conventional multilevel regression models. To achieve this goal,} 
we characterize the efficient influence function for both the cluster-average and the individual-average causal effect estimands and propose efficient estimators that allow for flexible covariate adjustment and simultaneously address cluster-dependent sampling. When the nuisance functions are estimated by parametric models, such as GEE or generalized linear mixed models, the proposed estimators are triply-robust, that is, they are consistent to the specified causal estimand if two out of the three nuisance functions are consistently estimated. When the nuisance working models are all consistently estimated, for example, by machine learning algorithms with cross-fitting, the proposed estimators achieve the semiparametric efficiency lower bounds given our causal models. {Compared to model-based covariate adjustment, the efficient estimators notably increase the flexibility in covariate adjustment and substantially enrich the toolbox for analyzing cluster-randomized experiments.}

Our results build on but differ from the existing literature on causal inference for cluster-randomized experiments. 
\cite{imai2009essential} and \cite{middleton2015unbiased} proposed cluster-level methods for estimating the average treatment effects but did not consider covariate adjustment to improve efficiency. \cite{schochet2021design} established the finite-population central limit theorem for linearly-adjusted estimators under blocked cluster randomization, and  \cite{su2021model} extended their results by providing a unified theory for a class of weighted average treatment effect estimands. These prior efforts considered a finite-population framework with linear working models and an independence working correlation structure. In contrast, we study causal effect estimators under a super-population framework and address robust and efficient estimation under a much wider class of working models that are not limited to linear working models. \citet{balzer2021two} and \citet{benitez2021comparative} applied targeted maximum likelihood estimation under hierarchical structural models for cluster-randomized experiments. However, they did not consider cluster size variation arising from cluster-dependent sampling. \cite{bugni2022inference} adapted moment-based estimators to address cluster-dependent sampling but did not consider covariate adjustment. {None of these prior results have provided efficient causal inference under cluster-dependent sampling. We therefore fill in this important gap by proposing efficient estimators to achieve flexible covariate adjustment and accommodate post-randomization cluster-dependent sampling.}

The remainder of the article is organized as follows. In Section~\ref{sec:def}, we formalize our super-population framework, present the causal estimands of interest, and structural assumptions for identification. Section \ref{sec:model-based} adapts GEE and linear mixed models to target our estimands. In Section \ref{sec: efficient-estimation}, we develop the efficient influence functions and propose our efficient estimators. In Sections \ref{sec: simulation} and \ref{sec:data-application}, we demonstrate our theoretical results via simulation studies and re-analyses of three cluster-randomized trials. Section \ref{sec: discussion} concludes with a discussion.


\section{Notation, estimands, and assumptions}\label{sec:def}

We consider a cluster-randomized experiment with $m$ clusters. For each cluster $i$, we let $N_i$ denote the total number of individuals in the underlying source population, $M_i$ be the observed number of individuals sampled into the study, $A_i\in\{0,1\}$ be the cluster-level treatment indicator, and $C_i$ be a $q$-dimensional vector of cluster-level covariates. For each individual $j$ in cluster $i$, we define $Y_{ij}$ as the outcome, $X_{ij}$ as a $p$-dimensional vector of individual-level covariates, and $S_{ij}\in\{0,1\}$ as the sampling indicator, i.e., recruited into the experiment. By definition, the observed cluster size $M_i = \sum_{j=1}^{N_i} S_{ij} \le N_i$. 

We proceed under the potential outcome framework \citep{neyman1990} and define $Y_{ij}(a)$ as the potential outcome and $S_{ij}(a)$ as the potential sampling indicator if cluster $i$ were assigned to treatment group $a \in \{0,1\}$. We assume consistency such that $Y_{ij} = A_iY_{ij}(1) + (1-A_i)Y_{ij}(0)$ and $S_{ij} = A_iS_{ij}(1) + (1-A_i)S_{ij}(0)$. As a result, the potential observed cluster size in a treated cluster, denoted as $M_i(1) = \sum_{j=1}^{N_i} S_{ij}(1)$, can be different from its counterpart in a control cluster, denoted as $M_i(0) = \sum_{j=1}^{N_i} S_{ij}(0)$. Defining $Y_i(a) = \{Y_{i1}(a),\ldots, Y_{iN_i}(a)\} \in \mathbb{R}^{N_i}$ as the collection of potential outcomes, $S_i(a) = \{S_{i1}(a),\ldots, S_{iN_i}(a)\}\in \mathbb{R}^{N_i}$ as the collection of potential sampling indicators, and $\mathrm{X}_i = (X_{i1},\ldots,X_{iN_i})^\top \in \mathbb{R}^{N_i \times p}$, we write the collection of random variables in a cluster as $W_i = \{Y_i(1), Y_i(0), S_i(1), S_i(0), N_i, C_i, \mathrm{X}_i\}$. We introduce the following assumptions on the complete, while not fully observed, data $\{(W_1,A_1),\dots, (W_m, A_m)\}$.

\begin{assumption}[Super-population]\label{asp:super}
(a) Random variables $W_1,\dots, W_m$ are mutually independent. (b) The source population size $N_i$ follows an unknown distribution $\mathcal{P}^{N}$ over a finite support on $\mathbb{N}^+$. (c) Given $N_i$, $W_i$ follows an unknown distribution $\mathcal{P}^{W\mid N}$ with finite second moments. 
\end{assumption}

\begin{assumption}[Cluster randomization]\label{asp:rand}
The treatment indicator $A_i$ for each cluster is an independent realization from a Bernoulli distribution $\mathcal{P}^{A}$ with $pr(A=1)=\pi\in (0,1)$. Furthermore, $(A_1,\dots, A_m)$ is independent of $(W_1,\dots, W_m)$.
\end{assumption}

\begin{assumption}[Cluster-dependent sampling]\label{asp: sampling}
For $a \in \{0,1\}$, the potential observed cluster size $M(a) = h_a(N,C,\epsilon_a)$ for some unknown function $h_a$ and exogenous random noise $\epsilon_a$ that is independent of $\{N,C,\mathrm{X},Y(a)\}$. Furthermore, for each possible $N$-dimensional binary vector ${s}$ with $M(a)$ ones, 
$pr\{S(a) = {s}\mid Y(a),M(a),\mathrm{X},N,C\} = {\binom{N}{M(a)}}^{-1}.$
\end{assumption}

Assumption \ref{asp:super}(a) implies that the data vectors from different clusters are independent, while the outcomes and covariates in the same cluster can be arbitrarily correlated. Assumption \ref{asp:super}(b)-(c) formalize the condition that $W_1, \ldots, W_m$ are marginally identically distributed according to a mixture distribution, $\mathcal{P}^{W\mid N} \times\mathcal{P}^{N}$. This technical condition is useful for handling the varying dimension of $W_i$ across clusters. 
{Assumption \ref{asp:rand} describes the simple cluster randomization design, which we use as a starting place to present our main results. In Sections 3 and 4, we also discuss how our results can be applied to stratified cluster randomization \citep{Zelen1974} and biased-coin cluster randomization \citep{Efron1971}, which are two typical restricted cluster randomization schemes.} 
Given Assumptions \ref{asp:super}-\ref{asp:rand}, the expectation on $(W_i,A_i)$ is taken with respect to $\mathcal{P}^{W\mid N} \times\mathcal{P}^{N}\times\mathcal{P}^{A}$.
Assumption \ref{asp: sampling} implies that the number of sampled individuals can depend on the assignment $A$ and cluster characteristics (source population size $N$ and cluster covariates $C$), but the sampling process is completely random given the number of sampled individuals and the source population size. This assumption relaxes the setting of \cite{bugni2022inference} to allow for arm-specific sampling. Important special cases of Assumptions \ref{asp: sampling} include full enrollment (or focusing on the population of sampled individuals) 
such that $M(a) = N$; random sampling with an arm-specific sample size such that $M(a) = m_a$ for some integer $m_a \le N$; and independent cluster-specific sampling such that each $S_{ij}(a)$ is independently determined by flipping a cluster-specific coin. Assumption \ref{asp: sampling} can be violated when sampling additionally depends on individual-level covariates, which are generally unobserved for nonparticipants. This sampling mechanism leads to post-randomization selection bias \citep{li2022clarifying}, and we leave this form of selection bias for separate work.


We define the class of cluster-average treatment effect ($\Delta_C$) and individual-average treatment effect ($\Delta_I$) estimands as
\begin{align*}
     \Delta_C = f\left\{\mu_C(1), \mu_C(0)\right\}, \quad \Delta_I = f\left\{\mu_I(1), \mu_I(0)\right\},
\end{align*}
where $f$ is a pre-specified smooth function determining the scale of effect measure and, for $a = 0,1$,
\begin{align*}
    \mu_C(a) = E\left\{\frac{\sum_{j=1}^{N_i} Y_{ij}(a)}{N_i} \right\}, \quad \mu_I(a) = \frac{E\left\{\sum_{j=1}^{N_i} Y_{ij}(a)\right\}}{E(N_i)}.
\end{align*}
For example, $f(x,y) = x-y$ leads to the difference estimands, $f(x,y) = x/y$ leads to the relative risk estimands, and $f(x,y)=x(1-y)/\{y(1-x)\}$ leads to the odds ratio estimands. These two classes of estimands $\Delta_C$ and $\Delta_I$ differ based on the corresponding treatment-specific mean potential outcomes, $\mu_C(a)$ versus $\mu_I(a)$.
The former represents the average potential outcome associated with treatment $a$ for a typical cluster along with its natural source population, while the latter represents the average potential outcome associated with treatment $a$ for a typical individual. Depending on the nature of the intervention and the endpoint, either or both estimands may be relevant in a given cluster-randomized experiment. These two estimands can be considered as the super-population analogs to those considered in \citet{su2021model} and \citet{kahan2022estimands} under the finite-population framework and a generalization of the difference estimands in \citet{bugni2022inference}. 
{When the source population size $N_i=n$ is a constant and each element in $\{Y_{ij}, j =1,\dots, n\}$ is assumed to be marginally identically distributed, we have $\mu_C(a) = \mu_I(a) = E[Y_{ij}(a)]$, a special case considered in \cite{wang2021mixed} where the two types of estimands are equal and bear the same interpretation as the typical estimand in an individual-randomized trial.} We refer to \citet{kahan2022estimands} for a more elaborate discussion on differentiating these two estimands and their example applications to cluster-randomized experiments.

Finally, we write the observed data for cluster $i$ as $\mathcal{O}_i = \{Y_i^o, M_i, A_i, N_i, C_i, \mathrm{X}_i^o\}$, where $Y_i^o = \{Y_{ij}: S_{ij}=1, j=1,\dots, N_i\} \in \mathbb{R}^{M_i}$ and $\mathrm{X}_i^o = \{X_{ij}: S_{ij}=1, j=1,\dots, N_i\} \in \mathbb{R}^{M_i \times p }$. The relationships among these random variables are summarized in Figure {S1 in the Supplementary Material}. The central task is to estimate $\Delta_C$ and $\Delta_I$ with $\mathcal{O}_1,\ldots,\mathcal{O}_m$. Of note, we assume that the source population size $N_i$ is available or can be elicited from either historical data or cluster stakeholders; this is practically feasible for cluster-randomized experiments conducted within schools, worksites, or healthcare delivery systems. When the source population size in each cluster is unknown, we will discuss how to estimate $\Delta_C$ in Section \ref{sec: efficient-estimation}, but $\Delta_I$ is generally not identifiable unless we equate the observed cluster size with the source population size by setting $N=M$. 

\section{Model-based covariate adjustment}\label{sec:model-based}
\subsection{Generalized estimating equations (GEE)}\label{subsec:gee}
One popular approach to analyze cluster-randomized experiments is through GEE, often specified through the marginal mean model, $g\{E(Y_{ij}\mid U_{ij})\} = U_{ij}^\top \beta$, where $g$ is the link function,  $U_{ij} = (1, A_i, L_{ij}^\top)^\top$ for user-specified covariates $L_{ij}$ as an arbitrary function of $(N_i, C_i, X_{ij})$, and $\beta = (\beta_0, \beta_A, \beta_L^\top)^\top$. We assume that $g$ is the canonical link, e.g., $g(x) =  x$ for continuous outcomes and $g(x) = \log\{x/(1-x)\}$ for binary outcomes. The parameters $\beta$ are estimated by solving the following estimating equations:
\begin{equation}\label{eq:GEE}
\sum_{i=1}^m \mathrm{D}_i^\top \mathrm{V}_i^{-1} (Y_i^o - \mu_i^o) = 0,
\end{equation}
where $\mu_i^o = \{E(Y_{ij}\mid U_{ij}): S_{ij} = 1\} \in \mathbb{R}^{M_i}$, $\mathrm{D}_i = {\partial \mu_i^o}/{\partial \beta^\top}$, $\mathrm{V}_i = \mathrm{Z}_i^{1/2}\mathrm{R}_i(\rho)\mathrm{Z}_i^{1/2}$ with $\mathrm{R}_i(\rho) \in \mathbb{R}^{M_i\times M_i}$ being a working correlation matrix and $\mathrm{Z}_i = \mathrm{diag}\{v(Y_{ij}|U_{ij}):S_{ij} =1\}\in \mathbb{R}^{M_i\times M_i}$ for some known variance function $v$. We consider two working correlation structures that are commonly used for analyzing cluster-randomized experiments: the independence correlation, $\mathrm{R}_i(\rho) = \mathrm{I}_{M_i}$, and the exchangeable correlation, $\mathrm{R}_i(\rho) = (1-\rho)\mathrm{I}_{M_i}+\rho 1_{M_i}1_{M_i}^\top$, where $\mathrm{I}_{M_i} \in \mathbb{R}^{M_i\times M_i}$ is the identity matrix and $1_{M_i} \in \mathbb{R}^{M_i}$ is a vector of ones. The estimator for $\beta$ is denoted by $\widehat{\beta}$, and the correlation parameter $\rho$ is estimated by a moment estimator $\widehat{\rho}$ described in Example 4 of \cite{liang1986longitudinal}. For analyzing cluster-randomized experiments, a conventional practice is to directly use the coefficient $\widehat{\beta}_A$ along with robust sandwich variance for inference. However, this practice can lead to an ambiguous treatment effect estimand even in the absence of covariate adjustment \citep{wang2022two}.

To estimate our target estimands $\Delta_C$ and $\Delta_I$, we propose weighted g-computation estimators, defined as $\widehat{\Delta}_{C}^{\textrm{GEE-g}} = f\left\{\widehat{\mu}_{C}^{\textrm{GEE-g}}(1),\widehat{\mu}_{C}^{\textrm{GEE-g}}(0)\right\}$ and $\widehat{\Delta}_{I}^{\textrm{GEE-g}} = f\left\{\widehat{\mu}_{I}^{\textrm{GEE-g}}(1), \widehat{\mu}_{I}^{\textrm{GEE-g}}(0)\right\}$, where, for $a = 0,1$,
\begin{align*}
    \widehat{\mu}_{C}^{\textrm{GEE-g}}(a) &= \frac{1}{m}\sum_{i=1}^m \frac{1}{M_i}\sum_{j:S_{ij}=1}g^{-1}\left(\widehat{\beta}_0 + \widehat{\beta}_A a +  \widehat{\beta}_{L}^\top L_{ij}\right), \\
    \widehat{\mu}_{I}^{\textrm{GEE-g}}(a) &= \frac{1}{\sum_{i=1}^m N_i}\sum_{i=1}^m \frac{N_i}{M_i}\sum_{j:S_{ij}=1} g^{-1}\left(\widehat{\beta}_0 + \widehat{\beta}_A a +  \widehat{\beta}_{L}^\top L_{ij}\right).
\end{align*}
{Here, the g-computation step refers to the average of model predictions for each cluster after setting the treatment assignment to be $a$, and has been commonly used as a population standardization technique for estimating marginal estimands in observational studies \citep{rosenbaum1987model}. To target each specific estimand, the weighting for each cluster is taken to be $1/m$ for estimating $\mu_c(a)$ (which gives equal weight to each cluster) and $N_i/\sum_{i=1}^m N_i$ for estimating $\mu_I(a)$ (which gives equal weight to each individual). The variance of $\widehat{\Delta}_{C}^{\textrm{GEE-g}}$ and $\widehat\Delta_I^{\textrm{GEE-g}}$ can be consistently estimated by the robust sandwich variance estimator (after applying the delta method), denoted by $\widehat{V}_C^{\textrm{GEE-g}}$ and $\widehat{V}_I^{\textrm{GEE-g}}$, respectively. Their explicit expressions are developed in the Supplementary Material.}

To proceed, we additionally make the following assumption, which implies that the observed cluster size within each arm is only subject to exogenous randomness.

{\begin{assumption}[Arm-specific random sampling]\label{asp: rand_sampling}
For $a \in \{0,1\}$, the potential observed cluster size $M(a) = h^*_a(\epsilon_a)$ for some unknown function $h^*_a$ and exogenous random noise $\epsilon_a$ that is independent of $\{N,C,\mathrm{X},Y(a)\}$. Furthermore, for each possible $N$-dimensional binary vector ${s}$ with $M(a)$ ones, $pr\{S(a) = {s}\mid Y(a),M(a),\mathrm{X},N,C\} = {\binom{N}{M(a)}}^{-1}.$
\end{assumption}}


{Assumption \ref{asp: rand_sampling} is a special case of Assumption \ref{asp: sampling} by substituting $h^*_a(\epsilon_a)$ for $h_a(N,C,\epsilon_a)$, and is plausible when the number of individuals observed in each cluster can at most depend on the cluster treatment status, regardless of other cluster-level characteristics. As one example, $M(1)=M(0)$ may follow a Poisson distribution with mean of $50$ but truncated within the range of $[5, 200]$, in which case the observed cluster size is completely random; such a within-cluster random sampling assumption has also been previously made for sample size and power calculations in cluster-randomized experiments \citep{shi2018sample}.
As another example, $M(a)\sim\text{Poisson}(50+20a)$ indicates an increased recruitment effort in the treated clusters.
} 
{
Intuitively, Assumption \ref{asp: rand_sampling} is needed because GEE does not model the sampling procedure and may cause bias when sampling is related to covariates as in Assumption \ref{asp: sampling}. Assumption \ref{asp: rand_sampling} resolves this bias by making the sampling procedure ignorable within each treatment group.
The requirement for Assumption \ref{asp: rand_sampling} indicates an inherent limitation of standard GEE in handling within-cluster sampling. 
In Remark~\ref{rmk1}, we provide special cases where Assumption \ref{asp: rand_sampling} can be relaxed to Assumption \ref{asp: sampling} without comprising robustness.
}

Under arm-specific random sampling, Theorem \ref{thm-GEE} below articulates several model specifications of GEE such that $\widehat{\Delta}_{I}^{\textrm{GEE-g}}$ and $\widehat{\Delta}_{C}^{\textrm{GEE-g}}$ are consistent and asymptotically normal, leading to valid statistical inference for our estimands. The regularity conditions required are moment and continuity assumptions  similar to those invoked in Theorem 5.31 of \cite{vaart_1998} for M-estimators. Particularly, model specifications (S2)-(S4) in Theorem \ref{thm-GEE} make no assumption on the underlying distribution for the potential outcomes, and hence allow for model-robust estimation. 


\begin{theorem}\label{thm-GEE}
(a) Under Assumptions \ref{asp:super}, \ref{asp:rand}, \ref{asp: rand_sampling} and regularity conditions in the Supplementary Material, $\left(\widehat{V}_C^{\textrm{GEE-g}}\right)^{-1/2}(\widehat{\Delta}_{C}^{\textrm{GEE-g}} - {\Delta}_{C}) \xrightarrow{d} \mathcal{N}(0, 1)$ if any of the following conditions holds: (S1) the mean model $g\{E(Y_{ij}\mid U_{ij})\} = U_{ij}^\top \beta$ is correctly specified; (S2) an independence working correlation structure is used; (S3) $g$ is the identity link function and the working variance is constant with $v(Y_{ij}\mid U_{ij}) = \sigma^2$; (S4) $L_{ij}$ does not vary within each cluster and thus is only a function of cluster-level covariates $(N_i, C_i)$. (b) If the estimating equations (\ref{eq:GEE}) are additionally weighted by the source population size $N_i$ for each cluster $i$, then $ \left(\widehat{V}_I^{\textrm{GEE-g}}\right)^{-1/2}(\widehat{\Delta}_{I}^{\textrm{GEE-g}} - {\Delta}_{I}) \xrightarrow{d} \mathcal{N}(0, 1)$ if any of (S1)-(S4) holds. 
\end{theorem}

{For model-robust inference via GEE, (S1) is trivial as a correctly specified mean model is expected to yield a valid causal effect estimator. 
Without a correct mean model, (S2) indicates that an independence working correlation structure admits valid causal inference, and contradicts the usual recommendation that the intracluster correlations should be modeled for analyzing cluster-randomized experiments \citep{murray1998design,donner2000design}. 
(S3) is the default specification for analyzing continuous outcomes and ensures model-robustness by proceeding with an ordinary least squares estimator. Finally, (S4) provides a more elaborate mean model specification adjusting for only cluster-level covariates. When none of (S2)-(S4) holds, e.g., logistic GEE with an exchangeable working correlation adjusting for individual-level covariates, the resulting weighted g-computation approach is not guaranteed to be consistent for our causal estimands. 
}

\begin{remark}\label{rmk1}
{In Theorem~\ref{thm-GEE}, Assumption \ref{asp: rand_sampling} can be replaced by Assumption \ref{asp: sampling} under (S1), or under (S2)-(S4) if each cluster is further weighted by $\{1+(M_i-1)\widehat{\rho}\}/M_i$ in solving the estimating equations (1). Under (S2)-(S4), this weighting serves two purposes: the numerator $\{1+(M_i-1)\widehat{\rho}\}$ removes the undesired weighting from the exchangeable working correlation, whereas the denominator offsets the observed cluster size effect. With such weights, the weighted g-computation approach is asymptotically equivalent to that based on an independence working correlation, thereby trading the estimation of intracluster correlation for model robustness.}
\end{remark}

\begin{remark}\label{rmk2}
{
Assumption \ref{asp:rand} does not hold when stratified cluster randomization or biased-coin cluster randomization is used. However, the consistency and asymptotic normality results in Theorem 1 (and also Theorem 2 below) still hold, with the only difference being that the variance estimators $\widehat{V}_C^{\textrm{GEE-g}}$ and $\widehat{V}_I^{\textrm{GEE-g}}$ may be conservative, in the sense that the asymptotic normal distribution has a variance smaller than 1. This result, as well as the variance difference, was provided in Theorem 1 of \cite{wang2021model}. An important special case where the variance estimators are consistent is $\pi = 0.5$, $f(x,y) = x-y$, and the strata variables are adjusted for as covariates. 
}
\end{remark}

\subsection{(Generalized) linear mixed models}\label{subsec: glmm}
Generalized linear mixed models are another popular method for analyzing cluster-randomized experiments. If we write $b_i$ as the random effect for cluster $i$, a typical generalized linear mixed model applied to cluster-randomized experiment often includes the following assumptions: (i) $b_1, \dots, b_m$ are independent, identically distributed from $\mathcal{N}(0,\tau^2)$, where $\tau^2$ is an unknown variance component; (ii) $Y_{i1}, \dots, Y_{iN_i}$ are conditionally independent given $(U_{i1},\dots, U_{iN_i}, b_i)$; and (iii) the distribution $Y_{ij}\mid (U_{i1},\dots, U_{iN_i}, b_i)$ is a member of the exponential family with $E(Y_{ij}\mid U_{i1},\dots, U_{iN_i}, b_i) = E(Y_{ij}\mid U_{ij}, b_i) = g^{-1}(U_{ij}^\top\alpha + b_i)$, where $g$ is the canonical link and $\alpha = (\alpha_0, \alpha_A,  \alpha_L^\top)^\top$.  Given the above assumptions, we estimate model parameters proceeds by maximizing the likelihood function, and a common practice is to consider $\alpha_A$ as the treatment effect parameter. However, the interpretation of $\alpha$ is conditional on random effects, and if the mixed model is misspecified, $\alpha$ lacks a direct connection to our marginal estimands, with an important exception that we detail below.

An interesting special case where the mixed model provides model-robust inference in cluster-randomized experiments is when a linear mixed model is considered as the working model. In this case, 
we can define the  weighted g-computation estimators as $\widehat{\Delta}_{C}^{\textrm{LMM-g}} = f\left\{\widehat{\mu}_{C}^{\textrm{LMM-g}}(1), \widehat{\mu}_{C}^{\textrm{LMM-g}}(0)\right\}$ and $\widehat{\Delta}_{I}^{\textrm{LMM-g}} = f\left\{\widehat{\mu}_{I}^{\textrm{LMM-g}}(1), \widehat{\mu}_{I}^{\textrm{LMM-g}}(0)\right\}$, where
\begin{align*}
    \widehat{\mu}_{C}^{\textrm{LMM-g}}(a) =\widehat{\alpha}_0 + \widehat{\alpha}_A a +    \widehat{\alpha}_{L}^\top \frac{\sum_{i=1}^m  \overline{L}^o_i}{m},\qquad 
    \widehat{\mu}_{I}^{\textrm{LMM-g}}(a) = \widehat{\alpha}_0 + \widehat{\alpha}_A a + \widehat{\alpha}_{L}^\top\frac{\sum_{i=1}^m N_i  \overline{L}^o_i}{\sum_{i=1}^m N_i}
\end{align*}
with $\overline{L}^o_i =  {M_i}^{-1}\sum_{j=1}^{N_i}S_{ij}L_{ij}$ being the average covariate value for cluster $i$ among the observed individuals. 
{These two weighted g-computation estimators are constructed in a similar fashion to those in Section \ref{subsec:gee}.} Furthermore, when the interest lies in the difference estimands with $f(x,y)=x-y$, we simply have $\widehat{\Delta}_C^{\textrm{LMM-g}}=\widehat{\alpha}_A$ and $\widehat{\Delta}_I^{\textrm{LMM-g}}=\widehat{\alpha}_A$ such that the treatment effect estimator from the linear mixed model can be taken as the average causal effect estimator, but these two coefficients are estimated via different weights applied to the log-likelihood as we explain in Theorem \ref{thm: lmm}. 
{For $\widehat{\Delta}_{C}^{\textrm{LMM-g}}$ and $\widehat{\Delta}_{I}^{\textrm{LMM-g}}$, we construct sandwich variance estimators $\widehat{V}_C^{\textrm{LMM-g}}$ and $\widehat{V}_I^{\textrm{LMM-g}}$ with their expressions given in the Supplementary Material.}

Theorem \ref{thm: lmm} shows that $\widehat{\Delta}_{C}^{\textrm{LMM-g}}$ and $\widehat{\Delta}_{I}^{\textrm{LMM-g}}$ are asymptotically valid if the observed cluster size is only subject to exogenous variation within each arm, even when the linear mixed model is arbitrarily misspecified.

\begin{theorem}\label{thm: lmm}
Under Assumptions \ref{asp:super}, \ref{asp:rand}, \ref{asp: rand_sampling} and regularity conditions provided in the Supplementary Material, $\left(\widehat{V}_C^{\textrm{LMM-g}}\right)^{-1/2}(\widehat{\Delta}_{C}^{\textrm{LMM-g}} - {\Delta}_{C}) \xrightarrow{d} \mathcal{N}(0, 1)$. If each cluster is weighted by $N_i$ in the log-likelihood function of the working linear mixed model, then $\left(\widehat{V}_I^{\textrm{GEE-g}}\right)^{-1/2}(\widehat{\Delta}_{I}^{\textrm{LMM-g}} - {\Delta}_{I}) \xrightarrow{d} \mathcal{N}(0, 1)$. 
\end{theorem}

Theorem \ref{thm: lmm} is the counterpart of Theorem \ref{thm-GEE} for linear mixed models in cluster-randomized experiments. {In the special case that $M_i(1)=M_i(0)$, $N_i=n$ is constant, $f(x,y)=x-y$, and $Y_{ij}$ are marginally identically distributed, Theorem \ref{thm: lmm} reduces to Theorem 1(a) of \citet{wang2021mixed}.} 
Compared to GEE with an identity link function and working exchangeable correlation, linear mixed models yield similar estimating equations, but the estimators for nuisance parameters are different, i.e., moment estimation for $\rho$ in GEE versus maximum likelihood estimation for $\sigma^2$, $\tau^2$ in linear mixed models, leading to slightly different asymptotic variances. 
{
For non-continuous outcomes, 
although linear mixed models can still provide valid inference, they are likely misspecified due to the Normal assumptions on random effects and noises. In contrast,  GEE is more flexible regarding the choice of link function and variance function, and hence more natural for handling non-continuous outcomes.}


{Both GEE and linear mixed models provide means to adjust for covariates, which can improve the precision over the unadjusted analysis. However, without further restrictions on the data-generating process, neither method is guaranteed to be equally or more efficient than an unadjusted analysis in general. For instance, \citet{wang2021mixed} has constructed an example (Scenario 1 of Section 4) showing that adjusting for covariates can even decrease the precision of $\widehat{\alpha}_A$ based on linear mixed models. 
Even with an independence correlation structure, \cite{su2021model} showed that GEE with identity link may also lose precision by covariate adjustment. 
These inherent limitations of model-based methods serve as a strong motivation for deriving more principled, and statistically efficient estimators that can maximize the precision gain from covariate adjustment in cluster-randomized experiments.}

Finally, while GEE and generalized linear mixed models are the most widely used approach for analyzing cluster-randomized experiments, there are two alternative estimators, the augmented GEE \citep{stephens2012augmented} and targeted maximum likelihood estimation \citep{balzer2021two,benitez2021comparative} that can also provide robust estimation under Assumptions \ref{asp:super}, \ref{asp:rand}, and \ref{asp: rand_sampling} when certain aspects of working models are misspecified. We provide discussions of those approaches in the Supplementary Material. 


\section{Efficient covariate adjustment}\label{sec: efficient-estimation}


\subsection{Efficient influence functions}
{For model-based covariate adjustment to provide model-robust inference, Theorems \ref{thm-GEE} and \ref{thm: lmm} largely require the observed cluster size to be independent of cluster characteristics $(N,C)$. This is a rather strong assumption and may be violated if the observed cluster size is proportional to the source population size (a common scenario in healthcare research as hospital volume is often associated with patient recruitment results) or if the observed cluster size depends on geographical location or other cluster characteristics. Additionally, the structure of the GEE and the linear mixed model estimators can limit their ability to leverage baseline covariates for maximum precision gain in cluster-randomized experiments, especially when the parametric modeling assumptions are incorrect.} To address such limitations, we develop more principled estimators to simultaneously optimize covariate adjustment and accommodate variable cluster sizes arising from post-randomization cluster-dependent sampling schemes. The proposed estimators are motivated by the efficient influence function, which is a non-parametric functional of observed data that characterizes the target estimand \citep{hines2022demystifying}. With the efficient influence function, one can derive the semiparametric efficiency lower bound for an estimand. That is, the asymptotic variance of all regular and asymptotically linear estimators over the underlying causal model is lower bounded by the variance of the efficient influence function. More importantly, recent advances in causal inference \citep[e.g.][]{van2011targeted, chernozhukov2018double} showed how to use the efficient influence functions to construct an efficient estimator, i.e., achieving the semiparametric efficiency bound, by incorporating machine learning algorithms. Theorem \ref{thm:EIF} provides the efficient influence functions for $\mu_C(a)$ and $\mu_I(a)$, from which the efficient influence functions for $\Delta_C$ and $\Delta_I$ can be obtained by the chain rule.

\begin{theorem}\label{thm:EIF}
(a) Given Assumptions \ref{asp:super}-\ref{asp: sampling}, the efficient influence function for $\mu_C(a)$ is
\begin{align*}
    \text{EIF}_{C}(a) &= \frac{I\{A=a\}}{\pi^a(1-\pi)^{1-a}}\left\{\overline{Y}^o - E\left(\overline{Y}^o\mid A=a, \mathrm{X}^o, M, N, C\right) \right\} \\
    &\quad + \frac{pr(A=a\mid M,N,C)}{\pi^a(1-\pi)^{1-a}}\left\{ E\left(\overline{Y}^o\mid A=a, \mathrm{X}^o, M, N, C\right) - E\left(\overline{Y}^o\mid A=a,  N, C\right)\right\} \\
    &\quad + E\left(\overline{Y}^o\mid A=a, N, C\right) - \mu_C(a),~~~~a\in\{0,1\},
\end{align*}
where $I\{A=a\}$ is an indicator function of $A=a$ and $\overline{Y}^o = M^{-1}\sum_{j=1}^{N} S_{\cdot j} Y_{\cdot j}$ refers to the average outcome among the observed participants in each cluster. {Here, $S_{\cdot j} = AS_{\cdot j}(1) + (1-A)S_{\cdot j}(0)$ and $Y_{\cdot j} = AY_{\cdot j}(1) + (1-A)Y_{\cdot j}(0)$, where $S_{\cdot j}(a)$ and $Y_{\cdot j}(a)$ are the $j$-th element of $S(a), Y(a)$ defined in $\mathcal{P}^W$, respectively.}
(b) Given Assumptions \ref{asp:super}-\ref{asp: sampling}, the efficient influence function for $\mu_I(a)$ is
\begin{align*}
    \text{EIF}_{I}(a) = \frac{N}{E(N)} \left\{\text{EIF}_{C}(a) + \mu_C(a) - \mu_I(a)\right\},~~~~a\in\{0,1\}.
\end{align*}
\end{theorem}

The efficient influence functions in Theorem \ref{thm:EIF} involve three nuisance functions, which we denote by $\eta_a^* = E\left(\overline{Y}^o\mid A=a, \mathrm{X}^o, M, N, C\right)$, $\zeta_a^* =E\left(\overline{Y}^o\mid A=a, N, C\right)$, and $\kappa_a^* =pr(A=a\mid M,N,C)$.
Compared to $\text{EIF}_{C}(a)$, $\text{EIF}_{I}(a)$ additionally includes a weight, $N/E(N)$, to target  the individual-average treatment effect. When $M$ is identical to the source population size $N$, the number of sampled individuals can be treated as a pre-randomization variable, and the efficient influence function for $\mu_C(a)$ reduces to that in \cite{balzer2019new}, where the only nuisance function is $\eta_a^*$. In more general settings as we consider here, two additional nuisance functions, $\zeta_a^*$ and $\kappa_a^*$, are needed to leverage  $\mathrm{X}^o$, which involves post-randomization information $S$, for addressing bias and improving efficiency. Next, we construct new estimators for $\Delta_C$ and $\Delta_I$ based on the efficient influence functions.


\subsection{Efficient estimators based on the efficient influence functions}\label{subsec: sle}
Based on Theorem \ref{thm:EIF}, we propose the following estimator for $\mu_C(a)$ and $\mu_I(a)$:
\begin{align*}
    \widehat{\mu}_{C}^{\textrm{Eff}}(a) = \frac{1}{m} \sum_{i=1}^m  D_i(a,\widehat{h}_a), \qquad \widehat{\mu}_{I}^{\textrm{Eff}}(a) = \frac{1}{\sum_{i=1}^m N_i} \sum_{i=1}^m N_i D_i(a, \widehat{h}_a),\numberthis\label{def:mu_eff}
\end{align*}
where $\widehat{h}_a = (\widehat{\eta}_a, \widehat{\zeta}_a, \widehat{\kappa}_a)$ are user-specified estimators for nuisance functions $h^*_a = (\eta_a^*, \zeta_a^*, \kappa_a^*)$ and 
\begin{align*}
    D_i(a, \widehat{h}_a)&=\frac{I\{A_i=a\}}{\pi^a(1-\pi)^{1-a}}\left\{\overline{Y}_i^o - \widehat{\eta}_a(\mathrm{X}_i^o, M_i, N_i, C_i) \right\} \\
    &\quad + \frac{\widehat{\kappa}_a(M_i,N_i,C_i)}{\pi^a(1-\pi)^{1-a}}\left\{ \widehat{\eta}_a(\mathrm{X}_i^o, M_i, N_i, C_i) - \widehat{\zeta}_a(N_i, C_i)\right\} +\widehat{\zeta}_a(N_i, C_i).
\end{align*}
Then the target estimands defined in Section \ref{sec:def} are estimated by $\widehat\Delta_C^{\textrm{Eff}} = f\left\{ \widehat{\mu}_{C}^{\textrm{Eff}}(1),  \widehat{\mu}_{C}^{\textrm{Eff}}(0)\right\}$ and $\widehat\Delta_I^{\textrm{Eff}} = f\left\{ \widehat{\mu}_{I}^{\textrm{Eff}}(1),  \widehat{\mu}_{I}^{\textrm{Eff}}(0)\right\}$. 
Among the many possibilities for estimating these nuisance functions, we primarily consider the following two approaches. 

The first approach considers parametric working models, that are, $\widehat{\eta}_a = \eta_a(\widehat\theta_{\eta,a}), \widehat{\zeta}_a =  \zeta_a(\widehat\theta_{\zeta,a}), \widehat{\kappa}_a = \kappa_a(\widehat\theta_{\kappa,a})$ for pre-specified functions  $\eta_a, \zeta_a, \kappa_a$ with finite-dimensional parameters $\theta_{\eta,a}, \theta_{\zeta,a}, \theta_{\kappa,a}$.
In this case, we use superscript ``$\textrm{Eff-PM}$'', e.g., $\widehat{\Delta}_C^{\textrm{Eff-PM}}$, to highlight the role of parametric models in estimation.
Example working models include GEE, generalized linear mixed models, penalized regression for variable selection, among others; in all cases, the associated parameters are estimated by solving estimating equations, and we assume that the estimating equations satisfy regularity conditions provided in the Supplementary Material such that the nuisance parameter estimators are  asymptotically linear. 
{With this method, we can compute the sandwich variance estimators $\widehat{V}_C^{\textrm{Eff-PM}}$ for $\widehat{\Delta}_C^{\textrm{Eff-PM}}$ and $\widehat{V}_I^{\textrm{Eff-PM}}$ for $\widehat{\Delta}_I^{\textrm{Eff-PM}}$, which are given in the Supplementary Material.  For the subsequent technical discussions, we denote the probability limit of $(\widehat\theta_{\eta,a}, \widehat\theta_{\zeta,a}, \widehat\theta_{\kappa,a})$ as $(\underline\theta_{\eta,a}, \underline\theta_{\zeta,a}, \underline\theta_{\kappa,a})$}

The second approach uses machine learning algorithms with cross-fitting to estimate all nuisance functions. 
For this case, we use superscript ``$\textrm{Eff-ML}$'', e.g., $\widehat{\Delta}_C^{\textrm{Eff-ML}}$, to indicate using machine learning algorithms.
We assume that each nuisance function estimator is consistent to the truth at an $m^{{1}/{4}}$ rate such that $\widehat{h}_a - h_a^* = o_p(m^{-1/4})$ in $L_2(\mathcal{P})$-norm. 
This $m^{1/4}$ rate can be achieved by many methods such as random forests \citep{Wager2016}, neural networks \citep{Farrell2021}, and boosting \citep{Luo2016}; {a further discussion on this rate is provided in the Supplementary Material}. 
In addition, we assume a regularity condition that $\widehat{\kappa}_a$ and $E\{(\eta_a^*)^2\mid M,N,C\}$ are uniformly bounded, a similar condition invoked in \citet{chernozhukov2018double} for controlling the remainder term and consistently estimating the variance. {In the cross-fitting step, we randomly partition $m$ clusters into $K$ parts with roughly equal sizes (the size difference is at most 1), and denote $\mathcal{O}^*_{k}$ as the $k$-th part and $\mathcal{O}^*_{-k} = \bigcup_{k'\in\{1,\ldots,K\}\setminus\{k\}}\mathcal{O}^*_{k'}$. For each $k$, we compute $\widehat{h}_{a,k}$, which is the nuisance function trained on $\mathcal{O}^*_{-k}$ and evaluated at $\mathcal{O}^*_{k}$, and the estimated nuisance function $\widehat{h}_a$ evaluated on all clusters is then the combination of $\widehat{h}_{a,k}$ for all $k =1, \dots, K$. We then plug in $\widehat{h}_a$ to Equation (\ref{def:mu_eff}) to compute the estimators for $\Delta_C$ and $\Delta_I$.} In practice, we recommend choosing $K$ such that $m/K \ge 10$.  
For variance estimation, we propose the following consistent estimators {based on the efficient influence function and cross-fitting:}
\begin{align*}
    \widehat{V}_{C}^{\textrm{Eff-ML}} &=   \frac{1}{m^2} \sum_{k=1}^{K}  \sum_{i \in \mathcal{O}_k} \left[ \sum_{a=0}^1  \dot{f}_{a,C}\left\{D_i(a, \widehat{h}_{a,k}) - \frac{1}{| \mathcal{O}_k|}\sum_{l \in \mathcal{O}_k} D_l(a, \widehat{h}_{a,k})\right\} \right]^2,\\
    \widehat{V}_{I}^{\textrm{Eff-ML}} &=   \frac{1}{(\sum_{i=1}^m N_i)^2} \sum_{k=1}^{K}  \sum_{i \in \mathcal{O}_k} \left[ \sum_{a=0}^1  \dot{f}_{a,I}\left\{N_iD_i(a, \widehat{h}_{a,k}) - \frac{1}{| \mathcal{O}_k|}\sum_{l \in \mathcal{O}_k} N_lD_l(\widehat{h}_{a,k})\right\} \right]^2,
\end{align*}
where $\dot{f}_{a,C}$ is the partial derivative of $f$ on $\mu_C(a)$ at \{$\widehat{\mu}_{C}^{\textrm{Eff-ML}}(1),  \widehat{\mu}_{C}^{\textrm{Eff-ML}}(0)\}$,  $\dot{f}_{a,I}$ is the partial derivative of $f$ on $\mu_I(a)$ at \{$\widehat{\mu}_I^{\textrm{Eff-ML}}(1),  \widehat{\mu}_I^{\textrm{Eff-ML}}(0)\}$, and $|\mathcal{O}_k|$ is the size of $\mathcal{O}_k$.

Theorem \ref{thm: Eff} summarizes the asymptotic behaviors of our proposed estimators under both strategies for estimating the nuisance functions.

\begin{theorem}\label{thm: Eff}
Given Assumptions \ref{asp:super}-\ref{asp: sampling} and above regularity conditions, 
\begin{enumerate}[(a)]
    \item if $\kappa_a({\underline\theta}_{\kappa,a}) =\kappa_a^*$ or $E\{\eta_a({\underline\theta}_{\eta,a})\mid M,N,C\} = \zeta_a({\underline\theta}_{\zeta,a})$, then $\left(\widehat{V}_C^{\textrm{Eff-PM}}\right)^{-1/2}(\widehat\Delta_C^{\textrm{Eff-PM}} -\Delta_C) \xrightarrow{d} \mathcal{N}(0, 1)$ and $\left(\widehat{V}_I^{\textrm{Eff-PM}}\right)^{-1/2}(\widehat\Delta_I^{\textrm{Eff-PM}} -\Delta_I) \xrightarrow{d} \mathcal{N}(0, 1)$;
    \item if $\widehat{h}_a - h_a^* = o_p(m^{-1/4})$ in $L_2(\mathcal{P})$-norm, then $\left(\widehat{V}_C^{\textrm{Eff-ML}}\right)^{-1/2}(\widehat\Delta_C^{\textrm{Eff-ML}} -\Delta_C) \xrightarrow{d} \mathcal{N}(0, 1)$ and $\left(\widehat{V}_I^{\textrm{Eff-ML}}\right)^{-1/2}(\widehat\Delta_I^{\textrm{Eff-ML}} -\Delta_I) \xrightarrow{d} \mathcal{N}(0, 1)$. Furthermore, $\widehat{V}_C^{\textrm{Eff-ML}}$ and $\widehat{V}_I^{\textrm{Eff-ML}}$ converge in probability to the semiparametric efficiency lower bounds of $\Delta_C$ and $\Delta_I$, respectively.
\end{enumerate}
\end{theorem}

For the proposed estimators with parametric working models, Theorem \ref{thm: Eff}(a) implies that the estimator is consistent if $\kappa_a$ is correctly specified, or the working models $\eta_a$ and $\zeta_a$ are compatible conditioning on $M,N,C$. In the special case that all cluster-level covariates are discrete, the latter condition can be satisfied by setting $$\eta_a(\underline\theta_{\eta,a}) = \zeta_a + \underline\beta^\top\left(\overline{X}^o - \sum_{n}\sum_c I\{N=n, C=c\} \underline\theta_{n,c}\right),$$ 
where $\underline\beta$ is an arbitrary $p$-dimensional vector and $\underline\theta_{n,c} = E\{\overline{X}^o\mid N=n,C=c\}$; 
in this special case, the resulting estimators are in fact robust to arbitrary working model misspecification. In more general cases with non-discrete cluster-level covariates, the proposed estimators are at least triply-robust. That is, they are consistent to their respective target estimands as long as two out of the three nuisance functions are correctly modeled, since $E(\eta^*_a\mid M,N,C) = \zeta^*_a$ as proved in the Supplementary Material.

For the proposed estimators using machine learning algorithms, efficiency can be achieved when all nuisance parameters are consistently estimated, leading to higher asymptotic precision than potentially misspecified parametric working models. For modeling $\eta_a^*$, since $\mathrm{X}^o$ is a matrix and its dimension may change across clusters, a feasible practical strategy is to fit $Y_{ij}$ on $(X_{ij}, M_i, N_i, C_i)$ and pre-specified summary statistics of $\mathrm{X}_i^o$ with fixed dimensions, e.g., $\overline{X}_i^o$, and then compute the cluster-average of predictions as the model fit.  Alternatively, one can directly model $\overline{Y}_i^o$ on $(M_i, N_i, C_i)$ and functions of $\mathrm{X}_i^o$, which can be potentially high-dimensional due to higher-order association between $\mathrm{X}_i^o$ and $\overline{Y}_i^o$.

\begin{remark}\label{rmk:3}
In practice, each of the two approaches for nuisance function estimation has its pros and cons. The machine learning methods are asymptotically more precise, but the efficiency gain typically requires at least a moderate number of clusters. In addition, although the cross-fitting procedure yields the desired convergence property, it may be prone to finite-sample bias especially when the sample size is limited \citep{hines2022demystifying}. This finite-sample bias can be potentially alleviated by using parsimonious parametric working models, thereby trading precision for better finite-sample performance characteristics. 

\end{remark}

\begin{remark}\label{rmk:4}
If the observed cluster size $M<N$, our proposed estimators require an accurate estimate of the source population size for efficient estimation. 
If  $N$ is not fully available for all clusters, one can still use the observed data $(\overline{Y}_i^o, A_i, C_i)$ to infer $\Delta_C$, and the efficient influence function for $\mu_C(a)$ becomes
$$\frac{I\{A=a\}}{\pi^a(1-\pi)^{1-a}}(\overline{Y}^o - \zeta_a^*)+ \zeta_a^* - \mu_C(a),$$ 
based on which an estimator for $\Delta_C$ can be constructed. 
For example, we can apply our proposed estimator with $\widehat{\eta}_a$ set to be equal to $\widehat{\zeta}_a$  and get a consistent estimator for $\Delta_C$ even when $\widehat{\zeta}_a$ is incorrectly specified. 
The individual-average treatment effect $\Delta_I$, however, is generally not identifiable without observing $N$.
\end{remark}

\begin{remark}\label{rmk:5}
{Under stratified cluster randomization or biased-coin cluster randomization, Theorem \ref{thm: Eff}(a) still holds except that the variance estimators may be conservative, following the same argument as in Remark \ref{rmk2}. In addition, Theorem \ref{thm: Eff}(b) remains unchanged under these two designs with the same assumption on the nuisance function estimators. The sketch proof of this result is provided in the Supplementary Material.}
\end{remark}

\section{Simulation experiments}\label{sec: simulation}
\subsection{Simulation design}
We conducted two simulation experiments to compare different methods for analyzing cluster-randomized experiments. The first simulation study focused on estimating the difference estimands of the cluster-average and individual-average treatment effect for continuous outcomes, while the second study focused on the relative risk estimands for binary outcomes. In each experiment, we considered a relatively small ($m=30$) or large ($m=100$) number of clusters, random observed cluster sizes (i.e., $M_i$ is independent of other variables as a special case of Assumption \ref{asp: rand_sampling}) or cluster-dependent observed cluster sizes (i.e., $M_i$ depends on treatment and cluster-level covariates as in Assumption~\ref{asp: sampling}). Combinations of these specifications are labeled as scenarios 1-4 in Table~\ref{tab: sim1} and Table~\ref{tab: sim2}.

In the first simulation experiment, we let $(N_i, C_{i1}, C_{i2})$, $i=1,\dots,m$ be independent draws from distribution $\mathcal{P}^N \times \mathcal{P}^{C_1\mid N} \times \mathcal{P}^{C_2\mid N,C_1}$, where $\mathcal{P}^N$ is uniform over support $\{10,50\}$, $\mathcal{P}^{C_1\mid N} = \mathcal{N}(N/10, 4)$, and $\mathcal{P}^{C_2\mid N,C_1} =\mathcal{B}[\textrm{expit}\{\log(N/10)C_1\}]$ is a Bernoulli distribution with $\textrm{expit}(x) = (1+e^{-x})^{-1}$. Next, for each individual in the source population, we generated the individual-level covariates from $X_{ij1} \sim \mathcal{B}\left({N_i}/{50}\right)$, $X_{ij2} \sim \mathcal{N}\left\{{\sum_{j=1}^{N_i} X_{ij1}}(2C_{i2}-1)/{N_i}, 9\right\}$, and potential outcomes from $  Y_{ij}(1) \sim \mathcal{N}\left\{ N_i/5 + N_i\sin(C_{i1})  (2C_{i2}-1)/30+ 5e^{X_{ij1}}\mid X_{ij2}\mid , 1\right\}$ and $  Y_{ij}(0) \sim \mathcal{N}\left\{\gamma_i + N_i\sin(C_{i1})  (2C_{i2}-1)/30+ 5e^{X_{ij1}}\mid X_{ij2}\mid , 1\right\}$
where $\gamma_i \sim \mathcal{N}(0,1)$ is a cluster-level random intercept to allow for a positive residual intracluster correlation. We set $M_i(1) = M_i(0) = 9 + \mathcal{B}(0.5)$ for the random observed cluster size scenario, and $M_i(1) =  N_i/5+5C_{i2}$, $M_i(0) = 3 I\{N_i=50\} + 3$ for the cluster-dependent observed cluster size scenario. Then, we independently sampled $A_i\sim\mathcal{B}(0.5)$ and defined $Y_{ij} = A_i Y_{ij}(1) + (1-A_i) Y_{ij}(0)$ and $M_i = A_i M_i(1) + (1-A_i) M_i(0)$. Finally, for each cluster, we uniformly sampled without replacement $M_i$ individuals, for whom $S_{ij}=1$, and the observed data in each simulation replicate are $\{Y_{ij}, A_i, M_i, C_{i1}, C_{i2}, X_{ij1}, X_{ij2}: S_{ij} = 1,i=1,\dots,m,j=1,\dots, N_i\}$. For the second simulation study, the data were generated following the first simulation study, except that the potential outcomes were drawn from the following Bernoulli distributions:
$Y_{ij}(1) \sim \mathcal{B}\left[\textrm{expit}\left\{ -N_i/20 + N_i\sin(C_{i1}) (2C_{i2}-1)/30+ 1.5e^{X_{ij1}}\mid X_{ij2}\mid ^{1/2}\right\}\right]$ and $Y_{ij}(0) \sim \mathcal{B}\left[\textrm{expit}\left\{ \gamma_i+N_i\sin(C_{i1}) (2C_{i2}-1)/30+ 1.5(2X_{ij1}-1)\mid X_{ij2}\mid ^{1/2} \right\}\right]$.

We compared the following methods. The unadjusted method \citep{bugni2022inference} is equivalent to our proposed method setting $\widehat{\eta}_a = \widehat{\zeta}_a$ to be a constant. GEE with weighted g-computation was implemented as described in Section \ref{subsec:gee} with an exchangeable working correlation for continuous outcomes and an independence working correlation for binary outcomes. Linear mixed models with weighted g-computation were implimented as described in Section~\ref{subsec: glmm} for both continuous and binary outcomes. {For each model-based method, all baseline covariates are adjusted for as linear terms, which misspecify the true data-generating distribution. However, the GEE estimator satisfies conditions (S2) for the first simulation and (S3) for the second simulation in Theorem \ref{thm-GEE}, yielding valid asymptotic inference under random observed cluster sizes. Likewise, the working linear mixed model is also misspecified, but the weighted g-computation estimator is consistent under random observed cluster sizes as implied by Theorem \ref{thm: lmm}. For our proposed efficient estimators, we considered both parametric working models and machine learning algorithms to estimate the nuisance functions. The former used generalized linear models for estimating $\eta_a$ and $\zeta_a$, and we set a correct working model for $\kappa_a$ such that the conditions in Theorem \ref{thm: Eff}(a) are satisfied. The latter exploited the Super Learner \citep{van2007super} for model-fitting with generalized linear models, regression trees, and neural networks (to obtain consistent estimators for each nuisance function), and facilitates the validation of the results in Theorem \ref{thm: Eff}(b). To summarize, GEE, linear mixed models, and $\widehat{\eta}_a$ adjusted for covariates  $(N_i,C_{i1},C_{i2},X_{ij1}, X_{ij2})$, whereas $\widehat{\zeta}_a$ adjusted for cluster-level covariates $(N_i,C_{i1},C_{i2})$.}



{For each approach, we used the proposed variance estimator and applied a degrees-of-freedom adjustment, i.e., multiplying the variance estimator by a factor $m/(m-5)$, where $5$ is the number of adjusted baseline covariates. This adjustment was motivated by a common technique used in the regression context for small-sample bias correction \citep{mackinnon1985some}, and we adapted it to potentially improve the finite-sample coverage probability of each variance estimator. In addition, we considered the $t$-distribution with $m-5$ degrees of freedom, instead of a Normal distribution, to better approximate the distribution of the standardized covariate-adjusted estimator in finite samples. This choice improved the coverage probability (closer to nominal level) by $1-3\%$ for $m=30$ and less than $1\%$ for $m=100$.} For each scenario of both simulation experiments, we randomly generated $10,000$ data sets and tested the above methods on each data set. {We focused on the following metrics for comparison: bias, empirical standard error (ESE) from the Monte Carlo replications, average of standard error estimates (ASE), and empirical coverage probability (CP) of the 95\% confidence interval.}

\subsection{Simulation results}
Table \ref{tab: sim1} summarizes the simulation results for continuous outcomes. Since the outcome distribution varies by the source population size, the cluster-average treatment effect $\Delta_C = 6$, which differs from the individual-average treatment effect $\Delta_I = 8.67$. Across all scenarios, the proposed methods with no covariate adjustment, parametric working models, or machine learning algorithms have negligible bias and nominal coverage, while the model-based methods, i.e., GEE and linear mixed models, show bias and under-coverage if the observed cluster sizes are cluster-dependent.
\begin{table}[htbp]
\renewcommand{\arraystretch}{0.8}
\centering
\caption{Results in the first simulation experiment with continuous outcomes. }\label{tab: sim1}
\resizebox{1\textwidth}{!}{
\begin{tabular}{lrrrrrrrrrr}
  \hline
  & & \multicolumn{4}{c}{\shortstack[c]{ Cluster-average treatment \\effect $\Delta_C = 6$}
  } & & \multicolumn{4}{c}{\shortstack[c]{ Individual-average treatment \\effect $\Delta_I = 8.67$}
  }\\
Setting  & Method & Bias & ESE   &  ASE &  CP &\ & Bias & ESE   &  ASE &  CP\\ 
  \hline
\multirow{5}{*}{\shortstack[l]{Scenario 1: \\ Small $m$ with\\ 
 random \\
 observed cluster sizes}}&  Unadjusted &   $0.07 $ & $ 4.95 $ & $ 4.88 $ & $ 0.95$ & $ $ & $ -0.09 $ & $ 4.63 $ & $ 4.29 $ & $ 0.93 $\\ 
 &     GEE-g  & $ 0.15 $ & $ 2.78 $ & $ 2.47 $ & $ 0.92$ & $ $ &  $-0.07 $ & $ 3.69 $ & $ 3.13 $ & $ 0.90 $\\ 
 &     LMM-g  & $ 0.15 $ & $ 2.78 $ & $ 2.79 $ & $ 0.96$ & $ $ &  $-0.07 $ & $ 3.68 $ & $ 2.95 $ & $ 0.90 $\\ 
  &  Eff-PM  &  $-0.01 $ & $ 3.20 $ & $ 2.69 $ & $ 0.93 $ & $ $ &  $-0.06 $ & $ 3.96$ & $ 3.40 $ & $ 0.91 $\\ 
   &    Eff-ML  & $0.03 $ & $ 2.16 $ & $ 2.04 $ & $ 0.95 $ & $ $ & $ 0.01 $ & $ 2.63 $ & $ 2.48 $ & $ 0.95 $\\ 
   \hline
\multirow{5}{*}{\shortstack[l]{Scenario 2: \\ Small $m$ with\\ 
 cluster-dependent \\
 observed cluster sizes}} &   Unadjusted  & $ -0.05 $ & $ 5.32 $ & $ 5.32 $ & $ 0.95 $ & $ $ & $-0.12 $ & $ 4.79 $ & $ 4.48 $ & $ 0.93 $\\  
 &     GEE-g  & $ 1.90 $ & $ 3.74 $ & $ 3.24 $ & $ 0.87 $ & $ $ & $ 0.73   $ & $ 4.41 $ & $ 3.66 $ & $ 0.89 $\\ 
 &    LMM-g   & $  1.66 $ & $ 3.60 $ & $ 3.68 $ & $ 0.94 $ & $ $ & $ 0.66 $ & $ 4.28 $ & $ 3.81 $ & $ 0.93$ \\ 
 &  Eff-PM  & $ -0.01 $ & $ 3.82 $ & $ 3.47 $ & $ 0.93 $ & $ $ & $ -0.05 $ & $ 4.30 $ & $ 3.96 $ & $ 0.93$ \\ 
  &    Eff-ML  & $ -0.21 $ & $ 3.40 $ & $ 3.46 $ & $ 0.96 $ & $ $ & $ -0.17 $ & $ 4.17 $ & $ 3.83 $ & $ 0.93$ \\
  \hline
\multirow{5}{*}{\shortstack[l]{Scenario 3: \\ Large $m$ with\\ 
 random \\
 observed cluster sizes}}&  Unadjusted &   $0.01 $ & $ 2.67 $ & $ 2.65 $ & $ 0.95 $ & $ $ & $-0.03 $ & $ 2.41$ & $ 2.37 $ & $ 0.94 $\\ 
 &     GEE-g   & $0.05 $ & $ 1.42 $ & $ 1.38 $ & $ 0.95 $ & $ $ & $-0.01 $ & $ 1.91 $ & $ 1.83 $ & $ 0.94 $\\ 
  &   LMM-g   & $ 0.05 $ & $ 1.42 $  & $  1.42 $  & $  0.95 $ & $ $ & $-0.01 $ & $ 1.91 $ & $ 1.55 $ & $ 0.90$ \\ 
 & Eff-PM  & $ 0.01 $ & $ 1.42 $ & $ 1.38 $ & $ 0.95 $ & $ $ & $-0.01 $ & $ 1.93 $ & $ 1.92 $ & $ 0.95 $ \\ 
  &    Eff-ML  & $ 0.04 $ & $ 0.70 $ & $ 0.71 $ & $ 0.95 $ & $ $ & $ 0.01 $ & $ 0.78 $ & $ 0.81 $ & $ 0.96$ \\ 
  \hline
\multirow{5}{*}{\shortstack[l]{Scenario 4: \\ Large $m$ with\\ 
 cluster-dependent \\
 observed cluster sizes}} & Unadjusted &  $0.05 $& $2.91 $ & $ 2.89 $ & $ 0.95 $ & $ $ & $0.01$ & $ 2.53 $ & $ 2.48 $ & $ 0.94 $\\ 
 &     GEE-g   & $  1.80 $ & $ 1.89 $ & $ 1.82 $ & $ 0.82 $ & $ $ & $ 0.74 $ & $ 2.21 $ & $ 2.12 $ & $ 0.92 $\\ 
&     LMM-g  & $   1.72 $ & $ 1.87 $ & $ 1.89 $ & $ 0.86 $ & $ $ & $ 0.73 $ & $ 2.20 $ & $ 2.00 $ & $ 0.91 $\\ 
  & Eff-PM  & $ -0.01 $ & $ 1.88 $ & $ 1.83 $ & $ 0.95 $ & $ $ & $ -0.01 $ & $ 2.20 $ & $ 2.16 $ & $ 0.94 $\\
  &    Eff-ML & $ -0.01 $ & $ 1.89 $ & $ 1.83 $ & $ 0.94 $ & $ $ & $ -0.01 $ & $ 2.20 $ & $ 2.13 $ & $ 0.94$ \\ 
  \hline
\end{tabular}
}
\vspace{5pt}

{
\raggedright \small
\setlength{\baselineskip}{1pt}
 Unadjusted: the unadjusted estimator. GEE-g: GEE with weighted g-computation. LMM-g: linear mixed models with weighted g-computation. Eff-PM: our proposed method with parametric working models. Eff-ML: our proposed method with machine learning algorithms. ESE: empirical standard error. ASE: average of estimated standard error. CP: coverage probability based on $t$-distribution.

}
\end{table}

For scenarios 1 and 3, since the observed cluster size is completely random, the model-based estimators  perform well, which confirms the theoretical results in Section \ref{sec:model-based}. Among all methods, our method with machine learning algorithms has the highest precision: its variance is $40\sim 93\%$ and $49\sim 90\%$ smaller than the other methods for estimating the $\Delta_C$ and $\Delta_I$, respectively, demonstrating its potential to flexibly leverage baseline covariates for improving study power. Our estimator that uses parametric models for covariate adjustment has comparable precision to model-based methods and is more efficient than the unadjusted estimator. 

For scenarios 2 and 4, more individuals are enrolled in treated clusters with a larger source population, leading to bias for methods that utilized individual-level data without adjusting for this cluster-dependent sampling scheme. Specifically, model-based methods have bias ranging from $0.66$ to $1.90$, and $1\%$ to $13\%$ under coverage. In contrast, our proposed methods show both validity and precision. The validity is reflected by the negligible bias and nominal coverage, and the precision is borne out by their smaller empirical variance than the unadjusted estimator.

Table~\ref{tab: sim2} summarizes the simulation results for binary outcomes, and the patterns are generally similar to those for continuous outcomes. In particular, the proposed methods remain valid across scenarios, but the advantage of machine-learning algorithms over parametric modeling is less obvious. Finally, comparing across sample size configurations, all methods have less stable performance with a smaller number of clusters. Specifically, when $m=30$, methods with covariate adjustment tend to underestimate the true standard error, causing $0\sim 5\%$ under coverage. When $m$ increases to $100$, the estimated standard errors match the empirical standard error, thereby implying the validity of our variance estimator.

\begin{table}[htbp]
\renewcommand{\arraystretch}{0.8}
\centering
\caption{Results in the second simulation experiment with binary outcomes. }\label{tab: sim2}
\resizebox{1\textwidth}{!}{
\begin{tabular}{lrrrrrrrrrr}
  \hline
  & & \multicolumn{4}{c}{\shortstack[c]{ Cluster-average treatment \\ effect $\Delta_C = 1.54$}
  } & & \multicolumn{4}{c}{\shortstack[c]{ Individual-average treatment \\effect $\Delta_I = 1.18$}
  }\\
Setting  & Method & Bias & ESE   &  ASE &  CP &\ & Bias & ESE   &  ASE &  CP\\ 
  \hline
\multirow{5}{*}{\shortstack[l]{Scenario 1: \\ Small $m$ with\\ 
 random \\
 observed cluster sizes}}&  Unadjusted &   $0.03 $ & $ 0.25 $ & $ 0.24 $ & $ 0.94$ & $ $ & $ 0.03 $ & $ 0.15 $ & $ 0.13 $ & $ 0.93 $\\ 
 &     GEE-g  & $ 0.01 $ & $ 0.23 $ & $ 0.19 $ & $ 0.92 $ & $ $ & $ 0.02 $ & $ 0.15 $ & $ 0.13 $ & $ 0.93 $\\ 
 &     LMM-g  & $ 0.01 $ & $ 0.22 $ & $ 0.19 $ & $ 0.92$ & $ $ & $ 0.02 $ & $ 0.15 $ & $ 0.13 $ & $ 0.93 $\\ 
 &  Eff-PM & $ 0.01 $ & $ 0.21 $ & $ 0.20 $ & $ 0.93$ & $ $ & $ 0.02 $ & $ 0.16 $ & $ 0.15 $ & $ 0.95 $\\ 
  &    Eff-ML  & $0.03 $ & $ 0.22 $ & $ 0.20 $ & $ 0.94$ & $ $ & $ 0.03 $ & $ 0.15 $ & $ 0.14 $ & $ 0.95 $\\ 
  \hline
\multirow{5}{*}{\shortstack[l]{Scenario 2: \\ Small $m$ with\\ 
 cluster-dependent \\
 observed cluster sizes}}&  Unadjusted & $0.05 $ & $ 0.30 $ & $ 0.29 $ & $ 0.94 $ & $ $ & $ 0.03 $ & $ 0.17 $ & $ 0.14 $ & $ 0.93 $\\
 &     GEE-g   & $ -0.18 $ & $ 0.30 $ & $ 0.20 $ & $ 0.69$ & $ $ &  $-0.05 $ & $ 0.16 $ & $ 0.17 $ & $ 0.89 $\\
 &     LMM-g   & $  -0.12 $ & $ 0.24 $ & $ 0.20 $ & $ 0.78$ & $ $ & $ -0.04 $ & $ 0.15 $ & $ 0.16 $ & $ 0.90$ \\ 
 &  Eff-PM  & $ 0.04 $ & $ 0.28 $ & $ 0.26 $ & $ 0.92 $ & $ $ & $0.03 $ & $ 0.17 $ & $ 0.17 $ & $ 0.95 $\\ 
  &   Eff-ML  & $0.06 $ & $ 0.30 $ & $ 0.30 $ & $ 0.93 $ & $ $ & $ 0.05 $ & $ 0.19 $ & $ 0.17 $ & $ 0.95$ \\ 
  \hline
\multirow{5}{*}{\shortstack[l]{Scenario 3: \\ Large $m$ with\\ 
 random \\
 observed cluster sizes}}&  Unadjusted &   $0.01 $ & $ 0.13 $ & $ 0.13 $ & $ 0.95$ & $ $ & $ 0.01 $ & $ 0.07 $ & $ 0.07 $ & $ 0.95 $\\ 
 &     GEE-g   & $0.00 $ & $ 0.11 $ & $ 0.10 $ & $ 0.95$ & $ $ & $ 0.00 $ & $ 0.07 $ & $ 0.07 $ & $ 0.95$ \\
 &    LMM-g   & $ 0.00 $ & $ 0.11 $ & $ 0.10 $ & $ 0.94$ & $ $ & $ 0.01 $ & $ 0.07 $ & $ 0.07 $ & $ 0.95$ \\
  &  Eff-PM  & $0.01 $ & $ 0.10 $ & $ 0.10 $ & $ 0.95$ & $ $ & $ 0.00 $ & $ 0.07 $ & $ 0.08 $ & $ 0.97$ \\ 
  &    Eff-ML  & $ 0.01 $ & $ 0.10 $ & $ 0.10 $ & $ 0.95$ & $ $ & $ 0.01 $ & $ 0.06 $ & $ 0.07 $ & $ 0.97 $\\ 
  \hline
\multirow{5}{*}{\shortstack[l]{Scenario 4: \\ Large $m$ with\\ 
 cluster-dependent \\
 observed cluster sizes}}&  Unadjusted &  $0.01 $ & $ 0.15 $ & $ 0.15 $ & $ 0.95$ & $ $ & $0.01 $ & $ 0.08 $ & $ 0.08 $ & $ 0.94$ \\ 
 &     GEE-g   & $  -0.20 $ & $ 0.10 $ & $ 0.10 $ & $ 0.44$ & $ $ & $-0.06 $ & $ 0.07 $ & $ 0.09 $ & $ 0.88 $\\ 
 &    LMM-g   & $   -0.14 $ & $ 0.11 $ & $ 0.10 $ & $ 0.67$ & $ $ & $ -0.05 $ & $ 0.07 $ & $ 0.08 $ & $ 0.89$ \\ 
  &  Eff-PM  & $ 0.01 $ & $ 0.13 $ & $ 0.13 $ & $ 0.95$ & $ $ & $ 0.00 $ & $ 0.07 $ & $ 0.08 $ & $ 0.97 $\\ 
  &    Eff-ML  & $ 0.01 $ & $ 0.14 $ & $ 0.13 $ & $ 0.94$ & $ $ & $ 0.01 $ & $ 0.08 $ & $ 0.08$& $ 0.97$ \\ 
  \hline
\end{tabular}
}
\vspace{5pt}

{
\raggedright \small
\setlength{\baselineskip}{1pt}
 Unadjusted: the unadjusted estimator. GEE-g: GEE with weighted g-computation. LMM-g: linear mixed models with weighted g-computation. Eff-PM: our proposed method with parametric working models. Eff-ML: our proposed method with machine learning algorithms. ESE: empirical standard error. ASE: average of estimated standard error. CP: coverage probability based on $t$-distribution.

}
\end{table}

{To test all methods with a higher degree of source population size heterogeneity, we repeated the first and second simulations with $N$ uniformly distributed over integers in $[10,100]$. Other changes required for this data-generating process are 
provided in the Supplementary Material. Tables S1 and S2 in the Supplementary Material summarize the simulation results, which are similar to the first two simulations, thereby demonstrating the capability of our methods in handling more heterogeneous source population sizes. Of note, an increasing source population size heterogeneity can affect the stability of machine learning algorithms if the number of clusters is small ($m=30$), as reflected by less accurate standard error estimators. However, the bias in the standard error estimator is due to several outlier point estimates and variance estimates in specific simulation iterations, and therefore does not result in under coverage. The bias of the standard error estimator vanished as the number of clusters increases to $100$.}

In the Supplementary Material, we provide additional simulation results for augmented GEE and targeted maximum likelihood estimation under the same settings. 
The target maximum likelihood estimator was unbiased when it only adjusted for cluster-level covariates, but it had bias for $\Delta_I$ if the sampling was cluster-dependent. In most scenarios, both methods were less precise than our proposed method coupled with machine learning estimators for the nuisance functions.

\section{Data applications}\label{sec:data-application}
\subsection{Three cluster-randomized experiments}

The Work, Family, and Health Study
(WFHS) is a cluster-randomized experiment designed to reduce work-family conflict and improve the health and well-being of employees \citep{WFHS}. Fifty-six study groups (clusters) were equally randomized to receive a workplace intervention (treatment) or not (control) with each cluster including 3–50 employees. The observed cluster size has mean $11.77$ and standard error $7.47$. We focused on the control over work hours outcome at the 6-month follow-up, a continuous outcome measure ranging from 1 to 5. We adjusted for the following covariates: cluster sizes and group job functions (core or supporting) at the cluster level, and baseline value of control over work hours and mental health score at the individual level.

The Pain Program for Active Coping and Training study (PPACT), supported by National Institute of Health Pragmatic Clinical Trials Collaboratory, is a pragmatic,  cluster-randomized experiment evaluating the effectiveness of a care–based cognitive behavioral therapy intervention for treating long-term opioid users with chronic pain \citep{debar2022primary}. One-hundred-six clusters of primary care providers (clusters) were equally randomized to receive the intervention or usual care. Each cluster contained 1-10 participants enrolled via phone screening. The observed cluster size has mean $2.03$ and standard error $4.9$. The primary outcome was the PEGS (pain intensity and interference with enjoyment of life, general activity, and sleep) score at 12 months, a continuous scale assessing pain impact as a composite of pain intensity and interference. We adjusted for the cluster sizes and 12 individual-level baseline variables including age, gender, disability, smoking status, body mass index, alcohol abuse, drug abuse, comorbidity, depression, number of pain types, average morphine dose, and heavy opioid usage.

The Improving rational use of artemisinin combination therapies through diagnosis-dependent subsidies (ACTS) study is a cluster-randomized experiment completed in western Kenya \citep{prudhomme2018improving}. Thirty-two community clusters were randomized in a 1:1 ratio to two arms: malaria rapid diagnostic tests with vouchers of artemisinin combination therapies provided for positive test results (treatment) versus standard package (control). The primary outcome was an indicator of receiving a malaria diagnostic test among fevers in the past four weeks at 12 months. Using survey sampling, the observed cluster sizes ranged from 39 to 129, with mean $56.62$ and standard error $16.09$. As the source population size was unknown, we focused on inference of the observed population in ACTS. We adjusted for cluster sizes and existence of health facilities at the cluster level, and gender and wealth index at the individual level.

These three cluster-randomized experiments cover different contexts including social science, chronic pain treatment, and infectious disease control; they also feature different sample size configurations, outcome types, and number of covariates. By re-analyzing these data sets, we aim to illustrate our methods in multiple real-world settings. {In addition, these three cluster-randomized experiments used three different randomization schemes. WFHS had a biased-coin  cluster randomization design \citep{bray2013integrative}, aiming to balance group job functions, number of vice presidents, and location. PPACT adopted simple randomization as we considered in Assumption \ref{asp:rand}. ACTS had a stratified cluster randomization design based on six strata defined by subcounties and the existence of health facilities at the cluster level. Although Assumption \ref{asp:rand} does not hold under the biased-coin or stratified cluster randomization, they do not affect the consistency of our considered estimators as pointed out in Remarks \ref{rmk2} and \ref{rmk:5}. To account for the variance reduction under these two restricted randomization designs, we adjusted for all available covariates balanced by the restricted randomization. In particular, since information on the number of vice presidents and location is not available from the WFHS data set we analyzed, the sandwich variance estimator may be slightly conservative but still valid. For the purpose of demonstrating our theoretical results, we do not further study the variance differences across different randomization designs in the data applications and leave a more systematic study for future research.}


\subsection{Results of data analysis}
For each data set, we estimated the cluster-average and individual-average treatment effects on the difference scale, using the unadjusted estimator, GEE, linear mixed models, and our proposed method with parametric working models and machine learning algorithms. We first set $M=N$ for each study, corresponding to the source population analysis for WFHS and the enrolled population analysis for PPACT and ACTS. For each estimator, we reported the point estimates, 95\% confidence interval based on $t$-distribution, and proportion variance reduction compared to the unadjusted estimator (PVR).

Table~\ref{tab:full-data-analysis} summarizes the results of our data analyses. Across all studies and both estimands, the unadjusted analysis has the widest confidence interval, and covariate adjustment can offer variance reduction as high as 65\%. While the weighted g-computation estimators have higher PVR than our proposed estimators in WFHS and PPACT for estimating the cluster-average treatment effect, they may be biased since Assumption~\ref{asp: rand_sampling} is violated under $M=N$; in contrast, our proposed methods can remove such bias and are hence more reliable. In the analysis of these three cluster-randomized experiments, machine learning algorithms did not show an apparent advantage over parametric working models. 

\begin{table}[ht]
\renewcommand{\arraystretch}{0.8}
\caption{Results of full data analyses.}\label{tab:full-data-analysis}
\centering
\resizebox{1\textwidth}{!}{
\begin{tabular}{crrrrrrrr}
  \hline
  & & \multicolumn{3}{c}{Cluster-average treatment effect
  }  &  & \multicolumn{3}{c}{Individual-average treatment effect
  }\\
 Study &Method & Estimate & 95\% C.I.   & PVR &\ & Estimate & 95\% C.I.   & PVR\\ 
  \hline
\multirow{5}{*}{WFHS} &Unadjusted &   0.17 & (0.00, 0.34) & -  & &0.17 & (0.03, 0.31) & -\\ 
   & GEE-g  & 0.21 & (0.11, 0.31) & 65\% & &0.22 & (0.11, 0.32) & 44\% \\ 
  &  LMM-g  & 0.21 & (0.10, 0.31) & 62\% & & 0.22 & (0.11, 0.32) & 44\% \\ 
   &  Eff-PM &  0.22 & (0.11, 0.33) & 58\% & &0.21 & (0.12, 0.30) & 58\%\\
  & Eff-ML & 0.22 & (0.10, 0.33) & 57\% & & 0.20 & (0.11, 0.30) & 55\% \\ 
   \hline
\multirow{5}{*}{PPACT} &Unadjusted &  -0.83 & (-1.31, -0.35) & -  & & -0.70 & (-1.12, -0.28) & - \\ 
   & GEE-g  & -0.53 & (-0.83, -0.22) & 60\% & & -0.53 & (-0.83, -0.22) & 47\%\\ 
   & LMM-g  &-0.53 & (-0.84, -0.21) & 56\%& & -0.52 & (-0.81, -0.22) & 52\%\\ 
   &  Eff-PM &  -0.59 & (-0.98, -0.21) & 35\% & & -0.51 & (-0.81, -0.21) &49\%\\
  & Eff-ML & -0.62 & (-0.99, -0.24) & 39\%  & & -0.54 & (-0.84, -0.24) & 50\%\\ 
   \hline
\multirow{5}{*}{ACTS} & Unadjusted &   0.08 & (-0.00, 0.16) & -  & &0.08 & (-0.00, 0.16) & -\\ 
  &  GEE-g  & 0.07 & (-0.01, 0.15) & 9\% & &0.07 & (-0.01, 0.15) & 9\%\\ 
  &  LMM-g  &0.07 & (-0.00, 0.15) & 12\% & & 0.07 & (-0.00, 0.15) & 12\%\\ 
  &   Eff-PM &  0.08 & (0.00, 0.15) & 20\% & &0.07 & (-0.00, 0.15) & 18\%\\
  & Eff-ML & 0.08 & (-0.00, 0.15) & 12\%   & & 0.07 & (-0.00, 0.15) & 18\%\\ 
   \hline
\end{tabular}}
\vspace{5pt}

{
\raggedright \small
\setlength{\baselineskip}{1pt}
 Unadjusted: the unadjusted estimator. GEE-g: GEE with weighted g-computation. LMM-g: linear mixed models with weighted g-computation. Eff-PM: our proposed method with parametric working models. Eff-ML: our proposed method with machine learning algorithms. ESE: empirical standard error. 95\% C.I.: 95\% confidence interval based on $t$-distribution. PVR: proportional variance reduction compared to the unadjusted estimator.

}
\end{table}

{To further illustrate our methods under the setting of cluster-dependent sampling, we performed a simulation study in the Supplementary Material based on the WFHS data. This analysis mimics the real-world setting, i.e., the distribution of outcome and covariates under control is based on real data. The simulation results, summarized in Table S5 in the Supplementary Material, showed consistent findings with our theoretical results.}

\section{Concluding remarks} \label{sec: discussion}
{Under the overarching goal to improve the current practice of covariate adjustment in cluster-randomized experiments, our contributions to the literature are two-fold. Above all, we clarified a set of sufficient conditions under which two model-based regression estimators, when combined with weighted g-computation, are robust for estimating the cluster-average treatment effect and individual-average treatment effect, even when their working models are arbitrarily misspecified. Given the frequency of their use in practice, our results serve as important clarifications for existing practice in cluster-randomized experiments and provide simple recipes for robust covariate adjustment through GEE and linear mixed models. Despite the simplicity and accessibility of these model-based estimators, their model-robustness property largely hinges on the arm-specific random sampling assumption. Furthermore, an incorrectly specified working model limits the ability to maximally leverage the precision gain from covariate adjustment. These limitations have motivated us to search for more principled and efficient strategies for covariate adjustment without compromising the model-robustness property for the two classes of estimands. 
Therefore, as a second contribution, we have 
derived the efficient influence functions and proposed efficient estimators for the two classes of estimands that allow for efficient covariate adjustment and additionally accommodate cluster-dependent sampling. The efficient estimators open the door for using a wider class of parametric working models or machine learning algorithms to learn the potentially complex data-generating mechanisms without affecting the validity of causal inference in cluster-randomized experiments.}

{Our asymptotic framework assumes that the source population size of each cluster is bounded. This is a convenient and yet practice assumption that avoids challenges in defining the two classes of causal estimands and in addressing the potentially high dimensionality of $W_i$. Although this assumption may appear strong, we can set the upper bound of the source population size to be large enough to accommodate most real-world settings without affecting our asymptotic theory. 
For example, the upper bound of the source population size may be 100 for cluster-randomized experiments when the randomization units are classrooms, whereas, in clinic settings,  this upper bound may be much larger but still considered to be finite for healthcare interventions. The extension of our current asymptotic theory to allow for potentially infinite source population size along with required restrictions on $\mathcal{P}^{W|N}\times \mathcal{P}^N$ is an area of future research. }

A common objective for covariate adjustment in cluster-randomized experiments is to address chance imbalance and improve precision \citep{su2021model}. However, how to best select the optimal set and functional forms of covariates is an open problem that remains to be addressed in future research. For cluster-randomized experiments, this problem may be more challenging due to the unknown intracluster correlations of the outcomes and covariates within each cluster. In many cases where the investigators can only include a limited number of clusters, there will be a trade-off between the loss of degrees of freedom by adjusting for weakly prognostic covariates and the potential asymptotic power gain by including more baseline variables. While the proposed estimators may be a useful vehicle to incorporate variable selection techniques in the working models, we maintain the recommendation of pre-specifying prognostic covariates for adjustment in the design stage based on subject-matter knowledge for practical applications. 

Throughout the article, we have defined the cluster-average treatment effect and individual-level treatment effect estimands as a function of the source population size $N_i$, which can differ from the observed cluster size $M_i$. Therefore, accurate identification of these estimands require knowledge of the source population size. Conceptually, this source population represents the set of eligible participants in each cluster that may be recruited had the investigator obtained unlimited financial and logistical resources, and is precisely the set of individuals that the intervention is designed to target. The availability of $N_i$ typically depends on the types of clusters and the resource of the study. For example, the source population size can be readily available if schools, worksites and villages are randomized, whereas the source population size may be estimated from historical data if clinics or hospitals are randomized. In this latter case, one may set $N_i=M_i$ to perform an enrolled population analysis, which implicitly assumes equivalence between the observed population and the source population in each cluster. In the case where $N_i>M_i$ but no information of $N_i$ is available, we stated in Remark \ref{rmk:3} that only the cluster-average treatment effect is identifiable. In future work, it would be worthwhile to establish alternative conditions to identify the individual-level treatment effect in the absence of knowledge on source population size.

For developing the asymptotic properties of the model-robust estimators, we have primarily focused on simple randomization {and discussed the implications under stratified cluster randomization and biased-coin cluster randomization. Beyond these randomization schemes, cluster rerandomization is a useful strategy} to address baseline imbalance and further improve the study power. \citet{lu2022design} recently developed the asymptotic theory for cluster rerandomization under a finite-population framework. It would be useful to further extend our results to accommodate cluster rerandomization.

\section*{Acknowledgement}
Research in this article was partially supported by a Patient-Centered Outcomes Research Institute Award\textsuperscript{\textregistered} (PCORI\textsuperscript{\textregistered} Award ME-2020C3-21072) and National Institute of Allergy and Infectious Diseases (NIAID) grants R01AI148127, K99AI173395. The statements presented in this article are solely the responsibility of the authors and do not necessarily represent the official views of PCORI\textsuperscript{\textregistered}, or its Board of Governors or Methodology Committee. 

\section*{Supplementary material}
Supplementary material includes a causal graph summarizing all random variables, regularity conditions for Theorems \ref{thm-GEE}-\ref{thm: Eff}, consistent variance estimators, proofs, additional simulation studies,  and a review of other methods for cluster-randomized experiments.

{
\bibliographystyle{apalike}
\bibliography{references}
}

\end{document}


\def\spacingset#1{\renewcommand{\baselinestretch}%
{#1}\small\normalsize} \spacingset{1}

\date{\vspace{-5ex}}

\maketitle


\spacingset{1.5}
\setcounter{page}{1}

\renewcommand\thesection{\Alph{section}}
In Section~\ref{suppsec: causal-graph}, we provide a causal graph summarizing the relationship among random variables under our causal framework.
In Section~\ref{suppsec: reg-condi}, we provide the regularity conditions needed for our theorems. 
In Section~\ref{sec: sandwich-variance}, we provide consistent variance estimators. 
In Section~\ref{suppsec: proof}, we provide the proofs for our theorems. In Section~\ref{suppsec: other methods}, we review other methods for cluster-randomized experiments mentioned in the main paper.
In Section~\ref{suppsec: additional-sim}, we provide additional simulations for increased source population size heterogeneity, larger sample sizes, and other methods for cluster-randomized experiments.
\section{A causal graph for observed and unobserved random variables}\label{suppsec: causal-graph}
\begin{figure}[H]
		\centering
		\resizebox{0.5\textwidth}{!}{
		\begin{tikzpicture}
			
			\tikzset{line width=1pt,inner sep=5pt,
				ell/.style={draw, inner sep=1.5pt,line width=1pt,minimum size=0.8cm}}

			\node[shape = circle, dashed, ell] (X) at (-2, -2) {\large $\bfX$};
			\node[shape = circle, ell] (C) at (-2, 0) {\large$\bC$};
			\node[shape = circle, ell] (N) at (-2, 2) {\large $N$};	
			\node[shape = circle, ell] (A) at (0, 0) {\large $A$};		
			\node[shape = circle, dashed, ell] (Y) at (2, 0) {\large$\bY$};	
			\node[shape = circle,  ell] (Yo) at (4, 0) {$\bY^o$};	
			\node[shape = circle,  ell] (M) at (2, 2) {\large$M$};	
			\node[shape = circle,  dashed, ell] (S) at (4, 2)
			{\large$\bS$};	
			\node[shape = circle,  ell] (Xo) at (4, -2){\large$\bfX^o$};	
			\draw[-stealth, line width=1pt, bend left](N) to (S);
			\draw[-stealth, line width=1pt, bend left](S) to (Xo);
			\draw[-stealth, line width=1pt, bend right](C) to (Y);
			
			\foreach \from/\to in {  A/Y,Y/Yo, N/M,C/M,A/M, X/Xo, M/S, S/Yo, N/Y, X/Y.south}
			\draw[-stealth, line width = 1pt] (\from) -- (\to);
			
			\draw[line width=1pt] (X) -- (C);
			\draw[line width=1pt] (N) -- (C);
			\draw[line width=1pt, bend right] (N) to (X);
			
		\end{tikzpicture}
		}
\caption{A graph representation of the causal relationship among random variables in cluster-randomized experiments. Observed and unobserved nodes have solid and dashed boundaries, respectively.}
\label{fig:dag-1}
\end{figure}
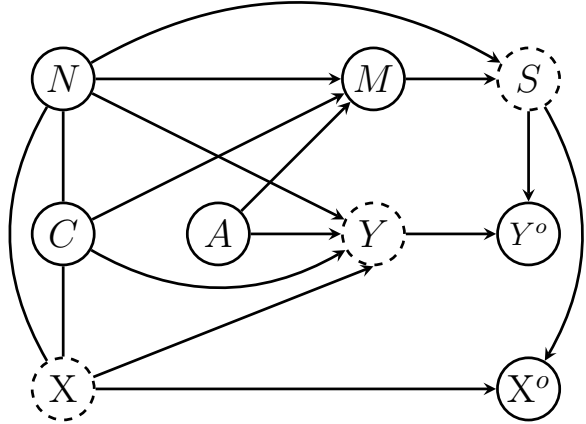 

\section{Regularity conditions for Theorems 1, 2, and 4}\label{suppsec: reg-condi}
Let $\bpsi(\bO, \btheta)$ be a vector of (generic) estimating functions of observed data $\bO$ and parameters $\btheta$. The parameters are estimated by solving $\sum_{i=1}^m \bpsi(\bO_i, \widehat\btheta) = \bzero$. The estimating functions we specified are defined in (\ref{eq: gee-psi-C}) and (\ref{eq: gee-psi-I}) for GEE-g and  (\ref{eq: lmm-psi-C}) and (\ref{eq: lmm-psi-I}) for LMM-g. For our proposed efficient estimator with user-specified parametric working models, we denote $\bpsi(\bO,\btheta)$ as the estimating equations for parameters $\btheta = (\btheta_{\eta,1}, \btheta_{\zeta,1}, \btheta_{\kappa,1}, \btheta_{\eta,0}, \btheta_{\zeta,0}, \btheta_{\kappa,0})$. We make the following regularity assumptions on $\bpsi(\bO, \btheta)$:
\begin{enumerate}
    \item $\btheta \in \mathrm{\Theta}$, a compact subset of the Euclidean space.
    \item The function $\btheta \mapsto \bpsi({o}, \btheta)$ is dominated by a  square-integrable function and twice continuously differentiable for every ${o}$ in the support of $\bO$ with nonsingular derivative matrix $\mathrm{B}_{\btheta}$ that is continuous on $\btheta$. Furthermore, the first and second order derivatives of $\btheta \mapsto \bpsi({o}, \btheta)$ are also dominated by a square-integrable function in a small neighborhood of $\underline{\btheta}$.
    \item There exists a unique solution in the interior of $\mathrm{\Theta}$, denoted as $\underline{\btheta}$, to the equations $E[\bpsi(\bO, \btheta)] =\bzero$.
    \item $\{\bpsi({o}, \btheta): ||\btheta-\underline\btheta|| <\delta\}$ is P-Donsker (Chapter 19 of \citealp{vaart_1998}) for some $\delta >0$. One sufficient condition is that $||\bpsi({o}, \btheta) - \bpsi({o}, \widetilde{\btheta})||_2 \le m({o}) ||\btheta- \widetilde{\btheta}||_2$ for some square-integrable function $m({o})$ and every $\btheta, \widetilde{\btheta}\in \mathrm{\Theta}$.
\end{enumerate}
Of note, regularity condition 3 does not imply a correctly specified model. Instead, it solely requires the uniqueness of maxima in quasi-likelihood estimation. This can be achieved by carefully designing the estimating equations and restricting parameter space $\mathrm\Theta$ to rule out degenerative solutions. For example, if $\psi$ is the estimating equations for a linear regression, then it is equivalent to the invertibility of the covariance matrix of covariates plus finite fourth moments of $O$.

For the proposed estimator, we additionally assume:
\begin{enumerate}
    \item[5.] For parametric working models, the class of functions $\{\eta_a(\btheta_{\eta,a}), \zeta_a(\btheta_{\zeta,a}), \kappa_a(\btheta_{\kappa,a}): \btheta_a \in \mathrm{\Theta}, a \in \{0,1\}\}$ is dominated by a square-integrable function, twice continuously differentialable in $\btheta_a$ with dominated first and second order derivatives,  and P-Donsker when $||\btheta_a -\underline{\btheta}_a||_2 < \delta$ for some $\delta > 0$, where $\btheta_a = (\btheta_{\eta,a}, \btheta_{\zeta,a}, \btheta_{\kappa,a})$ and $\underline{\btheta}_a = (\underline{\btheta}_{\eta,a}, \underline{\btheta}_{\zeta,a}, \underline{\btheta}_{\kappa,a})$.
    \item[6.] For data-adaptive estimation with crossing fitting, we assume that $\widehat{\kappa}_{a,k}$ and \\ $E[ \{\eta_a^*(\bfX^{o}, M,N,C)\}^2 | M,N,C ]$ are uniformly bounded.
\end{enumerate}

For our proposed methods with machine learning, the condition for achieving $\widehat{h}_a-h^*_a = o_p(m^{-1/4})$ varies by the nuisance function estimators. 
    However, a common condition across methods for the $m^{-1/4}$ rate is related to  the smoothness of $h^*_a$ and the number of explanatory variables in $h^*_a$. For example, the kernel regressor \citep{stone1980optimal} requires $h^*_a$ to have $\alpha$-order bounded, continuous derivatives such that $2\alpha > d$, where $d$ is the number of explanatory variables (together with other regularity conditions); while the requirement is $\alpha > d$ for the deep neural network \citep{Farrell2021} along with its own regularity conditions. 
    For random forests and boosting mentioned in the main paper, this requirement has a complex form, and the details can be found in Theorem 4 of \cite{Wager2016} (and \citealp{wager2018estimation}) for random forests and Theorem 2 of \cite{Luo2016} for boosting. 
    In general, most data-adaptive methods have much weaker restrictions for $h^*_a$ compared to parametric models and can hence approximate a much larger class of functions.

\section{Consistent variance estimators}\label{sec: sandwich-variance}

The variance of $\widehat{\Delta}_{C}^{\textrm{GEE-g}}$ can be estimated by the robust sandwich variance estimator, defined as 
$$\widehat{V}_C^{\textrm{GEE-g}} =m^{-2}\nabla \widehat{f}_C^{\textrm{GEE-g}}{}^\top \left(\sum_{i=1}^m \widehat{IF}_{C,i}^{\textrm{GEE-g}} \widehat{IF}_{C,i}^{\textrm{GEE-g}}{}^\top \right)\nabla \widehat{f}_C^{\textrm{GEE-g}},$$ 
where $\nabla \widehat{f}_C^{\textrm{GEE-g}}$ is the gradient of $f$ evaluated at $\left\{\widehat{\mu}_{C}^{\textrm{GEE-g}}(1),\widehat{\mu}_{C}^{\textrm{GEE-g}}(0)\right\}$ and
\begin{equation*}
\resizebox{\textwidth}{!}{
    $\widehat{IF}_{C,i}^{\textrm{GEE-g}} 
   = \left[\begin{array}{c}
    \displaystyle\left\{\frac{1}{m}\sum_{l=1}^m\frac{A_lM_l}{1 + (M_l-1)\widehat{\rho}}\right\}^{-1} \frac{A_i}{1 + (M_i-1)\widehat{\rho}}\{\overline{Y}_i^o-\overline\mu_i(1,\widehat\beta)\}   + \overline\mu_i(1,\widehat\beta) -\widehat\mu_C^{\textrm{GEE-g}}(1)    \\
  \displaystyle\left\{\frac{1}{m}\sum_{l=1}^m\frac{(1-A_l)M_l}{1 + (M_l-1)\widehat{\rho}}\right\}^{-1} \frac{1-A_i}{1 + (M_i-1)\widehat{\rho}}\{\overline{Y}_i^o-\overline{\mu}_i(0,\widehat\beta)\}   + \overline\mu_i(0,\widehat\beta) -\widehat{\mu}_C^{\textrm{GEE-g}}(0)
    \end{array}\right],$}
\end{equation*}
is the estimated influence function for cluster $i$,  $\overline{Y}_i^o = M_i^{-1} S_i^\top Y_i$, and $\overline\mu_i(a,\widehat\beta) = M_i^{-1} S_i^\top\mu_i(a,\widehat{\beta})$ with $\mu_i(a,\widehat{\beta})$ being the estimate of $\{g^{-1}(U_{ij}^\top \beta): j =1,\dots, N_i\}$ given $A=a$ for $a \in \{0,1\}$. 
Similarly, the variance of $\widehat\Delta_I^{\textrm{GEE-g}}$ is estimated by
$$\widehat{V}_I^{\textrm{GEE-g}} = (\sum_{i=1}^m N_i)^{-2}\nabla \widehat{f}_I^{\textrm{GEE-g}}{}^\top \left(\sum_{i=1}^m N_i^2 \widehat{IF}_{I,i}^{\textrm{GEE-g}} \widehat{IF}_{I,i}^{\textrm{GEE-g}}{}^\top \right)\nabla \widehat{f}_I^{\textrm{GEE-g}}$$
where $\widehat{f}_I^{\textrm{GEE-g}}$ and $\widehat{IF}_{I,i}^{\textrm{GEE-g}}$ are defined in  the same way as $\widehat{f}_C^{\textrm{GEE-g}}$ and $\widehat{IF}_{C,i}^{\textrm{GEE-g}}$ except that the subscript is substituted by $I$ in all estimated parameters.

The variance estimate of $\widehat{\Delta}_{C}^{\textrm{LMM-g}}$, denoted as $\widehat{V}_C^{\textrm{LMM-g}}$, is constructed in the same way as $\widehat{V}_C^{\textrm{GEE-g}}$ with $g$ being the identity link and $ \{\widehat\mu_C^{\textrm{GEE-g}}(1), \widehat\mu_C^{\textrm{GEE-g}}(0), \widehat\beta, \widehat{\rho}\}$ substituted by $ \{\widehat\mu_C^{\textrm{LMM-g}}(1), \widehat\mu_C^{\textrm{LMM-g}}(0), \widehat\alpha, \widehat{\tau}^2(\widehat{\tau}^2 + \widehat{\sigma}^2)^{-1}\}$; for $\widehat{\Delta}_{I}^{\textrm{LMM-g}}$, the variance estimate $\widehat{V}_I^{\textrm{LMM-g}}$ can be constructed in a similar way.

For our proposed method with parametric nuisance models, the parameters \\ $\theta = \{\mu_C(a), \theta_{\eta,a}, \theta_{\zeta,a}, \theta_{\kappa,a}: a = 0,1\}$, can be viewed as jointly estimated by solving estimating equations $\sum_{i=1}^m \psi_C(\mathcal{O}_i; \theta) = 0$ for a known function $\psi_C$, and the sandwich variance estimator $\widehat{V}_C^{\textrm{Eff-PM}}$ for $\widehat{\Delta}_C^{\textrm{Eff-PM}}$ is computed as
$$\nabla \widehat{f}_C^\top \left\{\sum_{i=1}^m \frac{\partial}{\partial \theta} \psi_C(\mathcal{O}_i; \theta) |_{\theta = \widehat{\theta}}\right\}^{-1} \left(\sum_{i=1}^m \psi_C(\mathcal{O}_i;  \widehat\theta) \psi_C(\mathcal{O}_i; \widehat\theta){}^\top \right)\left\{\sum_{i=1}^m \frac{\partial}{\partial \theta} \psi_C(\mathcal{O}_i; \theta) |_{\theta = \widehat{\theta}}\right\}^{-1}\nabla \widehat{f}_C,$$
where $\nabla \widehat{f}_C$ is the gradient of $f$ over $\theta$ evaluated at $\widehat{\theta}$. Following a similar procedure we can compute the sandwich variance estimator $\widehat{V}_I^{\textrm{Eff-PM}}$ for $\widehat{\Delta}_I^{\textrm{Eff-PM}}$.

\section{Proofs}\label{suppsec: proof}
\subsection{Lemmas}

\begin{lemma}\label{lemma1}
Given Assumptions 1-3, we have (i) $\bfX \perp (A, M) | (N, \bC)$, (ii) $(A,\bfX, \bY,\bC)\perp \bS | (M, N)$, and (iii) $\bfX\perp \bS \perp A | (M, N, \bC)$.
\end{lemma}
\begin{proof}
(i) Assumption 3 implies that $M(a) \perp \bfX | (N, \bC)$, which further implies $M \perp \bfX | (A, N, \bC)$ since $A \perp (M(a), \bfX, N, \bC)$ and $M = AM(1) + (1-A)M(0)$. Hence
\begin{align*}
    pr(\bfX|A,M, N,\bC) = \frac{pr(\bfX, A,M,N ,\bC)}{pr(A,M, N,\bC)} =  \frac{pr(M|\bfX,A,N,\bC)pr(\bfX,N,\bC)pr(A)}{pr(M|A, N,\bC)pr(A) pr(N,\bC)} = pr(\bfX|N,\bC).
\end{align*}
(ii) Assumptions 2-3 imply that $pr\{\bS(a) = {s}|A, \bY(a), \bfX, M(a), N, \bC\} = \binom{N}{M(a)}^{-1}$ and, hence, $pr\{\bS = {s}|A, \bY, \bfX, M, N, \bC\} =\sum_{a} I\{A=a\}pr\{\bS(a) = {s}|A=a, \bY(a), \bfX, M(a), N, \bC\} =  \binom{N}{M}^{-1}$, thereby indicating $(A,\bfX, \bY,\bC)\perp \bS | (M, N)$.

\noindent (iii) We have
\begin{align*}
    pr(\bfX,\bS,A|M,N,\bC) &= \frac{pr(\bfX,\bS,A,M,N,\bC)}{pr(M,N,\bC)} \\
    &= \frac{pr(\bS|M,N,\bfX,A,\bC)pr(\bfX|A,M,N,\bC)pr(A,M,N,\bC)}{pr(M,N,\bC)} \\
    &= pr(\bS|M,N,\bC) pr(\bfX|M,N,\bC) pr(A|M,N,\bC),
\end{align*}
where the last equation results from (i) and (ii).
\end{proof}

\begin{lemma}\label{lemma2}
Recall the notation $\overline{Y}(a) = \frac{1}{N} \sum_{j=1}^{N}Y_{\cdot j}(a)$, 
$\overline{Y} =\overline{Y}(A)$,
$\overline{Y}^o(a) = \frac{1}{M(a)} \sum_{j=1}^{N}S_{\cdot j}(a)Y_{\cdot j}(a)$, and $\overline{Y}^o = \overline{Y}^o(A)$.
Given Assumptions 1-3, we have 
$$E\left[\overline{Y}(a)\big| N, \bC\right] = E\left[\overline{Y}^o(a)\big| N, \bC\right] = E\left[\overline{Y}^o\big|A=a, N, \bC\right] = E\left[\overline{Y}^o\big|A=a, M, N, \bC\right].$$
\end{lemma}
\begin{proof}
We first prove $E\left[\overline{Y}(a)\big| N, \bC\right] = E\left[\overline{Y}^o(a)\big| N, \bC\right]$. Assumption 3 implies that, for each $j = 1,\ldots, N$,
\begin{align*}
    pr\{S_{\cdot j}(a) = 1|\bY(a), \bfX, M(a), N, \bC\} = \frac{\binom{N-1}{M(a)-1}}{\binom{N}{M(a)}} = \frac{M(a)}{N}.
\end{align*}
Hence, for $a = 0,1$,
\begin{align*}
    E\left[ \overline{Y}^o(a) \bigg| N,\bC\right] &= E\left[ E\left[ \frac{\sum_{j=1}^{N}S_{\cdot j}(a)Y_{\cdot j}(a)}{M(a)} \bigg| \bY(a), \bfX, M(a), N, \bC\right] \bigg| N,\bC\right] \\
    &= E\left[ \frac{\sum_{j=1}^{N}E[S_{\cdot j}(a)|\bY(a), \bfX, M(a), N, \bC] \times Y_{\cdot j}(a)}{M(a)}  \bigg| N,\bC\right] \\
    &=  E\left[ \frac{\sum_{j=1}^{N}Y_{\cdot j}(a)}{N} \bigg| N,\bC\right]
    \\
    & = 
    E\left[ \overline{Y}(a) \big| N,\bC\right]
\end{align*}
Next, since $A \perp \bW$ and $\overline{Y}^o(a)$ is a function of $\bW$, then $E\left[\overline{Y}^o\big|A=a, N, \bC\right] = E[a\overline{Y}^o(1) + (1-a)\overline{Y}^o(0)|A=a, N,\bC] = E[Y(a)|N,\bC]$ as desired.

Finally, we prove $ E\left[ \overline{Y}(a)| N, \bC\right] = E\left[\overline{Y}^o\big|A=a, M, N, \bC\right]$. Assumption 3 implies that $E\left[ \overline{Y}(a)| N, \bC\right] = E\left[ \overline{Y}(a)| M(a), N, \bC\right]$. Using again the conditional distribution of $\bS(a)$, we have $E\left[ \overline{Y}(a)| M(a), N, \bC\right] = E\left[ \overline{Y}^o(a)| M(a), N, \bC\right]$. Since $A \perp \bW$, we get \\ $E\left[ \overline{Y}^o(a)| M(a), N, \bC\right] =  E\left[ \overline{Y}^o|A=a, M, N, \bC\right]$, which completes the proof.
\end{proof}

\begin{lemma}\label{lemma: asymptotics-theta}
Let $\bO_1, \ldots, \bO_m$ be i.i.d. samples from a common distribution on $\bO$.
Given the regularity conditions 1-4 for $\bpsi(\bO,\btheta)$, we have $\widehat{\btheta} \xrightarrow{P} \underline{\btheta}$ and $m^{1/2}(\widehat{\btheta}  - \underline{\btheta}) \xrightarrow{d} N(0,\mathrm{V})$, where $\mathrm{V} = E[IF(\bO,\underline\btheta)IF(\bO,\underline\btheta)^\top]$ and $IF(\bO,\underline\btheta) =- E\left[\frac{\partial \bpsi(\bO,\btheta)}{\partial \btheta} \big|_{\btheta = \underline\btheta}\right]^{-1}\bpsi(\bO,\underline\btheta)$ is the influence function for $\widehat\btheta$.
Furthermore, the sandwich variance estimator $m^{-1} \sum_{i=1}^m \widehat{IF}(\bO_i,\widehat{\btheta})\widehat{IF}(\bO_i,\widehat{\btheta})^\top$ converges in probability to $V$, where $\widehat{IF}(\bO_i,\widehat{\btheta}) = \left\{m^{-1} \sum_{i=1}^m \frac{\partial \bpsi(\bO_i,\btheta)}{\partial \btheta} \big|_{\btheta = \widehat\btheta}\right\}^{-1}\bpsi(\bO_i, \widehat{\btheta})$.
\end{lemma}

\begin{proof}
By regularity conditions 1 and 2, Example 19.8 of \cite{vaart_1998} implies that $\{\bpsi(\bO,\btheta): \btheta \in \mathrm{\Theta}\}$ is P-Glivenko-Cantelli. Then, with the regularity condition 3, Theorem 5.9 of \cite{vaart_1998} shows that $\widehat{\btheta} \xrightarrow{P} \underline{\btheta}$. We then apply Theorem 5.31 of \cite{vaart_1998} to obtain the asymptotic normality. Regularity conditions 1-4 implies that the assumptions for Theorem 5.31 of \cite{vaart_1998} are satisfied with $\widehat{\eta} = \eta = 
\underline\eta$ set to be a constant. Then we have
\begin{align*}
    m^{1/2}(\widehat{\btheta}  - \underline{\btheta}) &= - m^{-1/2}E\left[\frac{\partial \bpsi(\bO,\btheta)}{\partial \btheta} \Big|_{\btheta = \underline\btheta}\right]^{-1}\left\{E[\bpsi(\bO,\underline\btheta) ] +\frac{1}{m}\sum_{i=1}^m\bpsi(\bO_i,\underline\btheta) \right\} + o_p(1) \\
    &= m^{-1/2}\sum_{i=1}^mIF(\bO_i,\underline\btheta) + o_p(1),
\end{align*}
which implies the desired asymptotic normality by the Central Limit Theorem.

We next prove the consistency of the sandwich variance estimator. First, we prove that $m^{-1} \sum_{i=1}^m \frac{\partial \bpsi(\bO_i,\btheta)}{\partial \btheta} \big|_{\btheta = \widehat\btheta} \xrightarrow{P} E\left[\frac{\partial \bpsi(\bO,\btheta)}{\partial \btheta} \big|_{\btheta = \underline\btheta} \right]$. Denoting $\dot{\psi}_{ij}(\widehat\btheta)$ as the transpose of the $j$th row of $\frac{\partial \bpsi(\bO_i,\btheta)}{\partial \btheta}\big|_{\btheta = \widehat\btheta}$, we apply the multivariate Taylor expansion to get
\begin{align*}
    m^{-1} \sum_{i=1}^m \dot{\psi}_{ij}(\widehat\btheta) -     m^{-1} \sum_{i=1}^m \dot{\psi}_{ij}(\underline\btheta) =  m^{-1} \sum_{i=1}^m \ddot{\psi}_{ij}(\widetilde\btheta) (\widehat\btheta-\underline\btheta)
\end{align*}
for some $\widetilde{\btheta}$ on the line segment between $\widehat{\btheta}$ and $\underline\btheta$ and $ \ddot{\psi}_{ij}$ being the derivative of $\dot{\psi}_{ij}$. By the regularity condition 2 and $\widehat{\btheta} \xrightarrow{P} \btheta$, we have $m^{-1} \sum_{i=1}^m \ddot{\psi}_{ij}(\widetilde\btheta) = O_p(1)$. As a result, $\widehat{\btheta} -\btheta =o_p(1)$ implies that  $ m^{-1} \sum_{i=1}^m \dot{\psi}_{ij}(\widehat\btheta) -     m^{-1} \sum_{i=1}^m \dot{\psi}_{ij}(\underline\btheta)  = o_p(1)$. Then the first step is completed by the fact that $m^{-1} \sum_{i=1}^m \dot{\psi}_{ij}(\underline\btheta) = E[\dot{\psi}_{ij}(\underline\btheta) ] + o_p(1)$, which results from law of large number and the regularity condition 2. Next, we prove $m^{-1} \sum_{i=1}^m \psi(\bO_i,\widehat{\btheta})\psi(\bO_i,\widehat{\btheta})^\top \xrightarrow{P} E[\psi(\bO,\underline{\btheta})\psi(\bO,\underline{\btheta})^\top]$ following a similar procedure to the first step. Letting $\psi_{ij}(\btheta)$ be the $j$th entry of $\psi(\bO_i,\btheta)$, we apply the multivariate Taylor expansion and get
\begin{align*}
    &m^{-1} \sum_{i=1}^m \psi_{ij}(\widehat\btheta)\psi(\bO_i,\widehat\btheta) -   m^{-1} \sum_{i=1}^m \psi_{ij}(\underline\btheta)\psi(\bO_i,\underline\btheta)  \\
    &=  m^{-1} \sum_{i=1}^m \left\{\psi(\bO_i,\widetilde\btheta)\dot{\psi}_{ij}(\widetilde\btheta)^\top + \psi_{ij}(\widetilde\btheta) \frac{\partial \bpsi(\bO_i,\btheta)}{\partial \btheta} \Big|_{\btheta = \widetilde\btheta} \right\} (\widehat\btheta-\underline\btheta) \\
    &= O_p(1)o_p(1),
\end{align*}
which, combined with law of large numbers on $m^{-1} \sum_{i=1}^m \psi_{ij}(\underline\btheta)\psi(\bO_i,\underline\btheta)$, implies the desired result in this step. Finally, by the Continuous Mapping Theorem, we have
\begin{align*}
    &m^{-1} \sum_{i=1}^m \widehat{IF}(\bO_i,\widehat{\btheta})\widehat{IF}(\bO_i,\widehat{\btheta})^\top\\
    &= \left\{m^{-1} \sum_{i=1}^m \frac{\partial \bpsi(\bO_i,\btheta)}{\partial \btheta} \big|_{\btheta = \widehat\btheta}\right\}^{-1}m^{-1} \sum_{i=1}^m \psi(\bO_i,\widehat\btheta)\psi(\bO_i,\widehat\btheta)  \left\{m^{-1} \sum_{i=1}^m \frac{\partial \bpsi(\bO_i,\btheta)}{\partial \btheta} \big|_{\btheta = \widehat\btheta}\right\}^{-1} \\
    &= \left\{E\left[\frac{\partial \bpsi(\bO,\btheta)}{\partial \btheta} \big|_{\btheta = \underline\btheta}\right]+o_p(1)\right\}^{-1}\left\{E[\bpsi(\bO,\underline\btheta)\bpsi(\bO,\underline\btheta)^\top]+o_p(1)\right\}\left\{E\left[\frac{\partial \bpsi(\bO,\btheta)}{\partial \btheta} \big|_{\btheta = \underline\btheta}\right]+o_p(1)\right\}^{-1} \\
    &= V+o_p(1).
\end{align*}
\end{proof}

\subsection{Proof of Theorem 1}
\begin{proof}[Proof of Theorem~1]
We first introduce a few notations. For each cluster $i$, let $j_{i,1}< \cdots < j_{i,M_i}$ be the ordered list of indices such that the observed outcomes are $\bY_i^o = (Y_{i,j_{i,1}}, \ldots, Y_{i, j_{i,M_i}}) \in \mathrm{R}^{M_i}$ and the observed individual-level covariates are $\bfX_i^o = (\bX_{i,j_{i,1}}, \ldots, \bX_{i, j_{i,M_i}})^\top \in \mathrm{R}^{M_i \times p}$. We define $\bfH_i = [{e}_{j_{i,1}}^{N_i} \,\,\, {e}_{j_{i,2}}^{N_i} \,\,\, \ldots \,\,\, {e}_{j_{i,M_i}}^{N_i}] \in \mathbb{R}^{N_i \times M_i}$, where ${e}_j^{N} \in \mathbb{R}^N$ is the standard orthonormal basis with the $j$-th entry 1 and the remaining entries 0. This allows us to write $\bY_i^o = \bfH_i^\top \bY_i$,  $\bfX_i^o = \bfH_i^\top \bfX_i$, $\bone_{M_i} = \bfH_i^\top \bone_{N_i}$, $\bfH_i\bone_{M_i} = \bS_i$, $\bfH_i^\top\bfH_i = \bfI_{M_i}$ and $\bfH_i\bfH_i^\top = \textrm{diag}\{\bS_i\}$, where $\bY_i = (\bY_{i1}, \ldots, \bY_{iN_i})^\top$ and  $\bfX_i = (\bX_{i1}, \ldots, \bX_{iN_i})^\top$. We further define $\widetilde\bmu_i = (E[Y_{i1}|\bU_{i1}], \ldots, E[Y_{iN_i}|\bU_{iN_i}])^\top$, $\widetilde\bfD_i = \frac{\partial \widetilde\bmu_i}{\partial  \bbeta}$ and $\widetilde\bfZ_i = \textrm{diag}\{v(Y_{ij}):j=1,\ldots, N_i\}$, which imply $\bfD_i = \bfH_i^\top \widetilde\bfD_i$ and $\bfZ_i^{-1/2} = \bfH_i^\top\widetilde\bfZ_i^{-1/2}\bfH_i$. Then the estimating equation~(1) becomes
\begin{align*}
   \sum_{i=1}^m \widetilde\bfD_i^\top \bfH_i \left(\bfH_i^\top\widetilde\bfZ_i^{-1/2}\bfH_i\right) \bfR_i^{-1}(\rho) \left(\bfH_i^\top\widetilde\bfZ_i^{-1/2}\bfH_i\right) \bfH_i^\top (\bY_i - \bmu_i) = \bzero.
\end{align*}
To simplify the above equation, we observe that (i) the canonical link function $g$ implies that $\widetilde\bfD_i = \widetilde\bfZ_i \bfU_i$, where $\bfU_i = (\bU_{i1}, \ldots, \bU_{iN_i})^\top$; (ii) $(\bfH_i\bfH_i^\top)\widetilde\bfZ_i^{-1/2} = \widetilde\bfZ_i^{-1/2}(\bfH_i\bfH_i^\top)$ since  $\bfH_i\bfH_i^\top$ and $\widetilde\bfZ_i^{-1/2}$ are both diagonal matrices; and (iii) $\bfR_i^{-1}(\rho) = \frac{1}{1-\rho} \bfI_{M_i} - \frac{\rho}{(1-\rho)(1+M\rho-\rho)}\bone_{M_i}\bone_{M_i}^\top$, altogether yielding
\begin{align*}
    &\sum_{i=1}^m \bfU_i^\top \widetilde\bfZ_i^{1/2} \bfH_i\left\{\frac{1}{1-\rho} \bfI_{M_i} - \frac{\rho}{(1-\rho)(1+M\rho-\rho)}\bone_{M_i}\bone_{M_i}^\top\right\} \bfH_i^\top \widetilde\bfZ_i^{-1/2}  (\bY_i - \bmu_i) = \bzero \\
     \Longleftrightarrow &\sum_{i=1}^m \bfU_i^\top \left\{\frac{1}{1-\rho}\textrm{diag}\{\bS_i\} - \frac{\rho}{(1-\rho)(1+M\rho-\rho)}\widetilde\bfZ_i^{1/2} \bS_i\bS_i^\top \widetilde\bfZ_i^{-1/2} \right\}  (\bY_i - \bmu_i) = \bzero \\
      \Longleftrightarrow &\sum_{i=1}^m \bfU_i^\top \mathrm{B}_i (\bY_i - \bmu_i) = \bzero,
\end{align*}
where $ \mathrm{B}_i = \frac{1}{1-\rho}\textrm{diag}\{\bS_i\} - \frac{\rho}{(1-\rho)(1+M_i\rho-\rho)}\widetilde\bfZ_i^{1/2} \bS_i\bS_i^\top \widetilde\bfZ_i^{-1/2}$.

Denoting $\btheta_C = (\mu_C(1), \mu_C(0), \bbeta, \rho)$ and  $\btheta_I = (\mu_I(1), \mu_I(0), \bbeta, \rho)$, we define the following estimating equations
\begin{align}
    \bpsi_C(\bO, \btheta_C) &= \left(\begin{array}{c}
        \mu_C(1) - \frac{1}{M}\bS^\top \bmu(1,\bbeta) \\
        \mu_C(0) - \frac{1}{M}\bS^\top \bmu(0,\bbeta) \\
        \bfU^\top \mathrm{B} \{\bY - \bmu(A,\bbeta)\} \\
        k(N)\rho - h(\bY,\bU,M,\bbeta)
    \end{array}\right),\label{eq: gee-psi-C} \\ 
    \bpsi_I(\bO,\btheta_I) &= \left(\begin{array}{c}
        N\mu_I(1) - \frac{N}{M}\bS^\top \bmu(1,\bbeta)  \\
        N\mu_I(0) - \frac{N}{M}\bS^\top \bmu(0,\bbeta)  \\
        N\bfU^\top \mathrm{B} \{\bY - \bmu(A,\bbeta)\} \\
        k(N)\rho - h(\bY,\bU,M,\bbeta)
    \end{array}\right), \label{eq: gee-psi-I}
\end{align}
where $\bmu(a,\bbeta)$ is the assumed mean function setting $A = a$ and $ k(N)\rho - h(\bY,\bU,M,\bbeta)$ is a prespecified estimating function for constructing the moment estimator for $\rho$, i.e., $k(N) = N(N-1)/2-p-q-3$ and $h(\bY,\bU,M,\bbeta) = \sum_{j<j'} S_{\cdot j}S_{\cdot j'}\frac{Y_{\cdot j}-\mu_{\cdot j}}{v(Y_{\cdot j})}\frac{Y_{\cdot j'}-\mu_{\cdot j'}}{v(Y_{\cdot j'})}$. Of note, if an independence working correlation structure is used, then the last estimating function simplifies to $\rho$. Then the GEE-g method described in the main paper equivalently solves $\sum_{i=1}^m \bpsi_{C}(\bO_i, \btheta_C) = \bzero$ and $\sum_{i=1}^m \bpsi_{I}(\bO_i,\btheta_I) = \bzero$.

We next show that $\widehat\mu_C^{\textrm{GEE-g}}(a)$ converges in probability to $\mu_C(a)$ for $a = 0,1$ given any condition among (S1-S4) described in the theorem statement. Given the regularity conditions, Lemma~\ref{lemma: asymptotics-theta} implies $\widehat{\btheta}_C \xrightarrow{P} \underline\btheta_C$, where $\underline\btheta_C = (\underline\mu_C^{\textrm{GEE-g}}(1), \underline\mu_C^{\textrm{GEE-g}}(0), \underline\bbeta_C, \underline\rho_C)$ solves $E[\bpsi_C(\bO,\underline\btheta_C)] = \bzero$.
Since Lemma~\ref{lemma1} implies $E[\frac{1}{M}\bS^\top \bmu(a,\bbeta)] = E[\frac{1}{M}E[\bS|M,N]^\top \bmu(a,\bbeta)] = E[\frac{1}{N}\bone_N^\top \bmu(a,\bbeta)]$, then $E[\bpsi_C(\bO,\underline\btheta_C)] = \bzero$ yields $\underline\mu_C^{\textrm{GEE-g}}(a) = E[\frac{\bone_N^\top \bmu(a,\underline\bbeta_C)}{N}]$.
To connect $\underline\mu_C^{\textrm{GEE-g}}(a)$ to $\mu_C(a)$, we focus on the equation $E[\bfU^\top \mathrm{B} \{\bY - \bmu(A,\underline\bbeta_C)\}] = \bzero$, which is implied by $E[\bpsi_C(\bO,\underline\btheta_C)] = \bzero$. First consider condition (S1) that  the mean model is correctly specified, i.e., $E[\bY|\bfU] = \bmu(A,\bbeta^*)$ for some $\bbeta^*$. Then we get $E[\bfU^\top E[\mathrm{B}|\bfU] \{ \bmu(A,\bbeta^*) - \bmu(A,\underline\bbeta_C)\}] = \bzero$ by taking the conditional expectation on $\bfU$ and the fact that $\bS \perp \bY |\bfU$, thereby indicating that $\underline\bbeta_C = \bbeta^*$ is a solution. By the regularity condition that $\bbeta_C$ is unique, we get $\underline\bbeta_C = \bbeta^*$, and hence $\underline\mu_C^{\textrm{GEE-g}}(a) = E[\frac{\bone_N^\top \bmu(a,\underline\bbeta_C)}{N}] = E[\frac{\bone_N^\top E[\bY(a)|\bfU]}{N}] = \mu_C(a)$. Of note, the above derivation does not use the Assumption 4. 

Next consider condition (S2) that an independence working correlation structure is used, which implies that $\mathrm{B} = \textrm{diag}\{\bS\}$. Since $\bfU = (\bone_N, A\bone_N, \mathrm{L})$, then the first two equations of $E[\bfU^\top \mathrm{B} \{\bY - \bmu(A,\underline\bbeta_C)\}] = \bzero$ are
\begin{align*}
    E[\bS^\top \{\bY - \bmu(A,\underline\bbeta_C)\}] &= 0, \\
    E[A\bS^\top \{\bY - \bmu(A,\underline\bbeta_C)\}] &= 0.
\end{align*}
By Lemma~\ref{lemma1}, $\bS \perp (A, \bY, \bX, \bC) | (M, N)$ and $E[\bS|M,N] = \frac{M}{N} \bone_N$. Hence, the above two equations imply that $E[M(a) \overline{Y}(a)] = E[\frac{M(a)}{N}\bone_N^\top \bmu(a,\underline\bbeta_C)]$ for $a=0,1$. By Assumption 4, we get $E[\overline{Y}(a)] = E[\frac{1}{N}\bone_N^\top \bmu(a,\underline\bbeta_C)]$ and hence $\mu_C(a)= \underline\mu_C^{\textrm{GEE-g}}(a)$. Next consider condition (c) that $g$ is the identity link function, i.e., $\bY$ is the continuous outcome. Since we use the canonical link function, then $\widetilde{\bfZ}_i = \sigma^2 \bfI_{N_i}$, indicating $\bone^\top_{N_i}\mathrm{B}_i = \frac{1}{1-\rho} \bS_i^\top - \frac{M_i\rho}{(1-\rho)(1+M_i\rho-\rho)} \bS_i^\top = \frac{1}{1+(M_i-1)\rho}\bS_i^\top$. Then the first two equations of $E[\bfU^\top \mathrm{B} \{\bY - \bmu(A,\underline\bbeta_C)\}] = \bzero$ are
\begin{align*}
    E\left[\frac{1}{1+(M-1)\underline{\rho}_C}\bS^\top \{\bY - \bmu(A,\underline\bbeta_C)\}\right] &= 0, \\
    E\left[\frac{1}{1+(M-1)\underline{\rho}_C}A\bS^\top \{\bY - \bmu(A,\underline\bbeta_C)\}\right] &= 0.
\end{align*}
Following a similar procedure as for condition (S2), we get $\mu_C(a)= \underline\mu_C^{\textrm{GEE-g}}(a)$ since $\frac{M}{1+(M-1)\underline{\rho}_C} \perp (\bY, \bfU)$ by Assumptions 3 and 4. Finally, for condition (S4) that GEE excludes $\bfX$, we get $\bmu(a, \bbeta) = \bone_N \mu_0(a, \bbeta)$ for a scalar $\mu_0(a, \bbeta)$ and $\widetilde\bfZ = q(\mu(A,\bbeta)) \bfI_N$ for some function $q$ since $v(Y_{ij})$ is a function of $\mu_{ij}$ in the assumed model. Then it is straightfroward that $\bone^\top_{N_i}\mathrm{B}_i = \frac{1}{1+(M_i-1)\rho}\bS_i^\top$, leading to the same estimating equations as for condition (S3). Therefore, the proof is the same. Given the result that $\widehat\mu_C^{\textrm{GEE-g}}(a) \xrightarrow{P} \mu_C(a)$ for $a = 0,1$, the Continuous Mapping Theorem and regularity conditions for $f$ imply that $\widehat{\Delta}_C^{\textrm{GEE-g}} \xrightarrow{P} \Delta_C$.

We next show that $\widehat\mu_I^{\textrm{GEE-g}}(a)$ converges in probability to $\mu_I(a)$ for $a = 0,1$. Likewise, $\widehat{\btheta}_I \xrightarrow{P} \underline\btheta_I$, where $\underline\btheta_I = (\underline\mu_I^{\textrm{GEE-g}}(1), \underline\mu_I^{\textrm{GEE-g}}(0), \underline\bbeta_I, \underline\rho_I)$ solves $E[\bpsi_I(\bO,\underline\btheta_I)] = \bzero$, and $\underline\mu_I^{\textrm{GEE-g}}(a) = \frac{E[\bone_N^\top \bmu(a,\underline\bbeta_I)]}{E[N]}$. The proof is similar to the proof for showing the consistency of $\widehat\mu_C(a)$ except that the estimating equation $E[N\bfU^\top \mathrm{B} \{\bY - \bmu(A,\underline\bbeta_C)\}] = \bzero$ involves a factor $N$. Consequently, conditions (S1-S4) will yield $E[\bone_N^\top \bY(a)] = E[\bone_N^\top \bmu(a, \underline\bbeta_I)]$, implying the desired consistency result. 

We next prove the asymptotic normality. By the regularity conditions, Lemma~\ref{lemma: asymptotics-theta} implies that
\begin{align*}
    m^{1/2}(\widehat{\btheta}_C-\underline\btheta_C) &= m^{-1/2}\sum_{i=1}^m\bfG_{C}^{-1} \bpsi_{C}(\bO_i,\underline\btheta_C) + o_p(\bone), \\
    m^{1/2}(\widehat{\btheta}_I-\underline\btheta_I) &= m^{-1/2}\sum_{i=1}^m\bfG_I^{-1} \bpsi_{I}(\bO_i,\underline\btheta_I) + o_p(\bone),
\end{align*}
where $\bfG_C = E\left[\frac{\partial }{\partial \btheta} \bpsi_C(\bO,\btheta_C) \big |_{\btheta_C = \underline\btheta_C}\right]$ and $\bfG_I = E\left[\frac{\partial }{\partial \btheta} \bpsi_I(\bO,\btheta_I) \big |_{\btheta_I = \underline\btheta_I}\right]$. Thus, the asymptotic normality is obtained by Central Limit Theorem.

To get the variance for $\widehat{\Delta}_C^{\textrm{GEE-g}}$ and $\widehat{\Delta}_I^{\textrm{GEE-g}}$, we compute the influence function for $(\widehat\mu_C^{\textrm{GEE-g}}(1), \widehat\mu_C^{\textrm{GEE-g}}(0))$ and $(\widehat\mu_I^{\textrm{GEE-g}}(1), \widehat\mu_I^{\textrm{GEE-g}}(0))$ by computing $\bfG_C^{-1}$ and  $\bfG_I^{-1}$. Under condition (S1-4), tedious algebra shows that the first two rows of $\bfG_C^{-1}$ are
\begin{align*}
    \left[\begin{array}{ccccc}
    1    &  0 & 0 & -\frac{1}{\pi E\left[\frac{M(1)}{1 + \{M(1)-1\}\underline{\rho}_C}\right]}  & \bzero\\
    0    & 1 & -\frac{1}{(1-\pi) E\left[\frac{M(0)}{1 + \{M(0)-1\}\underline{\rho}_C}\right]}  & \frac{1}{(1-\pi) E\left[\frac{M(0)}{1 + \{M(0)-1\}\underline{\rho}_C}\right]} & \bzero
    \end{array}\right],
\end{align*}
and the influence function for $(\widehat\mu_C^{\textrm{GEE-g}}(1), \widehat\mu_C^{\textrm{GEE-g}}(0))$ is 
\begin{align*}
    {IF}_C^{\textrm{GEE-g}} 
   = \left(\begin{array}{c}
    \frac{A}{\pi} E\left[\frac{M(1)}{1 + \{M(1)-1\}\underline{\rho}_C}\right]^{-1} \frac{1}{1 + (M-1)\underline{\rho}_C}\bS^\top\{\bY-\bmu(A,\underline\bbeta_C)\}   + \frac{1}{M}\bS^\top\bmu(1,\underline\bbeta_C) -\mu_C(1)    \\
     \frac{1-A}{1-\pi} E\left[\frac{M(0)}{1 + \{M(0)-1\}\underline{\rho}_C}\right]^{-1} \frac{1}{1 + (M-1)\underline{\rho}_C}\bS^\top\{\bY-\bmu(A,\underline\bbeta_C)\}   + \frac{1}{M}\bS^\top\bmu(0,\underline\bbeta_C) -\mu_C(0)
    \end{array}\right),
\end{align*}
leading to the asymptotic covariance matrix as $E[{IF}_C^{\textrm{GEE-g}}{IF}_C^{\textrm{GEE-g}}{^\top}]$. Then $V_C^{\textrm{GEE-g}}$, the asymptotic variance of $\widehat{\Delta}_C^{\textrm{GEE-g}}$, is $\nabla f^\top E[{IF}_C^{\textrm{GEE-g}}{IF}_C^{\textrm{GEE-g}}{^\top}] \nabla f $ by the Delta method, where $\nabla f$ is the gradient of $f$ evaluated at $(\mu_C(1), \mu_C(0))$. For $\btheta_I$, a similar procedure gives the influence function $(\widehat\mu_I^{\textrm{GEE-g}}(1), \widehat\mu_I^{\textrm{GEE-g}}(0))$ as 
\begin{align*}
    {IF}_I^{\textrm{GEE-g}} = \frac{N}{E[N]}\left(\begin{array}{c}
    \frac{A}{\pi} E\left[\frac{M(1)}{1 + \{M(1)-1\}\underline{\rho}_I}\right]^{-1} \frac{1}{1 + (M-1)\underline{\rho}_I}\bS^\top\{\bY-\bmu(A,\underline\bbeta_I)\}   + \frac{1}{M}\bS^\top\bmu(1,\underline\bbeta_I) -\mu_I(1)    \\
     \frac{1-A}{1-\pi} E\left[\frac{M(0)}{1 + \{M(0)-1\}\underline{\rho}_I}\right]^{-1} \frac{1}{1 + (M-1)\underline{\rho}_I}\bS^\top\{\bY-\bmu(A,\underline\bbeta_I)\}   + \frac{1}{M}\bS^\top\bmu(0,\underline\bbeta_I) -\mu_I(0)
    \end{array}\right),
\end{align*}
and the asymptotic variance $V_I^{\textrm{GEE-g}}$ is then $\nabla f^\top E[{IF}_I^{\textrm{GEE-g}}{IF}_I^{\textrm{GEE-g}}{^\top}] \nabla f $. 
By Lemma~\ref{lemma: asymptotics-theta}, $m \widehat{V}_C^{\textrm{GEE-g}} \xrightarrow{P} V_C^{\textrm{GEE-g}} $ and  $m \widehat{V}_I^{\textrm{GEE-g}} \xrightarrow{P} V_I^{\textrm{GEE-g}} $. Finally, Slutsky's Theorem implies the desired results.
\end{proof}
\subsection{Proof of Theorem 2}

\begin{proof}[Proof of Theorem~2] For continuous outcomes where $g$ is the identity link, the likelihood function becomes
\begin{align*}
    &\prod_{i=1}^m \int \prod_{j:S_{ij}=1}(2\pi \sigma^2)^{-1/2} \exp\left\{-\frac{1}{2\sigma^2}(Y_{ij} - \bU_{ij}^\top\balpha - b_i)^2\right\}(2\pi \tau^2)^{-1/2}\exp\left\{-\frac{b_i^2}{2\tau^2}\right\} \mathrm{d}b \\
    &= \prod_{i=1}^m |2\pi \bfSigma_i|^{-1/2} \exp\{- (\bY_i^o-\bfU_i^o\balpha)^\top\bfSigma_i^{-1}(\bY_i^o-\bfU_i^o\balpha)/2\},
\end{align*}
where $\bfSigma_i = \sigma^2 \bfI_{M_i} + \tau^2 \bone_{M_i}\bone_{M_i}^\top$, $\bY_i^o = \{Y_{ij}: S_{ij} = 1\}$ is the observed outcome vector, and $\bfU_i^o = \{\bU_{ij}: S_{ij} = 1\} \in \mathbb{R}^{M_i \times (p+q+3)}$ is the observed design matrix. We further define $\bfH_i$ as the in the proof of Theorem 1, $\widetilde\bfSigma_i = \sigma^2 \bfI_{N_i} + \tau^2 \bone_{N_i}\bone_{N_i}^\top$, $\bfU_i =(\bU_{i1}, \ldots, \bU_{iN_i})^\top$, and the log-likelihood function is
\begin{align*}
 -\frac{1}{2}\sum_{i=1}^m\left\{2\pi M_i + \log(|\bfH_i^\top\widetilde\bfSigma_i \bfH_i|) + (\bY_i - \bfU_i\balpha)^\top \bfH_i (\bfH_i\widetilde\bfSigma_i \bfH_i)^{-1}\bfH_i^\top (\bY_i - \bfU_i\balpha) \right\}, 
\end{align*}
whose derivative over $(\balpha, \sigma^2, \tau^2)$ is 
\begin{align*}
    -\frac{1}{2}\sum_{i=1}^m \left(\begin{array}{c}
    \bfU_i^{\top}\bfV_i(\bY_i - \bfU_i\balpha)  \\
    -\textrm{tr}(\bfV_i) + (\bY_i - \bfU_i\balpha)^\top \bfV_i^2(\bY_i - \bfU_i\balpha)\\
  -\bone_{N_i}^\top\bfV_i\bone_{N_i} + (\bY_i - \bfU_i\balpha)^\top \bfV_i\bone_{N_i}\bone_{N_i}^\top \bfV_i(\bY_i - \bfU_i\balpha)    
   \end{array}\right),
\end{align*}
where $\bfV_i = \bfH_i \left(\bfH_i^\top\widetilde\bfSigma_i \bfH_i\right)^{-1} \bfH_i^\top$ and $\textrm{tr}(\bfV_i)$ is the trace of $\bfV_i$. We then define the estimating equations as
\begin{equation}\label{eq: lmm-psi-C}
    \bpsi_C(\bO,\bfeta_C) = \left(\begin{array}{c}
    \mu_C(1) - \frac{1}{M}\bS^\top\bfU(1)\balpha \\
    \mu_C(0) - \frac{1}{M}\bS^\top\bfU(0)\balpha \\
    \bfU^{\top}\bfV(\bY - \bfU\balpha)  \\
    -\textrm{tr}(\bfV) + (\bY - \bfU\balpha)^\top \bfV^2(\bY - \bfU\balpha)\\
  -\bone_{N}^\top\bfV\bone_{N} + (\bY - \bfU\balpha)^\top \bfV\bone_{N}\bone_{N}^\top \bfV(\bY - \bfU\balpha)    
   \end{array}\right),
\end{equation}
where $\bfeta_C = (\mu_C(1), \mu_C(0), \balpha, \sigma^2, \tau^2)^\top$ and $\bfU(a)$ is $\bfU$ with $A$ substituted by $a$. The MLE $\widehat\bfeta_C$ for $\bfeta_C$ is then a solution to $\sum_{i=1}^m \bpsi_{C}(\bO_i,\bfeta_C) = \bzero$. 

We next prove consistency of $\widehat{\Delta}_C^{\textrm{LMM-g}}$ following a similar proof to the consistency of $\widehat\Delta_C^{\textrm{GEE-g}}$. Lemma~\ref{lemma: asymptotics-theta} implies that $\widehat{\bfeta}_C \xrightarrow{P} \underline\bfeta_C$, where $\underline\bfeta = (\underline\mu_C^{\textrm{LMM-g}}(1), \underline\mu_C^{\textrm{LMM-g}}(0), \underline\bfeta_C, \underline\sigma_C^2, \underline\tau_C^2)$ solves $E[\bpsi_C(\bO,\underline\bfeta_C)] = \bzero$. Then we have $E[\bfU^{\top}\underline\bfV(\bY - \bfU\underline\balpha)]=0 $, whose first two equations are
\begin{align*}
    E[\bone_N^\top \underline\bfV (\bY-\bfU\underline\balpha)] &= \bzero, \\
    E[A\bone_N^\top \underline\bfV (\bY-\bfU\underline\balpha)] &= \bzero.
\end{align*}
Since $\bfSigma_i^{-1} = \frac{1}{\sigma^2} \bfI_{M_i} - \frac{\tau^2}{\sigma^2(\sigma^2+M_i\tau^2)}\bone_{M_i}\bone_{M_i}^\top$, we get $\bone_N^\top\underline\bfV = \frac{1}{\underline\sigma_C^2 + M\underline\tau_C^2} \bS^\top$, which implies 
\begin{align*}
    E\left[\frac{1}{\underline\sigma_C^2 + M\underline\tau_C^2} \bS^\top (\bY-\bfU\underline\balpha)\right] &= \bzero, \\
    E\left[A\frac{1}{\underline\sigma_C^2 + M\underline\tau_C^2} \bS^\top (\bY-\bfU\underline\balpha)\right] &= \bzero.
\end{align*}
By Assumption 3 and Lemma~\ref{lemma1}, we get $\bS \perp (A, \bY, \bU) | (M,N)$ and $E[\bS|M,N] = \frac{M}{N} \bone_{N}$, which implies $E[\frac{M(a)}{\underline\sigma_C^2 + M(a)\underline\tau_C^2} (\overline{Y}(a)-\overline{\bU}(a)\underline\balpha)] = 0$ for $a = 0,1$, where $\overline{\bU}(a) = \bone_N^\top \bfU(a)/N$. Given Assumption 4, we got $\mu_C(a) = E[\overline{Y}(a)] =  E[\overline{\bU}(a)\underline\balpha] = \underline\mu_C(a)$. Then $\widehat{\Delta}_C \xrightarrow{P} \Delta_C$ can be obtained by the Continuous Mapping Theorem.

We next prove the asymptotic normality, which is also similar to the corresponding part in the proof of Theorem 1. By the regularity conditions, Lemma~\ref{lemma: asymptotics-theta} implies that
\begin{align*}
    m^{1/2}(\widehat{\bfeta}_C-\underline\bfeta_C) &= m^{-1/2}\sum_{i=1}^m\bfG_{C}^{-1} \bpsi_{C}(\bO_i,\underline\bfeta_C) + o_p(\bone), 
\end{align*}
where $\bfG_C = E\left[\frac{\partial }{\partial \bfeta} \bpsi_C(\bO,\bfeta_C) \big |_{\bfeta_C = \underline\bfeta_C}\right]$. Thus, the asymptotic normality is obtained by Central Limit Theorem. The influence function for $(\widehat\mu_C^{\textrm{LMM-g}}(1), \widehat\mu_C^{\textrm{LMM-g}}(0))$ is 
\begin{align*}
    {IF}_C^{\textrm{LMM-g}} = \left(\begin{array}{c}
    \frac{A}{\pi} E\left[\frac{M(1)}{ \underline\sigma_C^2 + M(1)\underline\tau_C^2}\right]^{-1} \frac{1}{\underline\sigma_C^2 + M\underline\tau_C^2} \bS^\top\{\bY-\bfU \underline{\balpha}\}   + \frac{1}{M}\bS^\top\bfU(1) \underline{\balpha} -\mu_C(1)    \\
     \frac{1-A}{1-\pi} E\left[\frac{M(0)}{ \underline\sigma_C^2 + M(0)\underline\tau_C^2}\right]^{-1}\frac{1}{\underline\sigma_C^2 + M\underline\tau_C^2} \bS^\top\{\bY-\bfU \underline{\balpha}\}   + \frac{1}{M}\bS^\top\bfU(0) \underline{\balpha} -\mu_C(0)
    \end{array}\right),
\end{align*}
and $V_C^{\textrm{LMM-g}} = \nabla f^\top E[{IF}_C^{\textrm{LMM-g}}{IF}_C^{\textrm{LMM-g}}{^\top}] \nabla f$,  where $\nabla f$ is the gradient of $f$ evaluated at $(\mu_C(1), \mu_C(0))$.

For estimating $\Delta_I$, we can follow the same procedure to get the desired result and hence omit the detailed proof here. The corresponding estimating equations are
\begin{align}\label{eq: lmm-psi-I}
    \bpsi_I(\bO,\bfeta_I) = \left(\begin{array}{c}
    N\mu_I(1) - \frac{N}{M}\bS^\top\bfU(1)\balpha \\
    N\mu_C(0) - \frac{N}{M}\bS^\top\bfU(0)\balpha \\
    N\bfU^{\top}\bfV(\bY - \bfU\balpha)  \\
    -\textrm{tr}(\bfV) + N(\bY - \bfU\balpha)^\top \bfV^2(\bY - \bfU\balpha)\\
  -\bone_{N}^\top\bfV\bone_{N} + N(\bY - \bfU\balpha)^\top \bfV\bone_{N}\bone_{N}^\top \bfV(\bY - \bfU\balpha)    
   \end{array}\right),
\end{align}
where $\bfeta_I = (\mu_I(1), \mu_I(0), \balpha, \sigma^2, \tau^2)^\top$, and the influence function for $(\widehat\mu_I^{\textrm{LMM-g}}(1), \widehat\mu_I^{\textrm{LMM-g}}(0))$ is 
\begin{align*}
    {IF}_I^{\textrm{LMM-g}} = \frac{N}{E[N]}\left(\begin{array}{c}
    \frac{A}{\pi} E\left[\frac{M(1)}{ \underline\sigma_I^2 + M(1)\underline\tau_I^2}\right]^{-1}\frac{1}{\underline\sigma_I^2 + M\underline\tau_I^2} \bS^\top\{\bY-\bfU \underline{\balpha}\}   + \frac{1}{M}\bS^\top\bfU(1) \underline{\balpha} -\mu_I(1)    \\
     \frac{1-A}{1-\pi} E\left[\frac{M(0)}{ \underline\sigma_I^2 + M(0)\underline\tau_I^2}\right]^{-1}\frac{1}{\underline\sigma_I^2 + M\underline\tau_I^2} \bS^\top\{\bY-\bfU \underline{\balpha}\}   + \frac{1}{M}\bS^\top\bfU(0) \underline{\balpha} -\mu_I(0)
    \end{array}\right).
\end{align*}
Then $V_I^{\textrm{LMM-g}} = \nabla f^\top E[{IF}_I^{\textrm{LMM-g}}{IF}_I^{\textrm{LMM-g}}{^\top}] \nabla f$.

The consistency of the sandwich variance estimators for $V_C^{\textrm{LMM-g}}$ and $V_I^{\textrm{LMM-g}}$ are implied by Lemma~\ref{lemma: asymptotics-theta} and the Continuous Mapping Theorem; and Slutsky's theorem implies the desired convergence results.
\end{proof}

\subsection{Proof of Theorem 3}
\begin{proof}[Proof of Theorem 3]
By Lemma~\ref{lemma2} and the iteration of conditional expectation, we have
\begin{align*}
    & E\left[\overline{Y}(a) \right]\\
    &= E\left[E\left[E\left[\overline{Y}^o|A=a,N,\bC\right]|N\right]\right]\\
    &= E\left[E\left[E\left[E\left[E\left[\overline{Y}^o|A=a,\bfX^o, M,N,\bC\right]|A=a,M,N,\bC\right]|A=a,N,\bC\right]|N\right]\right].
\end{align*}
For $E\left[\overline{Y}^o|A=a,\bfX^o, M,N,\bC\right]$, since $(\bfX^o, M)$ is a deterministic function of $\bfX, \bS$, we could denote $f(\bfX,\bS,N,\bC) =E\left[\overline{Y}^o|A=a,\bfX^o, M,N,\bC\right]$ for some function $f$. By Lemma~\ref{lemma1} (iii), we have $f(\bfX,\bS,N,\bC) \perp A | (M,N,\bC)$, which implies 
\begin{align*}
    & E\left[E\left[\overline{Y}^o|A=a,\bfX^o, M,N,\bC\right]|A=a,M,N,\bC\right] \\
    &\quad = E\left[f(\bfX,\bS,N,\bC)|A=a,M,N,\bC\right]\\
        &\quad = E\left[f(\bfX,\bS,N,\bC)|M,N,\bC\right]\\
    &\quad = E\left[E\left[\overline{Y}^o|A=a,\bfX^o, M, N,\bC\right]|M,N,\bC\right].
\end{align*}
Therefore, the estimand can be written as
\begin{align*}
     E\left[\frac{\sum_{j=1}^{N} Y_{\cdot j}(a)}{N} \right] = E\left[E\left[E\left[E\left[E\left[\overline{Y}^o|A=a,\bfX^o, M, N,\bC\right]|M,N,\bC\right]|A=a,N,\bC\right]|N\right]\right].
\end{align*}

Next, we follow the steps provided by \cite{hines2022demystifying} to compute the EIF of the target estimand. Denote the observed data for each cluster as $\bO_i = (\bY_i^o, \bfX_i^o, M_i, A_i, N_i, \bC_i)$, the observed data distribution as $\mathcal{P} = \mathcal{P}^{\bY^o| \bfX^o, M, A, N, \bC} \mathcal{P}^{\bfX^o |M, N, \bC} \mathcal{P}^{M| A, N, \bC} \mathcal{P}^{\bC|N} \mathcal{P}^{N}\mathcal{P}^{A}$, $ E\left[\frac{\sum_{j=1}^{N} Y_{\cdot j}(a)}{N} \right] = \Psi(\mathcal{P})$, and a parametric submodel $\mathcal{P}_t = t \widetilde{\mathcal{P}} + (1-t)\mathcal{P}$ for $t \in [0,1]$, where $\widetilde{\mathcal{P}}$ is a point-mass at $o_i$ in the support of $\mathcal{P}$. Furthermore, let $f$ denote the density of $\mathcal{P}$, $1_{\widetilde{o}}(o)$ denote the Dirac delta function for $\widetilde{o}$, i.e., the density of $\widetilde{\mathcal{P}}$, and $f_t = t 1_{\widetilde{o}}(o) + (1-t) f$. Then
\begin{align*}
    &\frac{d \Psi(\mathcal{P}_t)}{d t}\bigg|_{t=0}\\
    &= \frac{d}{d t}\int \overline{y}^o f_t(\by^o|a, \bfx^o, m, \bc, n) f_t( \bfx^o| m, \bc, n) f_t(m|a,\bc,n)f(\bc|n) f_t(n) d\by^o d\bfx^o dm d\bc dn \bigg|_{t=0} \\
    &= \int \overline{y}^o f(\by^o|a, \bfx^o, m, \bc, n) f( \bfx^o| m, \bc, n) f(m|a,\bc,n)f(\bc|n) f(n) \\
    &\quad \left\{\frac{1_{\widetilde{o}}(o)}{f(\by^o,a, \bfx^o, m, \bc, n)} - \frac{1_{\widetilde{a}, \widetilde{\bfx^o}, \widetilde{m}, \widetilde{\bc}, \widetilde{n}}(a, \bfx^o, m, \bc, n)}{f(a, \bfx^o, m, \bc, n)} + \frac{1_{\widetilde{\bfx^o}, \widetilde{m}, \widetilde{\bc}, \widetilde{n}}(\bfx^o, m, \bc, n)}{f( \bfx^o, m, \bc, n)}-\frac{1_{\widetilde{m}, \widetilde{\bc}, \widetilde{n}}(m, \bc, n)}{f(m, \bc, n)}\right.\\
    &\quad\quad  \left.+ \frac{1_{\widetilde{a},  \widetilde{m}, \widetilde{\bc}, \widetilde{n}}(a,m, \bc, n)}{f(a, m, \bc, n)} - \frac{1_{\widetilde{a}, \widetilde{\bc}, \widetilde{n}}(a, \bc, n)}{f(a, \bc, n)} + \frac{1_{\widetilde{\bc}, \widetilde{n}}(\bc, n)}{f(\bc, n)}-1\right\}d\by^o d\bfx^o dm d\bc dn \\
    &= \frac{1_{\widetilde{a}}(a)\{\overline{\widetilde{y}}^o-E[\overline{Y}^o|A=a, \bfX^o=\widetilde{\bfx}^o, M=\widetilde{m}, \bC=\widetilde{\bc}, N=\widetilde{n}]\}pr(A=a| M=\widetilde{m}, \bC=\widetilde{\bc}, N=\widetilde{n})}{pr(A=a|\bfX^o=\widetilde{\bfx}^o, M=\widetilde{m}, \bC=\widetilde{\bc}, N=\widetilde{n})pr(A=a|\bC=\widetilde{\bc}, N=\widetilde{n})}\\
    &\quad  + \{E[\overline{Y}^o|A=a, \bfX^o=\widetilde{\bfx}^o, M=\widetilde{m}, \bC=\widetilde{\bc}, N=\widetilde{n}] - E[\overline{Y}^o|A=a, M=\widetilde{m}, \bC=\widetilde{\bc}, N=\widetilde{n}]\}\\
    &\qquad \times \frac{pr(M = \widetilde{m}|A=a, \bC=\widetilde{\bc}, N=\widetilde{n})}{pr(M = \widetilde{m}| \bC=\widetilde{\bc}, N=\widetilde{n})} \\
    &\quad + \frac{1_{\widetilde{a}}(a)}{pr(A=a|\bC=\widetilde{\bc}, N=\widetilde{n})}\{E[\overline{Y}^o|A=a, M=\widetilde{m}, \bC=\widetilde{\bc}, N=\widetilde{n}] - E[\overline{Y}^o|A=a, \bC=\widetilde{\bc}, N=\widetilde{n}]\} \\
    &\quad + E[\overline{Y}^o|A=a, \bC=\widetilde{\bc}, N=\widetilde{n}] - \Psi(\mathcal{P}).
\end{align*}
By Lemma~\ref{lemma1} (iii) and Assumptions 2-3, we get $pr(A=a|\bfX^o=\widetilde{\bfx}^o, M=\widetilde{m}, \bC=\widetilde{\bc}, N=\widetilde{n}) = pr(A=a| M=\widetilde{m}, \bC=\widetilde{\bc}, N=\widetilde{n})$ and $pr(A=a|\bC=\widetilde{\bc}, N=\widetilde{n}) = pr(A=a)$. Therefore,
\begin{align*}
    EIF_a &= \frac{I\{A=a\}}{pr(A=a)}\left\{\overline{Y}^o - E\left[\overline{Y}^o\big|A=a, \bfX^o, M, N, \bC\right] \right\} \\
    &\quad + \frac{pr(M|A=a,N,\bC)}{pr(M|N,\bC)}\left\{ E\left[\overline{Y}^o\big|A=a, \bfX^o, M, N, \bC\right] - E\left[\overline{Y}^o\big|A=a,  M, N, \bC\right]\right\} \\
    &\quad + \frac{I\{A=a\}}{pr(A=a)}\left\{E\left[\overline{Y}^o\big|A=a,  M, N, \bC\right]-E\left[\overline{Y}^o\big|A=a, N, \bC\right]\right\} \\
    &\quad + E\left[\overline{Y}^o\big|A=a, N, \bC\right] - E\left[\overline{Y}^o |A=a\right],
\end{align*}
Lemma~\ref{lemma2} implies that $E\left[\overline{Y}^o\big|A=a,  M, N, \bC\right]=E\left[\overline{Y}^o\big|A=a, N, \bC\right]$. Additionally, from Assumption 2, we find $pr(A=a,N,C) = pr(A=a)pr(N,C)$, which leads to
\begin{align*}
\frac{pr(M|A=a,N,\bC)}{pr(M|N,\bC)}
=
\frac{pr(A=a,M,N,C)}{pr(M,N,C)} \frac{pr(N,C)}{pr(A=a,N,C)}
=
\frac{pr(A=a|M,N,C)}{\pi^a(1-\pi)^{1-a}} .
\end{align*}
Combining the established results, we get the desired formula of $EIF_a$. 

For estimand $\frac{E[\sum_{j=1}^{N} Y_{\cdot j}(a)]}{E[N]}$, Lemma~\ref{lemma2} implies
\begin{align*}
    &\frac{E[\sum_{j=1}^{N} Y_{\cdot j}(a)]}{E[N]} \\
    &= \frac{1}{E[N]} E[E[E[N \overline{Y}^o(a)|N,C]|N]] \\
    &=\frac{1}{E[N]} E\left[E\left[E\left[E\left[E\left[N\overline{Y}^o|A=a,\bfX^o, M, N,\bC\right]|M,N,\bC\right]|A=a,N,\bC\right]|N\right]\right].
\end{align*}
Define $\Psi(\mathcal{P}) = \frac{\Psi_1(\mathcal{P})}{\Psi_2(\mathcal{P})}$, where $\Psi_1(\mathcal{P}) = E[\sum_{j=1}^{N} Y_{\cdot j}(a)]$ and $\Psi_2(\mathcal{P}) = E[N]$, we have
\begin{align}\label{eq: chain-rule}
    \frac{d \Psi(\mathcal{P}_t)}{d t}\bigg|_{t=0} = \frac{1}{\Psi_2(\mathcal{P})} \frac{d \Psi_1(\mathcal{P}_t)}{d t}\bigg|_{t=0} - \frac{\Psi_1(\mathcal{P})}{\Psi_2(\mathcal{P})^2} \frac{d \Psi_2(\mathcal{P}_t)}{d t}\bigg|_{t=0},
\end{align}
where $\frac{d \Psi_1(\mathcal{P}_t)}{d t}\bigg|_{t=0}$ is the same as $EIF_a$ in Equation (1) except that $\overline{Y}^o$ is substituted by $N\overline{Y}^o$, and $\frac{d \Psi_2(\mathcal{P}_t)}{d t}\bigg|_{t=0} = N - E[N]$ by straightforward calculation. Hence, the influence function for $\frac{E[\sum_{j=1}^{N} Y_{\cdot j}(a)]}{E[N]}$ is
\begin{align*}
   & \frac{1}{E[N]}\left\{N\ EIF_a + N E[\overline{Y}^o|A=a]- E[N\overline{Y}^o|A=a]\right\} - \frac{E[N\overline{Y}^o|A=a]}{E[N]^2}(N-E[N]) \\
   & =    \frac{N}{E[N]} \left\{EIF_a + E[\overline{Y}^o|A=a] - \frac{E[N\overline{Y}^o|A=a]}{E[N]}\right\}.
\end{align*}

\end{proof}

\subsection{Proof of Theorem~4}
\begin{proof}
Denote $\mathbb{P}_m X = m^{-1} \sum_{i=1}^m X_i$ and $\mathbb{G}_m X = m^{1/2}\{\mathbb{P}_m  X - E[X]\}$ for any i.i.d. samples $X_1,\ldots, X_m$ from a distribution on $X$. Below we give the proof for estimating ${\Delta}_C$; the results for ${\Delta}_I$ can be obtained in a similar way.

We first prove the desired results for $\widehat{\Delta}^{\textrm{Eff}}_C$ based on parametric working models.
Define
\begin{align*}
    U_a\{\bO;\btheta_a, \mu_C(a)\} &= \frac{I\{A=a\}}{\pi^a(1-\pi)^{1-a}}\left\{\overline{Y}^o - \eta_a(\bfX^o, M, N, \bC;\btheta_{\eta,a}) \right\} \\
    &\quad + \frac{\kappa_a(M,N,\bC;\btheta_{\kappa, a})}{\pi^a(1-\pi)^{1-a}}\left\{ \eta_a(\bfX^o, M, N, \bC;\btheta_{ \eta, a}) - \zeta_a(N, \bC;\btheta_{ \eta,a})\right\} \\
    &\quad + \zeta_a(N, \bC;\btheta_{ \eta,a}) - \mu_C(a).
\end{align*}
Then we have $\mathbb{P}_m U_a\{\bO;\widehat{\btheta}_a, \widehat\mu_C^{\textrm{Eff}}(a)\} = 0$ for $a =0,1$. By definition, $\widehat{\btheta} = (\widehat{\btheta}_1, \widehat{\btheta}_0)$ is computed by solving estimating equations, $\mathbb{P}_m \bpsi(\bO; \widehat{\btheta}) = \bzero$ for a known function $\bpsi$. By Lemma~\ref{lemma: asymptotics-theta}, $m^{1/2}(\widehat{\btheta}-\underline{\btheta}) \xrightarrow{d} N(0,\bfV)$ for some matrix $\bfV$. 
Define the estimating equations as 
\begin{equation*}
    \widetilde{\bpsi}(\btheta, \bmu_C)= \left(\begin{array}{c}
        U_1\{\bO;\btheta_1, \mu_C(1)\}  \\
        U_0\{\bO;\btheta_0, \mu_C(0)\} 
    \end{array}\right),
\end{equation*}
where $\bmu_C = (\mu_C(1), \mu_C(0))$. By regularity condition 5, $U_1$ and $U_0$ are continuous in parameters and dominated by an integrable function, and hence P-Glivenko Cantelli by Example 19.8 of \cite{vaart_1998}. Then, Theorem 5.9 of \cite{vaart_1998} implies that $(\widehat\btheta_a, \widehat\bmu_C^{\textrm{Eff}}) \xrightarrow{P} (\underline\btheta_a, \underline\bmu_C^{\textrm{Eff}})$, where $\underline{\mu}_C^{\textrm{Eff}}(a)$ satisfies $E[ U_a\{\bO;\underline\btheta_a, \underline\mu_C^{\textrm{Eff}}(a)\} ] = 0$. To see when $\underline\mu_C^{\textrm{Eff}}(a) = \mu_C(a)$, we have, using Lemma~\ref{lemma2} and $A \perp (N,\bC)$,
\begin{align*}
    E[ U_a\{\bO;\underline\btheta_a, \underline\mu_C(a)\} ] &= \mu_C(a) - E\left[\frac{I\{A=a\} - \kappa_a(\underline\btheta_{\kappa,a})}{\pi^a(1-\pi)^{1-a}} \eta_a(\underline\btheta_{\eta,a})\right] \\
    &\quad  - E\left[\left\{\frac{ \kappa_a(\underline\btheta_{\kappa,a})}{\pi^a(1-\pi)^{1-a}}-1\right\} \zeta_a(\underline\btheta_{\zeta,a})\right] - \underline\mu_C^{\textrm{Eff}}(a) \\
    &= \mu_C(a) -E\left[\frac{\kappa_a^* - \kappa_a(\underline\btheta_{\kappa,a})}{\pi^a(1-\pi)^{1-a}} \left\{ \eta_a(\underline\btheta_{\eta,a}) -\zeta_a(\underline\btheta_{\zeta,a}) \right\}\right]- \underline\mu_C^{\textrm{Eff}}(a).
\end{align*}
Given the condition that (i) $\kappa_a({\underline\btheta}_{\kappa,a}) =\kappa_a^*$ or (ii) $E[\eta_a({\underline\btheta}_{\eta,a})|M,N,\bC] = \zeta_a({\underline\btheta}_{\zeta,a})$, we obtain $\underline\mu_C^{\textrm{Eff}}(a) = \mu_C(a)$. 
To briefly check the triple-robustness property of $\widehat{\bmu}_C^{\textrm{Eff}}$, if suffices to discuss the case when $\widehat{\eta}_a$ and $\widehat{\zeta}_a$ are consistent whereas $\widehat{\kappa}_a$ may be inconsistent. In fact, this case satisfies the second condition (ii) because $E[\eta_a({\underline\btheta}_{\eta,a})|M,N,\bC] = E[\eta_a^*|M,N,\bC] = \zeta^* = \zeta_a({\underline\btheta}_{\zeta,a})$. Therefore, if at least two nuisance functions are consistently estimated, $\widehat{\bmu}_C^{\textrm{Eff}}$ is consistent.

Given regularity conditions 1-5, we apply Theorem 5.31 of \cite{vaart_1998} and get 
\begin{align*}
    m^{1/2}(\widehat{\bmu}_C^{\textrm{Eff}} - \bmu_C) = m^{1/2}E\left[\widetilde{\bpsi}(\widehat\btheta, \underline\bmu_C^{\textrm{Eff}})\right] + \mathbb{G}_m\widetilde{\bpsi}\left(\underline\btheta, \underline\bmu_C^{\textrm{Eff}}\right) + o_p\left(1+m^{1/2}||E[\widetilde{\bpsi}(\widehat\btheta, \underline\bmu_C^{\textrm{Eff}})]||_2\right).
\end{align*}
By the continuity of $\widetilde{\bpsi}$ on $\btheta$ and the asymptotic normality of $\widehat\btheta$, the delta method implies $m^{1/2}E[\widetilde{\bpsi}(\widehat\btheta, \underline\bmu_C^{\textrm{Eff}})] = \mathbb{G}_m u(\underline\btheta, \underline\bmu_C^{\textrm{Eff}}) + o_p(1)$ for some function $u$ and $m^{1/2}||E[\widetilde{\bpsi}(\widehat\btheta, \underline\bmu_C^{\textrm{Eff}})]||_2 = O_p(1)$, leading to the asymptotic normality of $\widehat{\bmu}_C^{\textrm{Eff}}$. 
Specifically, letting $z(\btheta) = E[\widetilde{\bpsi}(\btheta, \underline\bmu_C^{\textrm{Eff}})]$, then $u = \{\nabla z(\btheta)|_{\btheta=\underline{\btheta}} \}^\top E[\frac{\partial}{\partial \btheta} \bpsi(O,\btheta)|_{\btheta=\underline{\btheta}}]^{-1} \bpsi(O,\underline{\btheta})$. 
Then we get the desired result for $\widehat{\Delta}_C^{\textrm{Eff}}$ by delta method. The asymptotic variance of $\widehat{\Delta}_C^{\textrm{Eff}}$ is $$\nabla f^\top E\left[\left\{u(\underline\btheta, \underline\bmu_C^{\textrm{Eff}})+\widetilde{\bpsi}(\underline\btheta, \underline\bmu_C^{\textrm{Eff}})\right\}\left\{u(\underline\btheta, \underline\bmu_C^{\textrm{Eff}})+\widetilde{\bpsi}(\underline\btheta, \underline\bmu_C^{\textrm{Eff}})\right\}^\top\right] \nabla f^\top.$$
The consistency of the variance estimators is implied by Lemma 3 and regularity conditions 2 and 5. Then Slutsky's Theorem implies the desired convergence result.  

We next prove the desired results for $\widehat{\Delta}^{\textrm{Eff}}_C$ based on data-adaptive estimation with cross-fitting. Let $\mathcal{O}_{k}, k=1,\ldots, K$ be the split observed data and $\mathcal{O}_{-k} = \bigcup_{k'\in\{1,\ldots,K\}\setminus\{k\}}\mathcal{O}_{k'}$ be the training data for $\mathcal{O}_k$. Without loss of generalization, we assume that each $\mathcal{O}_k$ has the same sample size. Let $h_a = (\eta_a, \kappa_a, \zeta_a)$ denote the nuisance functions, $\widehat{h}_{a, k}= (\widehat\eta_{a, k}, \widehat\kappa_{a, k}, \widehat\zeta_{a, k})$ denote data-adaptive estimation trained on $\mathcal{O}_{-k}$, and $h_a^* = (\eta_a^*, \kappa_a^*, \zeta_a^*)$ denote the true nuisance functions.

Denote $\mathbb{P}_k X = | \mathcal{O}_k |^{-1} \sum_{i \in \mathcal{O}_k} X_{i}$ and $\mathbb{G}_{k} = | \mathcal{O}_k |^{1/2} (\mathbb{P}_{k} X -  E[X_{i}])$. Finally, we define
\begin{align*}
    D_a(h_a) &= \frac{I\{A=a\}}{\pi^a(1-\pi)^{1-a}}\left\{\overline{Y}^o - \eta_a(\bfX^o, M, N, \bC) \right\} \\
    &\quad + \frac{\kappa_a(M,N,\bC)}{\pi^a(1-\pi)^{1-a}}\left\{ \eta_a(\bfX^o, M, N, \bC) - \zeta_a(N, \bC)\right\} +\zeta_a(N, \bC) ,
\end{align*}
and $\bD(h) = (D_1(h_1), D_0(h_0))$.
Given the above definitions, we have $\widehat{\bmu}_C^{\textrm{Eff}} =K^{-1} \sum_{k=1}^K \mathbb{P}_{k} \bD(\widehat{h}_{k})$ and $\bmu_C = E[\bD(h^*)]$ with $\widehat{h}_{k} = (\widehat{h}_{1,k}, \widehat{h}_{0,k})$ and $h^* = (h_1^*, h_0^*)$. Then,
\begin{align}\label{eq: data-adaptive-estimation}
    &
    m^{1/2}(\widehat{\bmu}_C^{\textrm{Eff}} - \bmu_C) 
    \nonumber
    \\
    &
    =
    K^{-1/2}
    \sum_{k=1}^K
    \Big[
      \mathbb{G}_k\bD(h^*)
      +
      \mathbb{G}_{k}\{\bD(\widehat{h}_k)-\bD(h^*)\}
      +
      \left(\frac{m}{K}\right)^{1/2}
      E[\bD(\widehat{h}_k)-\bD(h^*)| \mathcal{O}_{-k}] 
     \Big]
      .
\end{align}
The first term $K^{-1/2} \sum_{k=1}^K  \mathbb{G}_k\bD(h^*)$ provides the asymptotic normality, so it suffices to show that the latter two terms are $o_P(1)$. Specifically, we denote $R_1 = \mathbb{G}_{k}\{\bD(\widehat{h}_k)-\bD(h^*)\}$ and $R_2 = \left(\frac{m}{K}\right)^{1/2}E[\bD(\widehat{h}_k)-\bD(h^*) | \mathcal{O}_{-k} ]$, and we show that $R_1 = o_p(1)$ and $R_2 = o_p(1)$ in the rest of the proof.


For $R_1$, 
by Theorem 2.14.2 of \cite{vaart&wellner1996weak}, we get
\begin{align*}
      E\left[\big|\big|{\mathbb{G}}_{k}\{\bD(\widehat{h}_k)-\bD(h^*)\}\big|\big|\bigg|\mathcal{O}_{-k}\right] \le c E\left[\big|\big|\bD(\widehat{h}_k)-\bD(h^*)\big|\big|^2\bigg|\mathcal{O}_{-k}\right]^{1/2}
\end{align*}
for some constant $c$, where $||\cdot||$ is the $L_2$ vector norm. By the assumption that $||\widehat{h}_k - h^*||_2 \rightarrow 0$ and Markov's inequality, we have $E\left[\big|\big|\widehat{h}_k-h^*\big|\big|^2\bigg|\mathcal{O}_{-k}\right] \xrightarrow{P} 0$. Since $\widehat{h}_k$ is fixed conditioning on $\mathcal{O}_{-k}$, we have, for $a = 0,1$
\begin{align}					\label{eq-empiricalprocessbound}
   & E\left[\{\bD_a(\widehat{h}_k)-\bD_a(h^*)\}^2\bigg|\mathcal{O}_{-k}\right] \\
   \nonumber
   &= \frac{1}{\pi_a}E[\kappa_a^*(1-\kappa_a^*)(\eta_a^*-\widehat{\eta}_{a,k})^2|\mathcal{O}_{-k}] + E\left[\left\{\frac{1}{\pi_a}(\widehat{\kappa}_{a,k} - \kappa_a^*)(\widehat{\eta}_{a,k} - \widehat{\zeta}_{a,k}) + (\frac{\kappa^*_a}{\pi_a}-1)(\widehat{\zeta}_{a,k} - \zeta_a^*)\right\}^2\bigg| \mathcal{O}_{-k}\right] \\
   \nonumber
   &\le \frac{1}{\pi_a} E[\kappa_a^*(1-\kappa_a^*)(\eta_a^*-\widehat{\eta}_{a,k})^2|\mathcal{O}_{-k}] + \frac{2}{\pi_a^2}E[(\widehat{\kappa}_{a,k} - \kappa_a^*)^2(\widehat{\eta}_{a,k} - \widehat{\zeta}_{a,k})^2|\mathcal{O}_{-k}]\\
   \nonumber
   &\quad  + \frac{2}{\pi_a^2}E[(\kappa^*_a-{\pi_a})^2(\widehat{\zeta}_{a,k} - \zeta_a^*)^2|\mathcal{O}_{-k}] \\
   \nonumber
   &\le\frac{1}{\pi_a} E[(\eta_a^*-\widehat{\eta}_{a,k})^2|\mathcal{O}_{-k}] + \frac{2}{\pi_a^2}E[(\widehat{\kappa}_{a,k} - \kappa_a^*)^2(\widehat{\eta}_{a,k} - \widehat{\zeta}_{a,k})^2|\mathcal{O}_{-k}] + \frac{2}{\pi_a^2}E[(\widehat{\zeta}_{a,k} - \zeta_a^*)^2|\mathcal{O}_{-k}]\\
   \nonumber
   &\le \frac{1}{\pi_a} E[(\eta_a^*-\widehat{\eta}_{a,k})^2|\mathcal{O}_{-k}] + \frac{6}{\pi_a^2}E[(\widehat{\kappa}_{a,k} - \kappa_a^*)^2(\widehat{\eta}_{a,k} - {\eta}_{a}^*)^2|\mathcal{O}_{-k}] + \frac{6}{\pi_a^2}E[(\widehat{\kappa}_{a,k} - \kappa_a^*)^2(\widehat{\zeta}_{a,k} - {\zeta}_{a}^*)^2|\mathcal{O}_{-k}]\\
   \nonumber
   &\quad + \frac{6}{\pi_a^2}E[(\widehat{\kappa}_{a,k} - \kappa_a^*)^2({\eta}_{a}^* - {\zeta}_{a}^*)^2|\mathcal{O}_{-k}] + \frac{2}{\pi_a^2}E[(\widehat{\zeta}_{a,k} - \zeta_a^*)^2|\mathcal{O}_{-k}]\\
   \nonumber
   &= o_p(1) + O_p(1) o_p(1) +  O_p(1) o_p(1) + o_p(1) + o_p(1)
   \\
   \nonumber
   &= o_p(1),
\end{align}
where $\pi_a = \pi^a(1-\pi)^{(1-a)}$. In the above derivation, the first line results from algebra and $pr(A=a|M,N,\bC) = \kappa_a^*$, the second line uses the Cauchy-Schwarz inequality, the third line is implied by $\kappa_a^* \in [0,1]$, the fourth line again comes from the Cauchy-Schwarz inequality, and the fifth line results from the assumption that $E\left[||\widehat{h}_k-h^*||^2\big|\mathcal{O}_{-k}\right] = o_p(1)$ and regularity condition 6. In particular, $E[(\widehat{\kappa}_{a,k} - \kappa_a^*)^2({\eta}_{a}^* - {\zeta}_{a}^*)^2|\mathcal{O}_{-k}] = o_p(1)$ is obtained as follows:
\begin{align*}
	E[(\widehat{\kappa}_{a,k} - \kappa_a^*)^2({\eta}_{a}^* - {\zeta}_{a}^*)^2|\mathcal{O}_{-k}]
	&
	=
	E \Big[
	(\widehat{\kappa}_{a,k} - \kappa_a^*)^2
	E \big\{ ({\eta}_{a}^* - {\zeta}_{a}^*)^2 \, \big| \,	
	M,N,C  \big\} \, \Big| \,
	\mathcal{O}_{-k} \Big]
	\\
	&
	=
	E \Big[
	(\widehat{\kappa}_{a,k} - \kappa_a^*)^2
	E \big[ \big\{ {\eta}_{a}^* - E(\eta_a^* | M,N,C) \big\}^2 \, \big| \,	
	M,N,C \big] \, \Big| \,
	\mathcal{O}_{-k} \Big]
	\\
	&
	=
	E \big[
	(\widehat{\kappa}_{a,k} - \kappa_a^*)^2
	\text{var}\big( \eta_{a}^*  \, \big| \,	
	M,N,C  \big) \, \big| \,
	\mathcal{O}_{-k} \big]
	\\
	&
	\leq 
	E \big[
	(\widehat{\kappa}_{a,k} - \kappa_a^*)^2
	E\big\{ ( \eta_{a}^*)^2  \, \big| \,	
	M,N,C  \big\} \, \big| \,
	\mathcal{O}_{-k} \big]
	\\
	&
=
	O_p(1)E \big[
	(\widehat{\kappa}_{a,k} - \kappa_a^*)^2 \, \big| \,
	\mathcal{O}_{-k} \big]\\
	&= o_p(1) \ .
\end{align*}
The second equality is from $E[\eta_a^* |M,N, \bC]= E[\overline{Y}^o|M,N, \bC] =\zeta_a^*$. The third equality is from the definition of the conditional variance, and the first inequality is based on $\text{var}(X|Y) \leq E(X^2|Y)$. The upper bound in the last line is from regularity condition 6, and the asymptotic rate is from the assumption. Therefore, ${\mathbb{G}}_{k}\{\bD(\widehat{h}_k)-\bD(h^*)\} = o_p(1)$ given $\mathcal{O}_{-k}$, which implies $R_1 = o_p(1)$.

For $R_2$, we can compute that, for each $k$,
\begin{align*}
    &
    \left(\frac{m}{K}\right)^{1/2}E[\bD(\widehat{h}_k)-\bD(h^*) | \mathcal{O}_{-k} ]
    \\
    &= \frac{ 1}{\pi_a}\left(\frac{m}{K}\right)^{1/2}E[(\widehat{\kappa}_{a,k} - \kappa_a^*)(\widehat{\eta}_{a,k} - \widehat{\zeta}_{a,k}) | \mathcal{O}_{-k} ] \\
    &= \frac{ 1}{\pi_a}\left(\frac{m}{K}\right)^{1/2} E[(\widehat{\kappa}_{a,k} - \kappa_a^*)(\widehat{\eta}_{a,k} - \eta_a^*) | \mathcal{O}_{-k} ] - 
    \frac{ 1}{\pi_a}\left(\frac{m}{K}\right)^{1/2} E[(\widehat{\kappa}_{a,k} - \kappa_a^*)(\widehat{\zeta}_{a,k} - \zeta_a^*) | \mathcal{O}_{-k} ] \\
    &\le 
    \frac{ 1}{\pi_a}\left(\frac{m}{K}\right)^{1/2}
    E[(\widehat{\kappa}_{a,k} - \kappa_a^*)^2 | \mathcal{O}_{-k}]^{1/2}\{E[(\widehat{\eta}_{a,k} - \eta_a^*)^2 | \mathcal{O}_{-k}]^{1/2}
    + E[(\widehat{\zeta}_{a,k} - \zeta_a^*)^2 | \mathcal{O}_{-k}]^{1/2}\} \\
    &= o_p(1),
\end{align*}
where the first line comes from algebra and Lemma~\ref{lemma2}, the second line is implied by $E[\eta_a^* |M,N, \bC]= E[\overline{Y}^o|M,N, \bC] =\zeta_a^*$, yielding that $E[(\widehat{\kappa}_{a,k} - \kappa_a^*)(\eta_a^*-\zeta_a^*) | \mathcal{O}_{-k}] = E[E[\widehat{\kappa}_{a,k} - \kappa_a^*|\mathcal{O}_{-k}]E[\eta_a^*-\zeta_a^*|M,N,\bC]] = 0$, the third line results from the Cauchy-Schwartz inequality, and the last line uses the assumption that nuisance functions are estimated at $m^{{1}/{4}}$-rate.

Finally, the Equation~(\ref{eq: data-adaptive-estimation}) becomes $ m^{1/2}(\widehat{\bmu}_C^{\textrm{Eff}} - \bmu_C) 
= K^{-1/2}
    \sum_{k=1}^K 
      \mathbb{G}_k\bD(h^*)
      =
\mathbb{G}_m\bD(h^*) + o_p(1)$. Since $D_a(h_a^*) = EIF_{C,a} + \mu_C(a)$, we get $ m^{1/2}(\widehat{\bmu}_C^{\textrm{Eff}} - \bmu_C) = \mathbb{G}_m EIF_{C} + o_p(1)$, which indicates the asymptotic linearity. Then the asymptotic normality is implied by the Central Limit Theorem if we can show $EIF_C$ has finite second moments. This is implied by Assumption 1 and regularity condition 6.

For completeness, we show the consistency of the variance estimate of cross-fitting estimator. Let us denote $\widehat{\Sigma}_k = \big| \mathcal{O}_k \big|^{-1} \sum_{i \in \mathcal{O}_k} \big\{ D(\widehat{h}_k) -  \widehat{\bmu}_{k,C}^{\textrm{Eff}} \big\}\big\{ D(\widehat{h}_k) -  \widehat{\bmu}_{k,C}^{\textrm{Eff}} \big\}^\top $ where $\widehat{\bmu}_{k,C}^{\textrm{Eff}} = \mathbb{P}_k \{ D(\widehat{h}_k) \}$. The proposed variance estimate for $\widehat{\bmu}_{C}^{\textrm{Eff}}$ is $ \widehat{\Sigma} = 
K^{-1} \sum_{k=1}^{K} \widehat{\Sigma}_k$. Therefore, it suffices to show that $\widehat{\Sigma}_k$ is consistent for $ \Sigma=\text{var} \{ D (h^*) \} $. Let us consider the following decomposition of $ \widehat{\Sigma}_k - \Sigma $:
 	\begin{align*}
 		&
 		 \widehat{\Sigma}_k - \Sigma
 		=
 		\underbrace{ 		 
 		\widehat{\Sigma}_k
 		-
 		\overline{\Sigma}_k
 		}_{(A)}
 		+
 		\underbrace{
 		\overline{\Sigma}_k
 		-
 	\Sigma
 		 }_{(B)} \ , \\
 		&
 		\overline{\Sigma}_k
 		=
 		\frac{1}{\big| \mathcal{O}_k \big| }
 		\sum_{i \in \mathcal{O}_k}  
 		\big\{ D_i ( h^* ) - \mu_C \big\}	\big\{ D_i ( h^* ) - \mu_C \big\}^\top
 	\end{align*}
 	The second term $(B)$ is then $o_p(1)$ from the law of large numbers. 
 	Denoting $\Delta_i = D_i ( \widehat{h}_k ) -  \widehat{\bmu}_{k,C}^{\textrm{Eff}} -D_i ( h^* ) + \mu_C$, the first term $(A)$ is represented as
 	\begin{align*}
 		(A)
 		& =
 		\frac{1}{| \mathcal{O}_k | }
 		\sum_{i \in \mathcal{O}_k}
 		\bigg[ \Delta_i\Delta_i^\top + \Delta_i \big\{ D_i ( h^* ) - \mu_C \big\}^\top + \big\{ D_i ( h^* ) - \mu_C \big\}\Delta_i^\top
 		\bigg].
 	\end{align*}
 From the H\"older's inequality and the matrix norm, we find
 	\begin{align*}
 		|| (A) ||^2
 		&
 		\leq 
 		\bigg[ \frac{1}{|\mathcal{O}_k|} \sum_{i \in \mathcal{O}_k} \Delta_i^\top\Delta_i\bigg]^2 \\
 		&\quad 		+
  		2\bigg[ \frac{1}{|\mathcal{O}_k|} \sum_{i \in \mathcal{O}_k} \Delta_i^\top\Delta_i\bigg] 
 		\bigg[ \frac{1}{|\mathcal{O}_k|} \sum_{i \in \mathcal{O}_k} \big\{ D_i ( h^* ) - \mu_C \big\}^\top\big\{ D_i ( h^* ) - \mu_C \big\} \bigg]
 		\\
 		&
 		= 
 		\bigg[ \frac{1}{|\mathcal{O}_k|} \sum_{i \in \mathcal{O}_k} \Delta_i^\top\Delta_i\bigg]^2	+
  		2\bigg[ \frac{1}{|\mathcal{O}_k|} \sum_{i \in \mathcal{O}_k} \Delta_i^\top\Delta_i\bigg] \left\{\overline{\Sigma}_k(1,1) + \overline{\Sigma}_k(2,2)\right\}
 	\end{align*}
 	Since $\overline{\Sigma}_k = \Sigma + o_p(1)$ and $\Sigma$ is finite, we have  $\overline{\Sigma}_k(1,1) + \overline{\Sigma}_k(2,2) = O_p(1)$. Therefore, if $|\mathcal{O}_k|^{-1} \sum_{i \in \mathcal{O}_k} \Delta_i^\top\Delta_i$ is $o_p(1)$, we establish $(A)=o_p(1)$. We expand $|\mathcal{O}_k|^{-1} \sum_{i \in \mathcal{O}_k} \Delta_i^\top\Delta_i$ as follows:
 	\begin{align*}
 		\frac{1}{|\mathcal{O}_k|} \sum_{i \in \mathcal{O}_k} \Delta_i^\top\Delta_i
 	& =
 	\frac{1}{|\mathcal{O}_k|} \sum_{i \in \mathcal{O}_k} 
 	\Big|\Big|\big\{	D_i (\widehat{h}_k) - D_i(h^*)\big\}-\big\{ \widehat{\bmu}_{k,C}^{\textrm{Eff}} - \mu_C(a)\big\} \Big|\Big|^2
 	\\
 	&
 	\leq
 	\frac{2}{|\mathcal{O}_k|} 
 	\sum_{i \in \mathcal{O}_k} \big|\big|
 			D_i (\widehat{h}_k) - D_i(h^*)
 		\big|\big|^2
 		+
 		2
 \big|\big|
 			 \widehat{\bmu}_{k,C}^{\textrm{Eff}} - \mu_C(a)
\big|\big|^2
 		\\
 		&
 		=
 		2
 		E \Big[
 \big|\big|
 			D_i (\widehat{h}_k) - D_i(h^*)
\big|\big|^2
 		\, \Big| \, \mathcal{O}_{-k} \Big]
 		+ o_p(1)
 		+
 		2
\big|\big|
 			 \widehat{\bmu}_{k,C}^{\textrm{Eff}} - \mu_C(a)
\big|\big|^2
 		\\
 		& = o_p(1) \ .
 	\end{align*}
The inequality is from $(a-b)^2 \leq 2a^2 + 2b^2$. The third line is from applying the law of large numbers on $\big|\big|
 			D_i (\widehat{h}_k) - D_i(h^*)
 		\big|\big|^2$. The last line is from \eqref{eq-empiricalprocessbound} and the consistency result of $\widehat{\bmu}_{k,C}^{\textrm{Eff}}$, i.e., $\widehat{\bmu}_{k,C}^{\textrm{Eff}} - \mu_C(a) = o_p(1)$. Then Slutsky's Theorem implies the desired convergence result.

\end{proof}

\subsection{Results under stratified or biased-coin cluster randomization}
Under stratified or biased-coin cluster randomization, Theorem 1 of \cite{wang2021model} showed that the influence function, consistency, and asymptotic normality of an M-estimator is not changed compared to simple randomization under the same regularity conditions as we have in Section~\ref{suppsec: reg-condi}. 
Since our weighted g-computation estimators and proposed estimator with parametric nuisance models are all M-estimators, we obtain the consistency and asymptotic normality of these estimators under these two restricted randomization schemes.
Furthermore, the asymptotic variance under stratified or biased-coin cluster randomization is reduced by $\frac{1}{\pi(1-\pi)} E[E\{(A-\pi)IF(O;\underline{\btheta})|Z\}^2]$, where $IF(O;\underline{\btheta})$ is the influence function for $\widehat{\btheta}$, $\underline{\btheta}$ is the probability limit of $\widehat{\btheta}$, and $Z$ is a vector of dummy variables encoded by balanced strata covariates. For our setting with GEE or linear mixed models, when $\pi=0.5$ and $Z$ are adjusted for as covariates, we can follow the same proof as in Corollary 1 of \cite{wang2021model} to show that this variance difference is zero, yielding the consistency of our sandwich variance estimators.

For our proposed estimator with machine learning algorithms, Equation~(\ref{eq: data-adaptive-estimation}) still holds under stratified or biased-coin cluster randomization by direct algebra. For the remainder terms $ \mathbb{G}_{k}\{\bD(\widehat{h}_k)-\bD(h^*)\}$ and $
      \left(\frac{m}{K}\right)^{1/2}
      E[\bD(\widehat{h}_k)-\bD(h^*)| \mathcal{O}_{-k}] $ in Equation~(\ref{eq: data-adaptive-estimation}), we can follow the same proof to show that they are $o_p(1)$ under stratified or biased-coin cluster randomization given the $m^{1/4}$ convergence rate for nuisance function estimators. This is because the relevant proof only involves the convergence of nuisance function estimators but not the inter-cluster independence of observed data. Therefore, we get $ m^{1/2}(\widehat{\bmu}_C^{\textrm{Eff}} - \bmu_C) = \mathbb{G}_m EIF_{C} + o_p(1)$ and the asymptotic normality again follows from Theorem 1 of \cite{wang2021model}. 
      The variance difference between stratified versus simple randomization becomes $\frac{1}{\pi(1-\pi)} E[E\{(A-\pi)EIF_{C}(a)|Z\}^2]$ for $\widehat{\bmu}_C^{\textrm{Eff}}(a)$. However, Lemma 3 of \cite{wang2021model} and Lemma 2 implies that $E\{A\ EIF_{C}(a)|Z\} = \pi E\{ EIF_{C}(a)|Z\}$, indicating that the variance difference  is zero. 
\section{Other methods}\label{suppsec: other methods}
The augmented generalized estimating equations (Aug-GEE, \citealp{stephens2012augmented}) is an approach that incorporates the augmentation technique in semiparamteric theory to improve the robustness and precision of GEE. Instead of modeling the conditional expectation of $Y_{ij}$ on $A_i$ and $\bX_{ij}$, Aug-GEE solves the following estimating equations to estimate $(\beta_0, \beta_A)$:
\begin{equation}\label{eq:aug-GEE}
    \sum_{i=1}^m \bfD_i^\top \bfV_i^{-1} \{\bY_i^o - \bmu_i^o(A)\} - (A_i-\pi) \boldsymbol{\gamma}(N_i, \bC_i, \bfX_i^o) = \bzero,
\end{equation}
where the first part $\bfD_i^\top \bfV_i^{-1} \{\bY_i^o - \bmu_i^o(A)\}$ is the same as GEE except that covariates $(N,\bC, \bfX^o)$ are omitted, and the remaining part $-(A_i-\pi) \boldsymbol{\gamma}(N_i, \bC_i, \bfX_i^o)$ is the augmentation term with $\boldsymbol{\gamma}(N_i, \bC_i, \bfX_i^o)$ being an arbitrary $M_i$-dimensional function.
Popular choices of the augmentation term involves parametric working models for $\bY_i^o$ on $A_i, N_i, \bC_i, \bfX_i^o$ and can be found in \cite{stephens2012augmented}.
Given $(\widehat\beta_0, \widehat\beta_A)$ from Aug-GEE, we can follow the steps in Section 3 for GEE-g to construct Aug-GEE-g targeting $(\Delta_C, \Delta_I)$. Given similar conditions (Assumptions 1,2,4) in Theorem~1, we can get the consistency and asymptotically normality of Aug-GEE-g. Compared to GEE, the asymptotic results for Aug-GEE requires less assumptions:  a correctly-specified model for $E[Y_{ij}|\bU_{ij}]$ or (S2)-(S4) in Theorem 1 is not needed as long as the estimates for the augmentation term converge. The reason is that the variance fucntion $v(Y_{ij})$ is only a function $\mu(A_i)$ which does not vary within a cluster and hence satisfies (S4); however, this simplification of variance functions may impact the precision. In terms of precision, Aug-GEE has the potential to improve precision over GEE, while simulation studies showed that this efficiency gain is at most moderate, especially for studies with few clusters \citep{stephens2012augmented, benitez2021comparative}. 

Targeted maximum likelihood estimation (TMLE, \citealp{van2006targeted}) as a general framework for causal inference also has applications for clustered data.
Among them, two most relevant methods are 
cluster-level TMLE and hierarchical TMLE proposed by \cite{balzer2019new}. In our setting, both methods assume $M_i = N_i$. The cluster-level TMLE targets $\Delta_C$ and works on cluster averages of outcomes and covariates, i.e., $(\overline{Y}_i^o, A_i, N_i, \bC_i, \overline{\bX}_i^o)$. Although individual-level information is not utilized, cluster-level TMLE avoids the complexity of accounting for intracluster correlations. Thus, the cluster-level TMLE enjoys the precision gain from covariate adjustment without the need of a correctly specified mean model. 
However, under the potentially informative within-cluster sampling schemes as we consider in this paper, $\overline{\bX}_i^o$ is no longer independent of $A_i$, and directly adjusting for them may cause bias.
Hierarchical TMLE, also known as individual-level TMLE, performs TMLE on $E[Y_{ij}|\bU_{ij}]$ and aggregates model predictions to estimate $\Delta_I$. Unlike cluster-level TMLE, the asymptotic validity of hierarchical TMLE needs additional assumptions for convergence of TMLE given dependent data.
With the additional assumptions, hierarchical TMLE can also target $\Delta_C$ by modifying the cluster weights, and cluster-level TMLE can utilize the model predictions from the  hierarchical TMLE to further improve power.
Both TMLE-based methods have been demonstrated to be more precise than GEE and Aug-GEE in simulation studies  \citep{benitez2021comparative}, while less is known about their validity and efficiency when $M_i \ne N_i$ or the treatment effect is related to $N_i$ as we consider here. 

In addition to the above discussed methods, \cite{schochet2021design, su2021model} established the asymptotic theory for a class of linearly-adjusted estimators under the randomization inference framework, and \cite{bugni2022inference} proposed unadjusted estimators that accounts for treatment effects that vary by $N_i$. 
Under the super-population framework, their estimators can be viewed as special cases of GEE-g with identity link and independence working correlation structure (with $1/M_i$ weighting for the unadjusted estimator),  whose property has been discussed and hence omitted.

\section{Addition simulation studies}\label{suppsec: additional-sim}
\subsection{Simulations with increased source population size heterogeneity}
We replicated the two simulations in the main paper with the source population distribution $N_i \sim \textup{Uniform}\{10,50\}$ changed to $N_i \sim \textup{Uniform}[10,100]$.
For the first simulation study with continuous outcomes, we modify the data-generating distribution to accommodate this changes as follows:
\begin{align*}
    \mathcal{P}^{C_1|N} &= \mathcal{N}(N/20,4) \\
    \mathcal{P}^{C_2\mid N,C_1} &=\mathcal{B}[\textrm{expit}\{\log(N/20)C_1\}] \\
    X_{ij1} &\sim \mathcal{B}\left({N_i}/{100}\right) \\
    X_{ij2} &\sim \mathcal{N}\left\{{\sum_{j=1}^{N_i} X_{ij1}}(2C_{i2}-1)/{N_i}, 9\right\} \\
     Y_{ij}(1) &\sim \mathcal{N}\left\{ N_i/10 + N_i\sin(C_{i1})  (2C_{i2}-1)/60+ 5e^{X_{ij1}}\mid X_{ij2}\mid , 1\right\} \\
      Y_{ij}(0) &\sim \mathcal{N}\left\{\gamma_i + N_i\sin(C_{i1})  (2C_{i2}-1)/60+ 5e^{X_{ij1}}\mid X_{ij2}\mid , 1\right\}.
\end{align*}
We set $M_i(1) = M_i(0) = 9 + \mathcal{B}(0.5)$ for the random observed cluster size scenario, and $M_i(1) =  \textup{int}(N_i/10)+5C_{i2}$, $M_i(0) = 3 I\{N_i\ge 55\} + 3$ for the cluster-dependent observed cluster size scenario, where $ \textup{int}$ rounds a real number to an integer. 

For the second simulation study, the data were generated following the first simulation study, except that the potential outcomes were drawn from the following Bernoulli distributions:
\begin{align*}
    Y_{ij}(1) &\sim \mathcal{B}\left[\textrm{expit}\left\{ -N_i/40 + N_i\sin(C_{i1}) (2C_{i2}-1)/60+ 1.5e^{X_{ij1}}\mid X_{ij2}\mid ^{1/2}\right\}\right] \\
    Y_{ij}(0) &\sim \mathcal{B}\left[\textrm{expit}\left\{ \gamma_i+N_i\sin(C_{i1}) (2C_{i2}-1)/60+ 1.5(2X_{ij1}-1)\mid X_{ij2}\mid ^{1/2} \right\}\right].
\end{align*}
In summary, most of the changes are replacing $N_i$ by $N_i/2$ in the data-generating distribution.

The simulation results are provided in Tables~\ref{tab: sim1-2} and \ref{tab: sim2-2}, whose takeaways are similar to Tables 1 and 2 in the main paper.

\begin{table}[htbp]
\renewcommand{\arraystretch}{0.8}
\centering
\caption{Replication of the first simulation with increased population size variation. }\label{tab: sim1-2}
\resizebox{1\textwidth}{!}{
\begin{tabular}{lrrrrrrrrrr}
  \hline
  & & \multicolumn{4}{c}{\shortstack[c]{ Cluster-average treatment \\ effect $\Delta_C = 5.5$}  } & & \multicolumn{4}{c}{\shortstack[c]{ Individual-average treatment \\ effect $\Delta_I = 6.75$}}\\
Setting  & Method & Bias & ESE   &  ASE &  CP &\ & Bias & ESE   &  ASE &  CP\\ 
  \hline
\multirow{5}{*}{\shortstack[l]{Scenario 1: \\ Small $m$ with\\ 
 random \\
 observed cluster sizes}}&  Unadjusted &   0.03 & 3.78 & 3.74 & 0.95& & -0.04 & 4.12 & 3.87 & 0.93 \\ 
&    GEE-g & 0.10 & 2.59 & 2.29 & 0.92& & -0.09 & 3.11 & 2.62 & 0.91 \\
&    LMM-g & 0.10 & 2.59 & 2.65 & 0.96& & -0.09 & 3.11 & 2.71 & 0.93 \\
 & Eff-PM & 0.03 & 2.81 & 2.49 & 0.93& & -0.03 & 3.45 & 2.91 & 0.92 \\ 
  &   Eff-ML &0.04 & 2.06 & 1.81 & 0.94& & -0.07 & 2.33 & 2.10 & 0.94 \\ 
   \hline
\multirow{5}{*}{\shortstack[l]{Scenario 2: \\ Small $m$ with\\ 
 cluster-dependent \\
 observed cluster sizes}}&  Unadjusted & 0.01 & 4.34 & 4.32 & 0.95& & -0.07 & 4.44 & 4.23 & 0.94 \\ 
&    GEE-g  & 0.90 & 3.45 & 3.01 & 0.90& & 0.48 & 4.02 & 3.33 & 0.89 \\ 
&    LMM-g  & 0.77 & 3.31 & 3.41 & 0.96& & 0.38 & 3.78 & 3.47 & 0.94 \\ 
& Eff-PM & 0.04 & 3.77 & 3.54 & 0.94& & -0.01 & 4.06 & 3.70 & 0.93 \\ 
  &   Eff-ML & -0.02 & 3.56 & 3.32 & 0.94& & -0.01 & 4.02 & 3.61 & 0.93 \\ 
   \hline
\multirow{5}{*}{\shortstack[l]{Scenario 3: \\ Large $m$ with\\ 
 random \\
 observed cluster sizes}}&  Unadjusted &   -0.01 & 2.03 & 2.04 & 0.95& & -0.02 & 2.19 & 2.16 & 0.95 \\ 
&    GEE-g  &0.02 & 1.34 & 1.29 & 0.94& & -0.02 & 1.65 & 1.55 & 0.94 \\ 
&    LMM-g  & 0.02 & 1.34 & 1.35 & 0.95& & -0.02 & 1.65 & 1.42 & 0.91 \\ 
 & Eff-PM & 0.00 & 1.36 & 1.30 & 0.94& & -0.00 & 1.68 & 1.58 & 0.93 \\  
  &   Eff-ML & 0.02 & 0.61 & 0.62 & 0.95& & -0.06 & 0.68 & 0.70 & 0.96 \\ 
   \hline
\multirow{5}{*}{\shortstack[l]{Scenario 4: \\ Large $m$ with\\ 
 cluster-dependent \\
 observed cluster sizes}}&  Unadjusted & -0.01 & 2.35 & 2.35 & 0.95& & -0.02 & 2.38 & 2.36 & 0.95 \\ 
&    GEE-g  &  0.92 & 1.74 & 1.68 & 0.90& & 0.58 & 2.00 & 1.90 & 0.93 \\ 
&    LMM-g  &  0.87 & 1.73 & 1.75 & 0.93& & 0.55 & 1.98 & 1.82 & 0.92 \\ 
 & Eff-PM & 0.02 & 1.87 & 1.84 & 0.95& & 0.02 & 2.02 & 1.98 & 0.94 \\ 
  &   Eff-ML &-0.06 & 1.66 & 1.61 & 0.94& & -0.04 & 1.91 & 1.83 & 0.94 \\ 
   \hline
\end{tabular}
}
{
\raggedright \small
\setlength{\baselineskip}{1pt}
 Unadjusted: the unadjusted estimator. GEE-g: GEE with weighted g-computation. LMM-g: linear mixed models with weighted g-computation. Eff-PM: our proposed method with parametric working models. Eff-ML: our proposed method with machine learning algorithms. ESE: empirical standard error. ASE: average of estimated standard error. CP: coverage probability based on $t$-distribution.

}
\end{table}

\begin{table}[htbp]
\renewcommand{\arraystretch}{0.8}
\centering
\caption{Replication of the second simulation with increased population size variation.}\label{tab: sim2-2}
\resizebox{1\textwidth}{!}{
\begin{tabular}{lrrrrrrrrrr}
  \hline
  & & \multicolumn{4}{c}{\shortstack[c]{ Cluster-average treatment \\ effect $\Delta_C = 1.56$}  } & & \multicolumn{4}{c}{\shortstack[c]{ Individual-average treatment \\ effect $\Delta_I = 1.35$}}\\
Setting  & Method & Bias & ESE   &  ASE &  CP &\ & Bias & ESE   &  ASE &  CP\\ 
  \hline
\multirow{5}{*}{\shortstack[l]{Scenario 1: \\ Small $m$ with\\ 
 random \\
 observed cluster sizes}}&  Unadjusted &  0.02 & 0.21 & 0.20 & 0.94& & 0.03 & 0.17 & 0.16 & 0.93 \\ 
&    GEE-g & 0.00 & 0.19 & 0.17 & 0.92& & 0.03 & 0.17 & 0.14 & 0.92 \\ 
&    LMM-g & 0.00 & 0.18 & 0.17 & 0.92& & 0.03 & 0.16 & 0.14 & 0.92 \\
 & Eff-PM &0.00 & 0.19 & 0.17 & 0.93& & 0.01 & 0.17 & 0.15 & 0.93 \\ 
  &   Eff-ML &0.02 & 0.33 & 0.18 & 0.93& & 0.03 & 0.18 & 0.15 & 0.94 \\ 
   \hline
\multirow{5}{*}{\shortstack[l]{Scenario 2: \\ Small $m$ with\\ 
 cluster-dependent \\
 observed cluster sizes}}&  Unadjusted & 0.04 & 0.27 & 0.26 & 0.94& & 0.04 & 0.20 & 0.19 & 0.93 \\ 
&    GEE-g  & -0.08 & 0.30 & 0.20 & 0.83& & -0.03 & 0.20 & 0.18 & 0.90 \\
&    LMM-g  &  -0.07 & 0.21 & 0.19 & 0.85& & -0.02 & 0.18 & 0.17 & 0.91 \\ 
& Eff-PM & 0.03 & 0.28 & 0.25 & 0.93& & 0.03 & 0.20 & 0.19 & 0.94 \\ 
  &   Eff-ML & 0.04 & 0.35 & 0.27 & 0.94& & 0.04 & 0.27 & 0.20 & 0.94 \\ 
   \hline
\multirow{5}{*}{\shortstack[l]{Scenario 3: \\ Large $m$ with\\ 
 random \\
 observed cluster sizes}}&  Unadjusted &  0.01 & 0.11 & 0.11 & 0.95& & 0.01 & 0.09 & 0.09 & 0.95 \\ 
&    GEE-g  &0.00 & 0.09 & 0.09 & 0.94& & 0.01 & 0.08 & 0.07 & 0.94 \\ 
&    LMM-g  & 0.00 & 0.09 & 0.09 & 0.94 && 0.01 & 0.08 & 0.08 & 0.94 \\ 
 & Eff-PM & 0.01 & 0.09 & 0.09 & 0.94& & 0.01 & 0.08 & 0.08 & 0.95 \\ 
  &   Eff-ML & 0.01 & 0.10 & 0.09 & 0.94 && 0.01 & 0.08 & 0.08 & 0.95 \\ 
   \hline
\multirow{5}{*}{\shortstack[l]{Scenario 4: \\ Large $m$ with\\ 
 cluster-dependent \\
 observed cluster sizes}}&  Unadjusted & 0.01 & 0.14 & 0.14 & 0.95& & 0.01 & 0.10 & 0.10 & 0.95 \\ 
&    GEE-g  &  -0.11 & 0.10 & 0.10 & 0.74& & -0.05 & 0.09 & 0.10 & 0.90 \\ 
&    LMM-g  &  -0.09 & 0.10 & 0.10 & 0.79& & -0.04 & 0.09 & 0.09 & 0.91 \\
 & Eff-PM & 0.00 & 0.12 & 0.12 & 0.94& & 0.01 & 0.10 & 0.10 & 0.95 \\ 
  &   Eff-ML & 0.00 & 0.17 & 0.13 & 0.93& & 0.01 & 0.10 & 0.09 & 0.94 \\ 
   \hline
\end{tabular}
}
{
\raggedright \small
\setlength{\baselineskip}{1pt}
 Unadjusted: the unadjusted estimator. GEE-g: GEE with weighted g-computation. LMM-g: linear mixed models with weighted g-computation. Eff-PM: our proposed method with parametric working models. Eff-ML: our proposed method with machine learning algorithms. ESE: empirical standard error. ASE: average of estimated standard error. CP: coverage probability based on $t$-distribution.

}
\end{table}

\subsection{Simulations with increased sample size}

To further demonstrate our asymptotic results, we added a replication of the first two simulations with $m$ increased to 1000. The simulation results are summarized in Tables \ref{tab: sim3} and \ref{tab: sim4}. The results showed negligible bias, accurate variance estimators, and valid coverage probability for our proposed methods, which confirm our asymptotic theory.
\begin{table}[htbp]
\centering
\caption{Replication of the first simulation with $m=1000$. }\label{tab: sim3}
\renewcommand{\arraystretch}{0.8}
\resizebox{1\textwidth}{!}{
\begin{tabular}{lrrrrrrrrrr}
  \hline
  & & \multicolumn{4}{c}{\shortstack[c]{ Cluster-average treatment \\ effect $\Delta_C = 6$}  } & & \multicolumn{4}{c}{\shortstack[c]{ Individual-average treatment \\ effect $\Delta_I = 8.67$}}\\
Setting  & Method & Bias & ESE   &  ASE &  CP &\ & Bias & ESE   &  ASE &  CP\\ 
  \hline
\multirow{5}{*}{\shortstack[l]{Continuous outcomes \\ with random \\
 observed cluster sizes}}&  Unadjusted & -0.01 & 0.83 & 0.84 & 0.95& & -0.02 & 0.75 & 0.75 & 0.95 \\
&    GEE-g  & -0.00 & 0.44 & 0.44 & 0.95& & -0.01 & 0.60 & 0.59 & 0.94 \\ 
&    LMM-g  & -0.00 & 0.44 & 0.44 & 0.95& & -0.01 & 0.60 & 0.58 & 0.94 \\ 
& Eff-PM & -0.01 & 0.44 & 0.44 & 0.95& & -0.01 & 0.61 & 0.60 & 0.95 \\ 
  &   Eff-ML &  0.01 & 0.19 & 0.21 & 0.96& & 0.01 & 0.21 & 0.23 & 0.97 \\ 
   \hline
\multirow{5}{*}{\shortstack[l]{Continuous outcomes \\ with cluster-dependent \\
 observed cluster sizes}}&  Unadjusted &  0.01 & 0.91 & 0.91 & 0.95& & 0.00 & 0.78 & 0.79 & 0.95 \\ 
&    GEE-g  & 1.75 & 0.58 & 0.58 & 0.16& & 0.76 & 0.68 & 0.68 & 0.79 \\ 
&    LMM-g  & 1.74 & 0.58 & 0.58 & 0.15& & 0.76 & 0.68 & 0.62 & 0.74 \\ 
& Eff-PM & -0.01 & 0.57 & 0.58 & 0.95& & -0.00 & 0.68 & 0.68 & 0.95 \\ 
  &   Eff-ML & -0.01 & 0.59 & 0.59 & 0.95& & -0.00 & 0.68 & 0.68 & 0.95 \\ 
   \hline
\end{tabular}
}
\end{table}

\begin{table}[htbp]
\centering
\caption{Replication of the second simulation with $m=1000$.}\label{tab: sim4}
\renewcommand{\arraystretch}{0.8}
\resizebox{1\textwidth}{!}{
\begin{tabular}{lrrrrrrrrrr}
  \hline
  & & \multicolumn{4}{c}{\shortstack[c]{ Cluster-average treatment \\ effect $\Delta_C = 1.56$}  } & & \multicolumn{4}{c}{\shortstack[c]{ Individual-average treatment \\ effect $\Delta_I = 1.35$}}\\
Setting  & Method & Bias & ESE   &  ASE &  CP &\ & Bias & ESE   &  ASE &  CP\\ 
  \hline
\multirow{5}{*}{\shortstack[l]{Continuous outcomes \\ with random \\
 observed cluster sizes}}&  Unadjusted & 0.00 & 0.04 & 0.04 & 0.95& & 0.00 & 0.02 & 0.02 & 0.95 \\ 
&    GEE-g  & -0.00 & 0.03 & 0.03 & 0.95& & 0.00 & 0.02 & 0.02 & 0.95 \\ 
&    LMM-g  &  -0.00 & 0.03 & 0.03 & 0.95& & 0.00 & 0.02 & 0.02 & 0.95 \\ 
& Eff-PM & 0.00 & 0.03 & 0.03 & 0.95& & 0.00 & 0.02 & 0.02 & 0.97 \\ 
  &   Eff-ML & 0.00 & 0.03 & 0.03 & 0.95& & 0.00 & 0.02 & 0.02 & 0.98 \\
   \hline
\multirow{5}{*}{\shortstack[l]{Continuous outcomes \\ with cluster-dependent \\
 observed cluster sizes}}&  Unadjusted & 0.00 & 0.05 & 0.05 & 0.95& & 0.00 & 0.02 & 0.02 & 0.95 \\ 
&    GEE-g  & -0.20 & 0.03 & 0.03 & 0.00& & -0.06 & 0.02 & 0.03 & 0.34 \\ 
&    LMM-g  & -0.14 & 0.03 & 0.03 & 0.02& & -0.05 & 0.02 & 0.03 & 0.46 \\ 
& Eff-PM & 0.00 & 0.04 & 0.04 & 0.95& & 0.00 & 0.02 & 0.03 & 0.97 \\  
  &   Eff-ML &0.00 & 0.04 & 0.04 & 0.95& & 0.00 & 0.02 & 0.02 & 0.97 \\ 
   \hline
\end{tabular}
}
{
\raggedright \small
\setlength{\baselineskip}{1pt}
 Unadjusted: the unadjusted estimator. GEE-g: GEE with weighted g-computation. LMM-g: linear mixed models with weighted g-computation. Eff-PM: our proposed method with parametric working models. Eff-ML: our proposed method with machine learning algorithms. ESE: empirical standard error. ASE: average of estimated standard error. CP: coverage probability based on $t$-distribution.

}
\end{table}

\subsection{Simulations based on the WFHS data}
To further illustrate our methods under the setting of cluster-dependent sampling, we performed a simulation study based on the WFHS data. In each simulated data, \\
$\{(\widetilde{Y}_i(0), \widetilde{N}_i, \widetilde{C}_i, \widetilde{\mathrm{X}}_i), i =1,\dots, 56\}$ were sampled from the empirical distribution of observed data $\{(Y_i, N_i, C_i, \widetilde{\mathrm{X}}_i): i=1,\dots, 56\}$. Then, we randomly generated the treatment variable $\widetilde{A}_i \sim \mathcal{B}(0.5)$ and $\widetilde{Y}_{ij}(1) = \widetilde{Y}_{ij}(0) + N_i/3 +\mathcal{N}(0,1)$. Next, we set $\widetilde{M}_i =  \widetilde{A}_i \widetilde{C}_i( \widetilde{N}_i-2) + 2$, where $ \widetilde{C}_i$ is the cluster-level covariate of group job functions (core for $ \widetilde{C}_i =1$ and supporting for $ \widetilde{C}_i = 0$). Finally, for each cluster, we randomly sampled $ \widetilde{M}_i$ individuals without replacement, for whom $\widetilde{S}_{ij}=1$, and the observed data are $\{\widetilde{Y}_{ij}, \widetilde{A}_i, \widetilde{M}_i, \widetilde{C}_{i1},  \widetilde{\mathrm{X}}_i: \widetilde{S}_{ij} = 1,i=1,\dots,m,j=1,\dots, \widetilde{N}_i\}$ with $\widetilde{Y}_{ij} = \widetilde{A}_i\widetilde{Y}_{ij}(1) + (1-\widetilde{A}_i) \widetilde{Y}_{ij}(0)$.
We repeated the above steps 1,000 times and obtained 1,000 data sets, based on which we estimated both estimands as in Section~5 of the main paper. 

Table \ref{tab: resampled-data-analysis} gives the analysis results under the cluster-dependent sampling scheme. This analysis mimics the real-world setting, i.e., the distribution of outcome and covariates under control is based on real data and is unspecified. Consistent with our simulation results in Section 5 of the main paper, the model-based methods are biased for both estimands under cluster-dependent sampling, while our proposed methods have negligible bias and nominal coverage. Furthermore, our proposed methods implement flexible covariate adjustment via the working nuisance models and clearly demonstrate higher precision than the unadjusted estimator.

\begin{table}[htbp]
\renewcommand{\arraystretch}{0.8}
\caption{\small Results from the cluster-dependent sampling analyses based on $1,000$ simulation draws using the WFHS data set.
}\label{tab: resampled-data-analysis}
\centering
\resizebox{1\textwidth}{!}{
\begin{tabular}{lrrrrrrrrr}
  \hline
   & \multicolumn{4}{c}{\shortstack[c]{Cluster-average treatment effect\\ ${\Delta}_C = 3.92$}
  } & & \multicolumn{4}{c}{\shortstack[c]{Individual-average treatment effect\\ ${\Delta}_I = 5.37$}
  }\\
 Method & Bias & ESE   &  ASE &  CP &\ & Bias & ESE   &  ASE &  CP\\ 
  \hline
  Unadjusted &  0.02 & 0.50 & 0.48 & 0.93 & & 0.04 & 0.98 & 0.83 & 0.80 \\ 
    GEE-g  &  -0.71 & 0.40 & 0.40 & 0.56 & & -1.28 & 0.74 & 0.55 & 0.39 \\
    LMM-g  &-0.66 & 0.41 & 0.34 & 0.54 & & -1.15 & 0.83 & 0.41 & 0.33 \\ 
     Eff-PM &0.00 & 0.34 & 0.35 & 0.95& & 0.05 & 0.65 & 0.65 & 0.91 \\ 
   Eff-ML &  0.00 & 0.38 & 0.36 & 0.94 & & 0.05 & 0.66 & 0.65 & 0.90 \\ 
   \hline
\end{tabular}}
\vspace{5pt}

{
\raggedright \small
\setlength{\baselineskip}{1pt}
 Unadjusted: the unadjusted estimator. GEE-g: GEE with weighted g-computation. LMM-g: linear mixed models with weighted g-computation. Eff-PM: our proposed method with parametric working models. Eff-ML: our proposed method with machine learning algorithms. ESE: empirical standard error. ASE: average of estimated standard error. CP: coverage probability based on $t$-distribution.

}
\end{table}

\subsection{Simulations with other methods}
For our two simulation experiments described in Section 5, we additional implemented Aug-GEE-g and TMLE as two comparison methods, and describe our findings below. For Aug-GEE-g, we used the R package \verb"CRTgeeDR" \citep{prague2016crtgeedr} with weighted g-computation (defined in Section 3 of the main manuscript). For TMLE, we implemented the cluster-level TMLE with cluster-level data for $\Delta_C$ and hierarchical TMLE with individual-level data for $\Delta_I$, following the description in \cite{benitez2021comparative} to compute point estimates and variances. The nuisance functions in TMLE were estimated by an ensemble method of generalized linear models, regression trees, and neural networks using the \verb"tmle" R package \citep{gruber2012tmle}. Both methods adjusted for the same set of covariates as our proposed estimators.

The full simulation results are summarized in Tables~\ref{tab: sim1} and \ref{tab: sim2} below for completeness. Across experiments and scenarios, Aug-GEE-g has comparable performance to GEE-g in terms of bias and precision. TMLE for cluster-average treatment effect is valid across scenarios as expected, since it only uses cluster-level information. However, TMLE for individual-average treatment effects has large bias because it targets the average treatment effect among the enrolled population, rather than among the entire source population of interest. In terms of variance, TMLE can improve precision by covariate adjustment, while this precision gain lacks consistency across simulation scenarios and is generally smaller compared to our proposed methods with machine learning algorithms.
\renewcommand{\arraystretch}{0.8}
\begin{table}[htbp]
\centering
\caption{Results in the first simulation experiment with continuous outcomes.}\label{tab: sim1}
\resizebox{1\textwidth}{!}{
\begin{tabular}{lrrrrrrrrrr}
  \hline
  & & \multicolumn{4}{c}{\shortstack[c]{ Cluster-average treatment \\ effect $\Delta_C = 6$}  } & & \multicolumn{4}{c}{\shortstack[c]{ Individual-average treatment \\ effect $\Delta_I = 8.67$}}\\
Setting  & Method & Bias & ESE   &  ASE &  CP &\ & Bias & ESE   &  ASE &  CP\\ 
  \hline
\multirow{7}{*}{\shortstack[l]{Scenario 1: \\ Small $m$ with\\ 
 random \\
 observed cluster sizes}}&  Unadjusted &   $-$0.08 & 4.85 & 4.88 & 0.95& & $-$0.19 & 4.51 & 4.25 & 0.93 \\ 
&    GEE-g & 0.13 & 2.77 & 2.45 & 0.92& & $-$0.13 & 3.64 & 3.10 & 0.89 \\ 
&    LMM-g & 0.12 & 2.77 & 2.78 & 0.96& & $-$0.13 & 3.63 & 2.94 & 0.90 \\ 
&    Aug-GEE-g & $-$0.03 & 2.77 & 2.31 & 0.90 & & $-$0.11 & 3.76 & 2.66 & 0.84 \\
 &    TMLE &  $-$0.01 & 3.13 & 2.87 & 0.93 & & $-$2.27 & 1.96 & 1.55 & 0.66 \\ 
 & Eff-PM & $-$0.03 & 3.23 & 2.66 & 0.92 & & $-$0.10 & 3.89& 3.38 & 0.91 \\ 
  &   Eff-ML &0.03 & 2.24 & 2.05 & 0.95 & & 0.03 & 2.59 & 2.49 & 0.95 \\ 
   \hline
\multirow{7}{*}{\shortstack[l]{Scenario 2: \\ Small $m$ with\\ 
 cluster-dependent \\
 observed cluster sizes}}&  Unadjusted & 0.17 & 5.33 & 5.31 & 0.95 & &1.63 & 4.66 & 4.19 & 0.90 \\  
&    GEE-g  & 1.98 & 3.75 & 3.25 & 0.86 & & 0.87   & 4.53 & 3.66 & 0.88 \\ 
&    LMM-g  &  1.74 & 3.63 & 3.69 & 0.94 & & 0.79 & 4.26 & 3.83 & 0.93 \\ 
&    Aug-GEE-g  & 1.14 & 3.70 & 2.82 & 0.85 & & 1.40 & 4.57 & 3.52 & 0.86 \\
 &    TMLE &  0.10 & 3.84 & 3.61 & 0.94 & & 0.99 & 3.00 & 2.21 & 0.84 \\  
& Eff-PM & 0.08 & 3.87 & 3.47 & 0.93 & & 0.09 & 4.32 & 3.97 & 0.93 \\ 
  &   Eff-ML & $-$0.21 & 3.42 & 3.47 & 0.95 & & $-$0.06 & 4.20 & 3.84 & 0.93 \\
   \hline
\multirow{7}{*}{\shortstack[l]{Scenario 3: \\ Large $m$ with\\ 
 random \\
 observed cluster sizes}}&  Unadjusted &   0.01 & 2.62 & 2.65 & 0.96 & &$-$0.03 & 2.35 & 2.36 & 0.95 \\ 
&    GEE-g  &0.04 & 1.40 & 1.38 & 0.95 & &$-$0.01 & 1.89 & 1.83 & 0.94 \\ 
&    LMM-g  & 0.04 & 1.40 & 1.42 & 0.96 & &$-$0.01 & 1.89 & 1.55 & 0.90 \\ 
&    Aug-GEE-g  &$-$0.00 & 1.39 & 1.35 & 0.95 & &$-$0.01 & 1.91 & 1.58 & 0.90 \\ 
 &    TMLE &  0.07 & 1.81 & 1.52 & 0.92 & &$-$2.13 & 0.91 & 0.67 & 0.19 \\ 
 & Eff-PM & 0.00 & 1.40 & 1.38 & 0.95 & &$-$0.01 & 1.93 & 1.92 & 0.95 \\ 
  &   Eff-ML & 0.04 & 0.70 & 0.71 & 0.95 & & 0.00 & 0.77 & 0.81 & 0.97 \\ 
   \hline
\multirow{7}{*}{\shortstack[l]{Scenario 4: \\ Large $m$ with\\ 
 cluster-dependent \\
 observed cluster sizes}}&  Unadjusted &  $-$0.04 & 2.94 & 2.89 & 0.95 & & 1.49 & 2.45 & 2.38 & 0.88 \\ 
&    GEE-g  &  1.80 & 1.89 & 1.82 & 0.82 & & 0.74 & 2.21 & 2.12 & 0.91 \\ 
&    LMM-g  &   1.72 & 1.87 & 1.89 & 0.85 & & 0.72 & 2.20 & 2.00 & 0.91 \\ 
&    Aug-GEE-g  & 1.02 & 1.86 & 1.63 & 0.86 & & 0.92 & 2.28 & 2.26 & 0.90 \\ 
 &    TMLE & 0.08 & 2.12 & 1.92 & 0.92 & & 0.08 & 1.20 & 0.82 & 0.82 \\ 
 & Eff-PM & 0.03 & 1.89 & 1.83 & 0.94 & & 0.01 & 2.21 & 2.16 & 0.94 \\
  &   Eff-ML & 0.01 & 1.91 & 1.83 & 0.94 & & 0.00 & 2.23 & 2.13 & 0.94 \\ 
   \hline
\end{tabular}
}
{
\raggedright \small
\setlength{\baselineskip}{1pt}
 Unadjusted: the unadjusted estimator. GEE-g: GEE with weighted g-computation. LMM-g: linear mixed models with weighted g-computation. Eff-PM: our proposed method with parametric working models. Eff-ML: our proposed method with machine learning algorithms. ESE: empirical standard error. ASE: average of estimated standard error. CP: coverage probability based on $t$-distribution.

}
\end{table}

\begin{table}[htbp]
\centering
\caption{Results in the second simulation experiment with binary outcomes.}\label{tab: sim2}
\resizebox{1\textwidth}{!}{
\begin{tabular}{lrrrrrrrrrr}
  \hline
  & & \multicolumn{4}{c}{\shortstack[c]{ Cluster-average treatment \\ effect $\Delta_C = 1.54$}
  } & & \multicolumn{4}{c}{\shortstack[c]{ Participant-average treatment \\effect $\Delta_I = 1.18$}
  }\\
Setting  & Method & Bias & ESE   &  ASE &  CP &\ & Bias & ESE   &  ASE &  CP\\ 
  \hline
\multirow{7}{*}{\shortstack[l]{Scenario 1: \\ Small $m$ with\\ 
 random \\
 observed cluster sizes}}&  Unadjusted &   0.04 & 0.25 & 0.24 & 0.94& & 0.03 & 0.16 & 0.13 & 0.94 \\ 
&    GEE-g & 0.01 & 0.23 & 0.19 & 0.92 & & 0.03 & 0.22 & 0.13 & 0.93 \\ 
&    LMM-g & 0.01 & 0.22 & 0.19 & 0.92& & 0.03 & 0.16 & 0.13 & 0.93 \\ 
&    Aug-GEE-g & 0.02 & 0.20 & 0.18 & 0.93& & 0.03 & 0.16 & 0.11 & 0.89 \\ 
 &    TMLE & $-$0.02 & 0.30 & 0.24 & 0.89 & & 0.35 & 0.20 & 0.15 & 0.41 \\
 & Eff-PM & 0.01 & 0.21 & 0.20 & 0.93& & 0.02 & 0.16 & 0.15 & 0.95 \\ 
  &   Eff-ML &0.03 & 0.22 & 0.20 & 0.94& & 0.03 & 0.15 & 0.14 & 0.95 \\ 
   \hline
\multirow{7}{*}{\shortstack[l]{Scenario 2: \\ Small $m$ with\\ 
 cluster-dependent \\
 observed cluster sizes}}&  Unadjusted & 0.05 & 0.30 & 0.29 & 0.94 & & $-$0.03 & 0.15 & 0.16 & 0.92 \\
&    GEE-g  & $-$0.18 & 0.30 & 0.20 & 0.69& & $-$0.05 & 0.16 & 0.17 & 0.89 \\
&    LMM-g  &  $-$0.12 & 0.24 & 0.20 & 0.78& & $-$0.04 & 0.15 & 0.16 & 0.90 \\ 
&    Aug-GEE-g  & $-$0.15 & 0.23 & 0.16 & 0.66& & $-$0.02 & 0.24 & 0.13 & 0.81 \\ 
 &    TMLE &  0.01 & 0.33 & 0.30 & 0.92 & &0.12 & 0.17 & 0.11 & 0.80 \\ 
& Eff-PM & 0.04 & 0.28 & 0.26 & 0.92 & &0.03 & 0.17 & 0.17 & 0.95 \\ 
  &   Eff-ML &0.06 & 0.30 & 0.30 & 0.93 & & 0.05 & 0.19 & 0.17 & 0.95 \\ 
   \hline
\multirow{7}{*}{\shortstack[l]{Scenario 3: \\ Large $m$ with\\ 
 random \\
 observed cluster sizes}}&  Unadjusted &   0.01 & 0.13 & 0.13 & 0.95& & 0.01 & 0.07 & 0.07 & 0.95 \\ 
&    GEE-g  &0.00 & 0.11 & 0.10 & 0.94& & 0.00 & 0.07 & 0.07 & 0.95 \\
&    LMM-g  & 0.00 & 0.11 & 0.10 & 0.94& & 0.01 & 0.07 & 0.07 & 0.95 \\
&    Aug-GEE-g  &0.01 & 0.10 & 0.10 & 0.95& & 0.00 & 0.07 & 0.06 & 0.92 \\ 
 &    TMLE &   0.00 & 0.14 & 0.11 & 0.89& & 0.35 & 0.10 & 0.09 & 0.01 \\
 & Eff-PM &0.01 & 0.10 & 0.10 & 0.95& & 0.00 & 0.07 & 0.08 & 0.97 \\ 
  &   Eff-ML & 0.01 & 0.10 & 0.10 & 0.95& & 0.01 & 0.06 & 0.07 & 0.97 \\ 
   \hline
\multirow{7}{*}{\shortstack[l]{Scenario 4: \\ Large $m$ with\\ 
 cluster-dependent \\
 observed cluster sizes}}&  Unadjusted &  0.01 & 0.15 & 0.15 & 0.95& & $-$0.05 & 0.07 & 0.09 & 0.91 \\ 
&    GEE-g  &  $-$0.20 & 0.10 & 0.10 & 0.44& & $-$0.06 & 0.07 & 0.09 & 0.88 \\ 
&    LMM-g  &   $-$0.14 & 0.11 & 0.10 & 0.67& & $-$0.05 & 0.07 & 0.08 & 0.89 \\ 
&    Aug-GEE-g  & $-$0.19 & 0.10 & 0.08 & 0.37& & $-$0.07 & 0.07 & 0.06 & 0.74 \\ 
 &    TMLE & 0.02 & 0.19 & 0.14 & 0.90& & 0.10 & 0.08 & 0.06 & 0.69 \\ 
 & Eff-PM & 0.01 & 0.13 & 0.13 & 0.94& & 0.00 & 0.07 & 0.08 & 0.97 \\ 
  &   Eff-ML & 0.01 & 0.13 & 0.13 & 0.94& & 0.01 & 0.08 & 0.08 & 0.97 \\ 
   \hline
\end{tabular}
}
\end{table}

\clearpage 

\bibliographystyle{apalike}
\bibliography{references}